\documentclass{sig-alternate}

\usepackage{amsmath}
\usepackage{color}
\usepackage{float}
\usepackage{subfig}
\usepackage{array}
\usepackage{amssymb}
\usepackage{algorithm}
\usepackage{algorithmic}

\ifodd 1
\else

\fi

\begin{document}

\conferenceinfo{e-Energy'13,} {May 21--24, 2013, Berkeley, California, USA.}
\CopyrightYear{2013}
\crdata{978-1-4503-2052-8/13/05}
\clubpenalty=10000
\widowpenalty = 10000

%
%

\title{Dynamic Provisioning in Next-Generation Data Centers with On-site
Power Production}

\numberofauthors{2}
\author{
\alignauthor
Jinlong Tu, Lian Lu and Minghua Chen\\
       \affaddr{Department of Information Engineering}\\
       \affaddr{The Chinese University of Hong Kong}\\
\alignauthor
Ramesh K. Sitaraman\\
       \affaddr{Department of Computer Science}\\
       \affaddr{University of Massachusetts at Amherst}\\
       \affaddr{Akamai Technologies}\\
}
\newcounter{copyrightbox}
\maketitle

\newtheorem{prop}{Proposition}
\newtheorem{lem}{Lemma}
\newdef{defn}{Definition}
\newtheorem{thm}{Theorem}

\begin{abstract}
The critical need for clean and economical sources of energy is transforming
data centers that are primarily energy consumers to also energy
producers. We focus on minimizing the operating costs of next-generation
data centers that can jointly optimize the energy supply  from on-site
generators and the power grid, and the energy demand  from servers as well as power conditioning and cooling systems.
We formulate the cost minimization problem and present an offline
optimal algorithm. For ``on-grid'' data centers that use only the
grid, we devise a deterministic online algorithm that achieves the
best possible competitive ratio of $2-\alpha_{s}$, where $\alpha_{s}$
is a normalized look-ahead window size. The competitive ratio of an online algorithm is defined as the maximum ratio (over all possible inputs) between the algorithm's cost (with no or limited look-ahead) and the offline optimal assuming complete future information. We remark that the results hold as long as the overall energy demand (including server, cooling, and power conditioning) is a convex and increasing function in the total number of active servers and also in the total server load.
For ``hybrid'' data centers
that have on-site power generation in addition to the grid, we develop
an online algorithm that achieves a competitive ratio of at most \textmd{\normalsize {\small $\frac{P_{\max}\left(2-\alpha_{s}\right)}{c_{o}+c_{m}/L}\left[1+2\frac{P_{\max}-c_{o}}{P_{\max}(1+\alpha_{g})}\right]$},
where }$\alpha_{s}$ and $\alpha_{g}$  are normalized look-ahead
window sizes, $P_{\max}$ is the maximum grid power price, and $L$,
$c_{o}$, and $c_{m}$ are parameters of an on-site generator.

Using extensive workload traces from Akamai
with the corresponding grid power prices, we simulate our offline
and online algorithms in a realistic setting. Our offline
(resp., online) algorithm achieves a cost reduction of 25.8\% (resp.,
20.7\%) for a hybrid data center and 12.3\% (resp., 7.3\%) for an
on-grid data center. The cost reductions are quite significant and
make a strong case for a joint optimization of energy supply and energy
demand in a data center. A hybrid data center provides about 13\%
additional cost reduction over an on-grid data center representing
the additional cost benefits that on-site power generation provides
over using the grid alone.

\end{abstract}

\category{F.1.2}{Modes of Computation}Online computation
\category{G.1.6}{Optimization}Nonlinear programming
\category{I.1.2}{Algorithms}Analysis of algorithms
\category{I.2.8}{Problem Solving, Control Methods, and Search}Scheduling

\terms{Algorithms, Performance}


\section{Introduction}\label{sec:intro}

Internet-scale cloud services that deploy large distributed systems
of servers around the world are revolutionizing all aspects of human
activity. The rapid growth of such services has lead to a significant
increase in server deployments in data centers around the world. Energy
consumption of data centers account for roughly 1.5\% of the global
energy consumption and is increasing at an alarming rate of about
15\% on an annual basis \cite{Koomey2010}. The  surging global energy
demand  relative to its supply has caused the price of electricity
to rise, even while other operating expenses of a data center such
as network bandwidth have decreased precipitously. Consequently, the
energy costs now represent a large fraction of the operating expenses
of a data center today \cite{barroso2007case},  and decreasing the
energy expenses has become a central concern for data center operators.


The emergence of energy as a central consideration for enterprises
that operate large server farms is drastically altering the traditional
boundary between a data center and a power utility (c.f. Figure \ref{fig:data.center}).
Traditionally, a data center hosts servers but buys electricity from
an utility company through the power grid. However, the criticality
of the energy supply is leading data centers to broaden their role
to also generate much of the required power on-site, decreasing their
dependence on a third-party utility. While data centers have always
had generators as a short-term backup for when the grid fails, on-site
generators for sustained power supply is a newer trend. For instance,
Apple recently announced that it will build a massive data center
for its iCloud services with 60\% of its energy coming from its on-site
generators that use ``clean energy'' sources such as fuel cells
with biogas and solar panels \cite{apple2012}. As another example,
eBay recently announced that it will add a 6 MW facility to its existing
data center in Utah that will be largely powered by on-site fuel cell
generators \cite{ebay2012}. The trend for \emph{hybrid} data centers
that generate electricity on-site (c.f. Figure \ref{fig:data.center})
with reduced reliance on the grid is driven by the confluence of several
factors. This trend is also mirrored in the broader power industry
where the centralized model for power generation with few large power
plants is giving way to a more distributed generation model \cite{borbely2001distributed}
where many smaller on-site generators produce power that is consumed
locally over a ``micro-grid''.

A key factor favoring on-site generation is the potential for cheaper
power than the grid, especially during peak hours. On-site generation
also reduces transmission losses that in turn reduce the effective
cost, because the power is generated close to where it is consumed.
In addition, another factor favoring on-site generation is a requirement
for many enterprises to use cleaner renewable energy sources, such
as Apple's mandate to use 100\% clean energy in its data centers \cite{applemandate2012}.
Such a mandate is more easily achievable with the enterprise generating
all or most of its power on-site, especially since recent advances
such as the fuel cell technology of Bloom Energy \cite{bloomenergy2012}
make on-site generation economical and feasible. Finally, the risk
of service outages caused by the failure of the grid, as happened
recently when thunderstorms brought down the grid causing a denial-of-service
for Amazon's AWS service for several hours \cite{amazonoutage2012},
has provided greater impetus for on-site power generation that can
sustain the data center for extended periods without the grid.

Our work focuses on the key challenges that arise in the emerging
hybrid model for a data center that is able to simultaneously optimize
\emph{both} the generation and consumption of energy (c.f. Figure
\ref{fig:data.center} ). In the traditional scenario, the utility
is responsible for energy provisioning (\textbf{EP}) that has the
goal of supplying energy as economically as possible to meet the energy
demand, albeit the utility has no detailed knowledge and no control
over the server workloads within a data center that drive the consumption
of power. Optimal energy provisioning by the utility in isolation
is characterized by the unit commitment problem \cite{1295033,shiina2004stochastic}
that has been studied over the past decades. The energy provisioning
problem takes as input the demand for electricity from the consumers
and determines which power generators should be used at what time
to satisfy the demand in the most economical fashion. Further, in
a traditional scenario, a data center is responsible for capacity
provisioning (\textbf{CP}) that has the goal of managing its server
capacity to serve the incoming workload from end users while reducing
the total energy demand of servers, as well as power conditioning and
various cooling systems, but without detailed knowledge or control
over the power generation. For instance, dynamic provisioning
of server capacity by turning off some servers during periods of low
workload to reduce the energy demand has been studied
in recent years \cite{lin2011dynamic,MathewSS12,beloglazov2011taxonomy,labpaper}.

The convergence of power generation and consumption within a single
data center entity and the increasing impact of energy costs requires
a new integrated approach to both energy provisioning (\textbf{EP})
and capacity provisioning (\textbf{CP}). A key contribution of our
work is formulating and developing algorithms that simultaneously
manage on-site power generation, grid power consumption, and server
capacity with the goal of minimizing the operating cost of the data
center.

\paragraph*{Online vs. Offline Algorithms}
In designing algorithms for optimizing the operating cost of a hybrid
\begin{figure}[t!]
\centering{}\includegraphics[width=0.95\columnwidth]{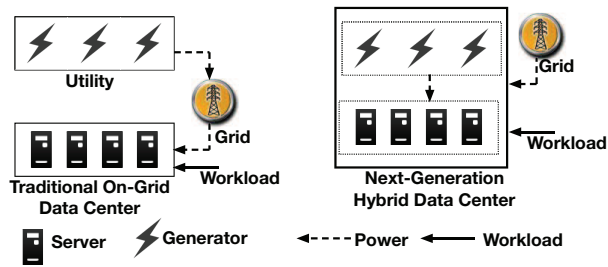}\caption{\label{fig:data.center}While an ``on-grid'' data center derives
all its power from the grid, next-generation ``hybrid'' data centers
have additional on-site power generation. }
\end{figure}
data center, there are three time-varying inputs: the server workload
$a(t)$ generated by service requests from users and the price of
a unit energy from the grid $p(t)$, and the total power consumption
function $g_{t}$ for each time $t$ where $1\leq t\leq T$. We begin
by investigating \emph{offline} algorithms that minimize the operating
cost with perfect knowledge of the entire input sequence $a(t)$,
$p(t)$ and $g_{t}$, for $1\leq t\leq T$. However, in real-life,
the time-varying input sequences are not knowable in advance. In particular,
the optimization must be performed in an \emph{online} fashion where
decisions at time $t$ are made with the knowledge of inputs $a(\tau)$,$p(\tau)$
and $g_{\tau}$, for $1\leq\tau\leq t+w$, where $w\geq0$ is a small
(possibly zero) look-ahead window. Specifically, an online algorithm
has no knowledge of inputs beyond the look-ahead window, \emph{i.e.},
for time $t+w<\tau\leq T$. We assume the inputs within the look-ahead are perfectly known when analyzing the algorithm performance. In practice, short-term demand or grid price can be estimated rather accurately by various techniques including pattern analysis and time series analysis and prediction \cite{gmach2007workload,elecpriceforecast}.
As is typical in the study of online algorithms
\cite{borodin1998online}, we seek theoretical guarantees for our
online algorithms by computing the \emph{competitive ratio} that is
ratio of the cost achieved by the online algorithm for an input to
the optimal cost achieved for the same input by an offline algorithm.
The competitive ratio is computed under a worst case scenario where
an adversary picks the worst possible inputs for the online algorithm.
Thus, a small competitive ratio provides a strong guarantee that the
online algorithm will achieve a cost close to the offline optimal
even for the worst case input.

\paragraph*{Our Contributions}
A key contribution of our work is to formulate and study data center
cost minimization (\textbf{DCM}) that integrates energy procurement
from the grid, energy production using on-site generators, and dynamic
server capacity management. Our work jointly optimizes the two components
of \textbf{DCM}: energy provisioning (\textbf{EP}) from the grid and
generators and capacity provisioning (\textbf{CP}) of the servers.
\begin{itemize}
\item We theoretically evaluate the benefit of joint optimization by showing
that optimizing energy provisioning (\textbf{EP}) and capacity provisioning
(\textbf{CP}) separately results in a factor loss of optimality $\rho=LP_{\max}/\left(Lc_{o}+c_{m}\right)$
 compared to optimizing them jointly, where $P_{\max}$
is the maximum grid power price, and $L,$ $c_{o,}$ and $c_{m}$
are the capacity, incremental cost, and base cost of an on-site generator
respectively. Further, we derive an efficient offline optimal algorithm
for hybrid data centers that jointly optimize \textbf{EP} and \textbf{CP}
to minimize the data center's operating cost.

\item For on-grid data centers, we devise an online deterministic algorithm that achieves
a competitive ratio of $2-\alpha_{s}$, where $\alpha_{s}\in[0,1]$
is the normalized look-ahead window size. Further, we show that our
algorithm has the best competitive ratio of any deterministic online
algorithm for the problem (c.f. Table \ref{tab:Summary-of-algorithmic}).
For the more complex hybrid data centers, we devise an online deterministic algorithm
that achieves a competitive ratio of $\frac{P_{\max}\left(2-\alpha_{s}\right)}{c_{o}+c_{m}/L}\left[1+2\frac{P_{\max}-c_{o}}{P_{\max}(1+\alpha_{g})}\right]$,
where $\alpha_{s}$ and $\alpha_{g}$ are normalized look-ahead window
sizes. Both online algorithms perform better as the look-ahead window
increases, as they are better able to plan their current actions based
on knowledge of future inputs. Interestingly, in the on-grid case,
we show that there exists \emph{fixed} threshold value for the look-ahead
window for which the online algorithm matches the offline optimal
in performance achieving a competitive ratio of 1, \emph{i.e.}, there is
no additional benefit gained by the online algorithm if its look-ahead
is increased beyond the threshold.

\begin{table}[t!]
{\small }%
\begin{tabular}{|c|c|c|}
\hline
\multicolumn{1}{|c|}{\textbf{\footnotesize Competitive }} & \textbf{\footnotesize On-grid} & \textbf{\footnotesize Hybrid}\tabularnewline
\textbf{\footnotesize{} Ratio} &  & \tabularnewline
\hline
{\scriptsize No Look-ahead } & {\small 2} & {\small $\frac{2P_{\max}}{c_{o}+c_{m}/L}\left[1+2\frac{P_{\max}-c_{o}}{P_{\max}}\right]$}\tabularnewline
\hline
{\scriptsize With Look-ahead} & {\small $2-\alpha_{s}$} & {\small $\frac{P_{\max}\left(2-\alpha_{s}\right)}{c_{o}+c_{m}/L}\left[1+2\frac{P_{\max}-c_{o}}{P_{\max}(1+\alpha_{g})}\right]$}\tabularnewline
\hline
\end{tabular}{\small \par}

{\footnotesize \caption{\label{tab:Summary-of-algorithmic} {\small Summary of algorithmic
results. The on-grid results are the best possible for any deterministic
online algorithm.} }
}
\end{table}

\item Using extensive workload traces from Akamai and the
corresponding grid prices, we simulate our offline and online algorithms
in a realistic setting with the goal of empirically evaluating their
performance. Our offline optimal (resp., online) algorithm achieves
a cost reduction of 25.8\% (resp., 20.7\%) for a hybrid data center
and 12.3\% (resp., 7.3\%) for an on-grid data center. The cost reduction
is computed in comparison with the baseline cost achieved by the current
practice of statically provisioning the servers and using only the power
grid. The cost reductions are quite significant and make a strong
case for utilizing our joint cost optimization framework. Furthermore,
our online algorithms obtain almost the same cost reduction as the
offline optimal solution even with a small look-ahead of 6 hours, indicating
the value of short-term prediction of inputs.

\item A hybrid data center provides about 13\% additional cost reduction
over an on-grid data center representing the additional cost benefits
that on-site power generation provides over using the grid alone.
Interestingly, it is sufficient to deploy a partial on-site generation
capacity that provides 60\% of the peak power requirements of the
data center to obtain over 95\% of the additional cost reduction.
This provides strong motivation for a traditional on-grid data center
to deploy at least a partial on-site generation capability to save
costs.

\end{itemize}

\section{The Data Center Cost Minimization Problem}\label{sec:problem.formulation}

We consider the scenario where a data center can jointly optimize
energy production, procurement, and consumption so as to minimize
its operating expenses. We refer to this data center cost minimization
problem as \textbf{DCM}. To study \textbf{DCM}, we model
how energy is produced using on-site power generators, how it can
be procured from the power grid, and how data center capacity can
be provisioned dynamically in response to workload. While some of
these aspects have been studied independently, our work is unique
in optimizing these dimensions simultaneously as next-generation data
centers can. Our algorithms minimize cost by use of techniques such
as: (i) dynamic capacity provisioning of servers -- turning off unnecessary
servers when workload is low to reduce the energy consumption
(ii) opportunistic energy procurement -- opting between the on-site
and grid energy sources to exploit price fluctuation, and (iii) dynamic
provisioning of generators - orchestrating which generators produce
what portion of the energy demand. While prior literature has considered
these techniques in isolation, we show how they can be used in coordination
to manage both the supply and demand of power to achieve substantial
cost reduction.
\begin{table}[tbh]
{\footnotesize }%
\begin{tabular*}{0.95\columnwidth}{@{\extracolsep{\fill}}l|>{\raggedright}p{0.75\columnwidth}}
\hline
\textbf{\footnotesize Notation}{\footnotesize{} } & \textbf{\footnotesize Definition}\tabularnewline
\hline
{\footnotesize $T$ } & {\footnotesize Number of time slots }\tabularnewline
{\footnotesize $N$} & {\footnotesize Number of on-site generators}\tabularnewline
{\footnotesize $\beta_{s}$ } & {\footnotesize Switching cost of a server (\$)}\tabularnewline
{\footnotesize $\beta_{g}$ } & {\footnotesize Startup cost of an on-site generator (\$)}\tabularnewline
{\footnotesize $c_{m}$} & {\footnotesize Sunk cost of maintaining a generator in its active state per slot (\$)}\tabularnewline
{\footnotesize $c_{o}$} & {\footnotesize Incremental cost for an active generator to output
an additional unit of energy (\$/Wh)}\tabularnewline
{\footnotesize $L$} & {\footnotesize The maximum output of a generator (Watt)}\tabularnewline
{\footnotesize $a(t)$} & {\footnotesize Workload at time $t$ }\tabularnewline
{\footnotesize $p(t)$} & {\footnotesize Price per unit energy drawn from the grid at $t$ ($P_{\min}\leq p(t)\leq P_{\max}$)
(\$/Wh)}\tabularnewline
{\footnotesize $x(t)$} & {\footnotesize Number of active servers at $t$}\tabularnewline
{\footnotesize $s(t)$} & {\footnotesize Total server service capability at $t$ }\tabularnewline
{\footnotesize $v(t)$} & {\footnotesize Grid power used at $t$ (Watt)}\tabularnewline
{\footnotesize $y(t)$ } & {\footnotesize Number of active on-site generators at $t$}\tabularnewline
{\footnotesize $u(t)$} & {\footnotesize Total power output from active generators at $t$ (Watt)}\tabularnewline
{\footnotesize $g_{t}(x(t),a(t))$} & {\footnotesize Total power consumption as a function of $x(t)$ and
$a(t)$ at $t$ (Watt)}\tabularnewline
\hline
\end{tabular*}{\footnotesize \par}

{\footnotesize Note: we use bold symbols to denote vectors, \emph{e.g.},
$\boldsymbol{x}=\langle x(t)\rangle$. Brackets indicate the unit.}
\caption{\label{tab:notations} Key notation.}
\end{table}
\subsection{Model Assumptions}

We adopt a discrete-time model whose time slot matches the timescale
at which the scheduling decisions can be updated. Without loss of
generality, we assume there are totally $T$ slots, and each has a
unit length.

\textbf{Workload model}. Similar to existing work \cite{chase2001managing,pinheiro2001load,doyle2003model},
we consider a ``mice'' type of workload for the data center where
each job has a small transaction size and short duration. Jobs arriving in a slot get
served in the same slot. Workload can be split among active servers
at arbitrary granularity like a fluid. These assumptions model a ``request-response''
type of workload that characterizes serving web content or hosted
application services that entail short but real-time interactions
between the user and the server. The workload to be served at time
$t$ is represented by $a(t)$. Note that we do not rely on any specific
stochastic model of $a(t)$.

\textbf{Server model}. We assume that the data center consists of
a sufficient number of homogeneous servers, and each has unit service capacity,
\emph{i.e.}, it can serve at most one unit workload per slot, and the same
power consumption model. Let $x(t)$ be the number of active servers
and $s(t)\in[0,x(t)]$ be the total server service capability at
time $t$. It is clear that $s(t)$ should be larger than $a(t)$ to get the workload served in the same slot.  We model the aggregate server power consumption as $b(t)\triangleq f_{s}\left(x(t),s(t)\right)$,
an increasing and convex function of $x(t)$ and $s(t)$. That is, the first and second order partial derivatives in $x(t)$ and $s(t)$ are all non-negative. Since $f_{s}\left(x(t),s(t)\right)$ is increasing in $s(t)$, it is optimal to  always set $s(t)=a(t)$. Thus, we have $b(t)=f_{s}\left(x(t),a(t)\right)$ and $x(t)\geq a(t)$.

This power consumption model is quite general and captures many common
server models. One example is the commonly adopted standard linear
model \cite{barroso2007case}:
\[
f_{s}\left(x(t),a(t)\right)=c_{idle}x(t)+(c_{peak}-c_{idle})a(t),
\]
where $c_{idle}$ and $c_{peak}$ are the power consumed by an server
at idle and fully utilized state, respectively. Most servers today
consume significant amounts of power even when idle. A holy grail for
server design is to make them {}``power proportional''
by making $c_{idle}$ zero \cite{palasamudram2012using}.

Besides, turning a server on entails switching cost \cite{MathewSS12},
denoted as $\beta_{s}$, including the amortized service interruption
cost, wear-and-tear cost, \emph{e.g.}, component procurement, replacement
cost (hard-disks in particular) and risk associated with server switching.
It is comparable to the energy cost of running a server for several
hours \cite{lin2011dynamic}.

In addition to servers, power conditioning and cooling systems
also consume a significant portion of power. The three%
\footnote{The other two, networking and lighting, consume little
power and have less to do with server utilization. Thus, we do not model
the two in this paper.%
} contribute about 94\% of overall power consumption
and their power draw vary drastically with server utilization
\cite{pelley2009understanding}. Thus, it is important to model
the power consumed by power conditioning and cooling systems.

\textbf{Power conditioning system model}. Power conditioning system
usually includes power distribution units (PDUs) and uninterruptible
power supplies (UPSs). PDUs transform the high voltage power distributed
throughout the data center to voltage levels appropriate for servers.
UPSs provides temporary power during outage. We model the power consumption
of this system as $f_{p}(b(t))$, an increasing and convex function
of the aggregate server power consumption $b(t)$.

This model is general and one example is a quadratic function adopted in a comprehensive study
on the data center power consumption \cite{pelley2009understanding}: $f_{p}(b(t))=C_{1}+\pi_{1}b^{2}(t)$,
where $C_{1}>0$ and $\pi_{1}>0$ are constants depending on specific
PDUs and UPSs.

\textbf{Cooling system model}. We model the power consumed by the
cooling system as $f_{c}^{t}(b(t))$, a time-dependent (\emph{e.g.}, depends
on ambient weather conditions) increasing and convex  function of
$b(t)$.

This cooling model captures many common cooling systems. According
to \cite{liu2012renewable}, the power consumption of an outside air
cooling system can be modelled as a time-dependent cubic function
of $b(t)$: $f_{c}^{t}(b(t))=K_{t}b^{3}(t),$ where $K_{t}>0$ depends
on ambient weather conditions, such as air temperature, at time $t$.
According to \cite{pelley2009understanding}, the power draw of a
water chiller cooling system can be modelled as a time-dependent quadratic
function of $b(t)$: $f_{c}^{t}(b(t))=Q_{t}b^{2}(t)+L_{t}b(t)+C_{t},$
where $Q_{t},L_{t},C_{t}\geq0$ depend on outside air and chilled
water temperature at time $t$. Note that all we need is $f_{c}^{t}(b(t))$
is increasing and convex in $b(t)$.

\textbf{On-site generator model}. We assume that the data center has
$N$ units of homogeneous on-site generators, each having an power
output capacity $L$. Similar to generator models studied in the unit
commitment problem \cite{kazarlis1996genetic}, we define a generator
startup cost $\beta_{g}$, which typically involves heating up
cost, additional maintenance cost due to each startup (\emph{e.g.}, fatigue
and possible permanent damage resulted by stresses during startups),
$c_{m}$ as the sunk cost of maintaining a generator in its active
state for a slot, and $c_{o}$ as the incremental cost for an active
generator to output an additional unit of energy. Thus, the total
cost for $y(t)$ active generators that output $u(t)$ units of energy
at time $t$ is $c_{m}y(t)+c_{o}u(t)$.

\textbf{Grid model}. The grid supplies energy to the data center in
an ``on-demand'' fashion, with time-varying price $p(t)$ per unit
energy at time $t$. Thus, the cost of drawing $v(t)$ units of energy
from the grid at time $t$ is $p(t)v(t)$. Without loss of generality,
we assume $0\leq P_{\min}\leq p(t)\leq P_{\max}$.

To keep the study interesting and practically relevant, we make the
following assumptions: (i) the server and generator turning-on cost
are strictly positive, \emph{i.e.}, $\beta_{s}>0$ and $\beta_{g}>0$. (ii)
$c_{o}+c_{m}/L<P_{\max}.$ This ensures that the minimum on-site energy
price is cheaper than the maximum grid energy price. Otherwise, it
should be clear that it is optimal to always buy energy from the
grid, because in that case the grid energy is cheaper and incurs
no startup costs.

\subsection{Problem Formulation}\label{sec:formulation.DCM}

Based on the above models, the data center total power consumption
is the sum of the server, power conditioning system and the cooling
system power draw, which can be expressed as a time-dependent function
of $b(t)$ ($b(t)=f_{s}(x(t),a(t))$ ):
\[
b(t)+f_{p}(b(t))+f_{c}^{t}(b(t))\triangleq g_{t}(x(t),a(t)).
\]
We remark that $g_{t}(x(t),a(t))$ is increasing and convex in $x(t)$
and $a(t)$. This is because it is the sum of three increasing and
convex functions. \emph{Note that all results we derive in this paper apply to any
$g_{t}(x,a)$ as long as it is increasing and convex in $x$ and $a$.}

Our objective is to minimize the data center total cost in entire
horizon $[1,T]$, which is given by
\begin{align}
 & \mbox{Cost}(x,y,u,v)\triangleq\sum_{t=1}^{T}\left\{ v(t)p(t)+c_{o}u(t)+c_{m}y(t)\right.\label{eq:MP-cost-func}\\
 & \qquad\left.+\beta_{s}[x(t)-x(t-1)]^{+}+\beta_{g}[y(t)-y(t-1)]^{+}\right\} ,\nonumber
\end{align}
which includes the cost of grid electricity, the running cost of
on-site generators, and the switching cost of servers and on-site
generators in the entire horizon $[1,T]$. Throughout this paper,
we set initial condition $x(0)=y(0)=0.$

We formally define the data center cost minimization problem as a
non-linear mixed-integer program, given the workload $a(t)$,
the grid price $p(t)$ and the time-dependent function $g_{t}(x,a),$
for $1\leq t\leq T$, as time-varying inputs.
\begin{eqnarray}
\min_{x,y,u,v} &  & \mbox{Cost}(x,y,u,v)\label{eq:MP-objective}\\
\mbox{s.t.} &  & u(t)+v(t)\geq g_{t}(x(t),a(t)),\label{eq:MP-power-balance-constraint}\\
 &  & u(t)\leq Ly(t),\label{eq:MP-energy-capacity-constraint}\\
 &  & x(t)\geq a(t),\label{eq:MP-server-workload-constraint}\\
 &  & y(t)\leq N,\label{eq:MP-generator-number-constraint}\\
 &  & x(0)=y(0)=0,\label{eq:MP-boundary-conditions}\\
\mbox{var} &  & x(t),y(t)\in\mathbb{N}^{0},u(t),v(t)\in\mathbb{R}_{0}^{+},\, t\in[1,T],\nonumber
\end{eqnarray}
where $\left[\cdot\right]^{+}=\max(0,\cdot)$, $\mathbb{N}^{0}$ and
$\mathbb{R}_{0}^{+}$ represent the set of non-negative integers and
real numbers, respectively.

Constraint \eqref{eq:MP-power-balance-constraint} ensures the total
power consumed by the data center is jointly supplied by the generators
and the grid. Constraint \eqref{eq:MP-energy-capacity-constraint}
captures the maximal output of the on-site generator. Constraint \eqref{eq:MP-server-workload-constraint}
specifies that there are enough active servers to serve
the workload. Constraint \eqref{eq:MP-generator-number-constraint}
is generator number constraint. Constraint \eqref{eq:MP-boundary-conditions}
is the boundary condition.

Note that this problem is challenging to solve. First, it is a non-linear
mixed-integer optimization problem. Further, the objective function
values across different slots are correlated via the switching costs
$\beta_{s}[x(t)-x(t-1)]^{+}$ and $\beta_{g}[y(t)-y(t-1)]^{+}$, and
thus cannot be decomposed. Finally, to obtain an online solution we
do not even know the inputs beyond current slot.

Next, we introduce a proposition to simplify the structure of the
problem. Note that if $\left(x(t)\right)_{t=1}^{T}$ and $\left(y(t)\right)_{t=1}^{T}$
are given, the problem in \eqref{eq:MP-objective}-\eqref{eq:MP-boundary-conditions}
reduces to a linear program and can be solved independently for each
slot. We then obtain the following.
\begin{prop}
Given any $x(t)$ and $y(t)$, the $u(t)$ and $v(t)$ that
minimize the cost in \eqref{eq:MP-objective} with any $g_{t}(x,a)$ that is increasing in x and a, are given
by: $\forall t\in[1,T]$,
\[
u(t)=\begin{cases}
0, & \mbox{if }p(t)\leq c_{o},\\
\min\left(Ly(t),g_{t}(x(t),a(t))\right), & \mbox{otherwise,}
\end{cases}
\]
and
\[
v(t)=g_{t}(x(t),a(t))-u(t).
\]
\end{prop}
Note that $u(t),v(t)$ can be computed using \emph{only} $x(t),y(t)$
at current time $t$, thus can be determined in an online fashion.

Intuitively, the above proposition says if the on-site energy price $c_{o}$
is higher than the grid price $p(t)$, we should buy energy from the
grid; otherwise, it is the best to buy the cheap on-site energy up
to its maximum supply $L\cdot y(t)$ and the rest (if any) from the
more expensive grid. With the above proposition, we can reduce the
non-linear mixed-integer program in \eqref{eq:MP-objective}-\eqref{eq:MP-boundary-conditions}
with variables $\boldsymbol{x}$, $\boldsymbol{y}$,
$\boldsymbol{u}$, and $\boldsymbol{v}$ to the following integer
program with only variables $\boldsymbol{x}$ and $\boldsymbol{y}$:
\begin{eqnarray}
&  & \mathbf{DCM}:\nonumber \\
& \min & \sum_{t=1}^{T}\left\{ \psi\left(y(t),p(t),d_{t}(x(t))\right)+\beta_{s}[x(t)-x(t-1)]^{+}\right.\nonumber \\
&  & \;\left.+\beta_{g}[y(t)-y(t-1)]^{+}\right\} \label{eq:DCM-objective}\\
& \mbox{s.t.} & x(t)\geq a(t),\nonumber \\
&  & \eqref{eq:MP-generator-number-constraint},\eqref{eq:MP-boundary-conditions},\nonumber \\
& \mbox{var} & x(t),y(t)\in\mathbb{N}^{0},\, t\in[1,T],\nonumber
\end{eqnarray}
where $d_{t}(x(t))\triangleq g_{t}(x(t),a(t))$, for the ease of presentation in later sections, is increasing and convex in $x(t)$ and $\psi\left(y(t),p(t),d_{t}(x(t))\right)$ replaces
the term $v(t)p(t)+c_{o}u(t)+c_{m}y(t)$ in the original cost function
in \eqref{eq:MP-objective} and is defined as
\begin{eqnarray}
 &  & \psi\left(y(t),p(t),d_{t}(x(t))\right)\label{eq:the original cost function}\\
 & \triangleq & \begin{cases}
c_{m}y(t)+p(t)d_{t}(x(t)), & \mbox{if }p(t)\leq c_{o},\\
c_{m}y(t)+c_{o}Ly(t)+ & \mbox{if }p(t)>c_{o}\mbox{ and }\\
p(t)\left(d_{t}(x(t))-Ly(t)\right), & d_{t}(x(t))>Ly(t),\\
c_{m}y(t)+c_{o}d_{t}(x(t)), & \mbox{else.}
\end{cases}\nonumber
\end{eqnarray}

As a result of the analysis above, it suffices to solve the above
formulation of \textbf{DCM} with only variables $\boldsymbol{x}$
and $\boldsymbol{y}$, in order to minimize the data center operating
cost.

\subsection{An Offline Optimal Algorithm}\label{ssec:optimal.offline.algo}
We present an offline optimal algorithm for solving problem \textbf{DCM}
using Dijkstra's shortest path algorithm \cite{dijkstra1959note}.
We construct a graph $G=(V,E),$ where each vertex denoted by the
tuple $\langle x,y,t\rangle$ represents a state of the data center
where there are $x$ active servers, and $y$ active generators at
time $t$. We draw a directed edge from each vertex $\langle x(t-1),y(t-1),t-1\rangle$
to each possible vertex $\langle x(t),y(t),t\rangle$ to represent
the fact that the data center can transit from the first state to
the second state. Further, we associate the cost of that transition
shown below as the weight of the edge:
\begin{eqnarray*}
 &  & \psi\left(y(t),p(t),d_{t}(x(t))\right)+\beta_{s}[x(t)-x(t-1)]^{+}\\
 &  & +\beta_{g}[y(t)-y(t-1)]^{+}.
\end{eqnarray*}
Next, we find the minimum weighted path from the initial state represented
by vertex $\langle0,0,0\rangle$ to the final state represented by
vertex $\langle0,0,T+1\rangle$ by running Dijkstra's algorithm on
graph $G$. Since the weights represent the transition costs, it is
clear that finding the minimum weighted path in $G$ is equivalent
to minimizing the total transitional costs. Thus, our offline algorithm
provides an optimal solution for problem \textbf{DCM}.

\begin{thm}
\label{thm:offline.algo.performance}The algorithm described
above finds an optimal solution to problem \textbf{DCM} in time
$O\left(M^{2}N^{2}T\log\left(MNT\right)\right)$,
where $T$ is the number of slots, $N$ the number of generators
and $M=\max_{1\leq t\leq T}\left\lceil a(t)\right\rceil $.\end{thm}

\begin{proof}
Since the numbers of active servers and generators are at most $M$
and $N$, respectively, and there are $T+2$ time slots, graph
$G$ has $O(MNT)$ vertices and $O(M^{2}N^{2}T)$ edges. Thus, the
run time of Dijkstra's algorithm on graph $G$ is $O\left(M^{2}N^{2}T\log\left(MNT\right)\right)$.\end{proof}

\textbf{Remark}:
In practice, the time-varying input sequences ($p(t)$, $a(t)$
and $g_{t}$) may not be available in advance and hence it may be
difficult to apply the above offline algorithm. However, an offline
optimal algorithm can serve as a benchmark, using which we can evaluate
the performance of online algorithms.

\section{The Benefit of Joint Optimization}\label{sec:decomposition}

Data center cost minimization (\textbf{DCM}) entails the joint optimization
of both server capacity that determines the energy demand and on-site power
generation that determines the energy supply. Now consider the situation where
the data center optimizes the energy demand and supply separately.

First, the data center dynamically provisions the server capacity
according to the grid power price $p(t)$.
More formally, it solves the \emph{capacity provisioning} problem
which we refer to as \textbf{CP} below.
\begin{eqnarray*}
\textbf{CP}: & \min & \sum_{t=1}^{T}\left\{ p(t)\cdot d_{t}(x(t))+\beta_{s}[x(t)-x(t-1)]^{+}\right\} \\
 & \mbox{s.t.} & x(t)\geq a(t),\\
 &  & x(0)=0,\\
 & \mbox{var} & x(t)\in\mathbb{N}^{0},\, t\in[1,T].
\end{eqnarray*}

Solving problem \textbf{CP} yields $\boldsymbol{\bar{x}}$. Thus,
the total power demand at time $t$ given $\bar{x}(t)$ is $d_{t}(\bar{x}(t))$.
Note that $d_{t}(\bar{x}(t))$ is not just server power consumption,
but also includes consumption of power conditioning and cooling systems,
as described in Sec. \ref{sec:formulation.DCM}.

Second, the data center minimizes the cost of satisfying the power
demand due to $d_{t}(\bar{x}(t))$, using both the grid and the on-site
generators. Specifically, it solves the \emph{energy provisioning}
problem which we refer to as \textbf{EP} below.
\begin{eqnarray*}
 \hspace{-0.3cm} &  & \textbf{EP}:\\
 \hspace{-0.3cm} & \min & \sum_{t=1}^{T}\left\{ \psi\left(y(t),p(t),d_{t}(\bar{x}(t))\right)+\beta_{g}[y(t)-y(t-1)]^{+}\right\} \\
 \hspace{-0.3cm} &  & y(0)=0,\\
 \hspace{-0.3cm} & \mbox{var} & y(t)\in\mathbb{N}^{0},\, t\in[1,T].
\end{eqnarray*}

Let $(\boldsymbol{\bar{x}},\boldsymbol{\bar{y}})$ be the solution
obtained by solving \textbf{CP} and \textbf{EP} separately in sequence
and $\left(\boldsymbol{x}^{*},\boldsymbol{y}^{*}\right)$ be the solution
obtained by solving the joint-optimization \textbf{DCM}. Further,
let ${\rm C_{DCM}}(\boldsymbol{x},\boldsymbol{y})$ be the value of
the data center's total cost for solution $(\boldsymbol{x},\boldsymbol{y})$,
including both generator and server costs as represented by the objective
function \eqref{eq:DCM-objective} of problem \textbf{DCM}.  The
additional benefit of joint optimization over optimizing independently
is simply the relationship between ${\rm C_{DCM}}(\boldsymbol{\bar{x}},\boldsymbol{\bar{y}})$
and ${\rm C_{DCM}}\left(\boldsymbol{x}^{*},\boldsymbol{y}^{*}\right)$.
It is clear that $(\boldsymbol{\bar{x}},\boldsymbol{\bar{y}})$ obeys
all the constraints of \textbf{DCM} and hence is a feasible solution
of \textbf{DCM}. Thus, ${\rm C_{DCM}}\left(\boldsymbol{x}^{*},\boldsymbol{y}^{*}\right)\leq{\rm C_{DCM}}(\boldsymbol{\bar{x}},\boldsymbol{\bar{y}}).$
We can measure the factor loss in optimality $\rho$ due to optimizing
separately as opposed to optimizing jointly on the worst-case input
as follows:
\begin{eqnarray*}
\rho\triangleq\max_{\mbox{all inputs}}\frac{{\rm C_{DCM}}(\boldsymbol{\bar{x}},\boldsymbol{\bar{y}})}{{\rm C_{DCM}}\left(\boldsymbol{x}^{*},\boldsymbol{y}^{*}\right)}.
\end{eqnarray*}
The following theorem characterizes the benefit of joint optimization
over optimizing independently.

\begin{thm}
\label{thm:offline decomposition} The factor loss
in optimality $\rho$ by solving the problem \textbf{CP} and \textbf{EP}
in sequence as opposed to optimizing jointly is given by $\rho=LP_{\max}/\left(Lc_{o}+c_{m}\right)$
and it is tight.
\end{thm}
\begin{proof}
Refer to Appendix \ref{sub:Proof-of-Theorem 2}.
\end{proof}

The above theorem guarantees that for \emph{any} time duration $T$, \emph{any }workload $\boldsymbol{a}$,
\emph{any} grid price $\boldsymbol{p}$ and \emph{any} function $g_{t}(x,a)$
as long as it is increasing and convex  in $x$ and $a$, solving
problem \textbf{DCM} by first solving \textbf{CP} then solving \textbf{EP}
in sequence yields a solution that is within a factor $LP_{\max}/\left(Lc_{o}+c_{m}\right)$
of solving \textbf{DCM} directly. Further, the ratio is tight in that
there exists an input to \textbf{DCM} where the ratio ${\rm C_{DCM}}(\boldsymbol{\bar{x}},\boldsymbol{\bar{y}})/{\rm C_{DCM}}\left(\boldsymbol{x}^{*},\boldsymbol{y}^{*}\right)$
equals $LP_{\max}/\left(Lc_{o}+c_{m}\right).$

The theorem shows in a quantitative way that a larger price discrepancy
between the maximum grid price and the on-site power yields a larger
gain by optimizing the energy provisioning and capacity provisioning
jointly. Over the past decade, utilities have been exposing a greater
level of grid price variation to their customers with mechanisms such
as time-of-use pricing where grid prices are much more expensive during
peak hours than during the off-peak periods. This likely leads to
larger price discrepancy between the grid and the on-site power. In
that case, our result implies that a joint optimization of power and
server resources is likely to yield more benefits to a hybrid data
center.

Besides characterizing the benefit of jointly optimizing power and
server resources, the decomposition of problem \textbf{DCM} into problems
\textbf{CP} and \textbf{EP} provides a key approach for our online
algorithm design. Problem \textbf{DCM} has an objective function with
mutually-dependent coupled variables $\boldsymbol{x}$ and $\boldsymbol{y}$
 indicating the server and generator states, respectively. This coupling
(specifically through the function $\psi\left(y(t),p(t),d_{t}(x(t))\right)$
) makes it difficult to design provably good online algorithms. However,
instead of solving problem \textbf{DCM} directly, we devise online
algorithms to solve problems \textbf{CP} that involves only server
variable $\boldsymbol{x}$ and \textbf{EP} that involves only the
generator variables $\boldsymbol{y}$. Combining the online algorithms
for \textbf{CP} and \textbf{EP} respectively yields the desired online
algorithm for \textbf{DCM}.

\section{Online Algorithms for On-Grid Data Centers}\label{sec:GCSR}

We first develop an online algorithm for \textbf{DCM} for an \emph{on-grid}
data center, where there is no on-site power generation, a scenario
that captures most data centers today. Since \emph{on-grid} data center
has no on-site power generation, solving \textbf{DCM} for it reduces
to solving problem \textbf{CP} described in Sec. \ref{sec:decomposition}.

Problems of this kind have been studied in the literature (see \emph{e.g.},
\cite{lin2011dynamic,labpaper}). The difference of our work from
\cite{lin2011dynamic,labpaper} is as follows (also summarized in
Table \ref{tab:Summary-of-differences}). From the modelling aspect,
we explicitly take into account power consumption of both cooling
and power conditioning systems, in addition to servers. From the formulation
aspect, we are solving a different optimization problem, \emph{i.e.}, an
integer program with convex and increasing objective function. From
the theoretical result aspects, we achieve a small competitive ratio
of $2-\alpha_{s}$, which quickly decreases to $1$ as look-ahead
window $w$ increase.

\begin{table}
{\footnotesize }%
\begin{tabular}{|c|c|c|c|}
\hline
 & \textbf{\footnotesize Cooling \& } & \textbf{\footnotesize Optimization} & \textbf{\footnotesize Competitive}\tabularnewline
 & \textbf{\footnotesize Power} & \textbf{\footnotesize Type}{\footnotesize{} } & \textbf{\footnotesize Ratio}\tabularnewline
 & \textbf{\footnotesize Conditioning} &  & \tabularnewline
\hline
{\footnotesize LCP} & {\footnotesize No} & obj: convex & {\footnotesize 3}\tabularnewline
{\footnotesize \cite{lin2011dynamic}} &  & var: continuous & \tabularnewline
\hline
{\footnotesize CSR} & {\footnotesize No} & obj: linear & {\footnotesize $2-\alpha_{s}$}\tabularnewline
{\footnotesize \cite{labpaper}} &  & var: integer & \tabularnewline
\hline
{\footnotesize GCSR } &  & obj: convex & \tabularnewline
{\footnotesize this} & {\footnotesize Yes} & and increasing & {\footnotesize $2-\alpha_{s}$}\tabularnewline
{\footnotesize work} &  & var: integer & \tabularnewline
\hline
\end{tabular}{\footnotesize \par}

{\footnotesize Note that $\alpha_{s}$ is the normalized look-ahead
window size, whose representations are different under the different settings of \cite{labpaper}
and our work.}{\footnotesize \par}

\caption{\footnotesize\label{tab:Summary-of-differences} Comparison of the algorithm GCSR proposed in this paper, CSR in \cite{labpaper}, and LCP in \cite{lin2011dynamic}.}
\end{table}

Recall that \textbf{CP} takes as input the workload $\boldsymbol{a}$,
the grid price $\boldsymbol{p}$ and the time-dependent function $g_{t},\ \forall t$
and outputs the number of active servers $\boldsymbol{x}$. We construct
solutions to \textbf{CP} in a divide-and-conquer fashion. We will
first decompose the demand $\boldsymbol{a}$ into sub-demands and
define corresponding sub-problem for each server, and then solve
capacity provisioning \emph{separately }for each sub-problem. Note
that the key is to correctly decompose the demand and define the subproblems
so that the combined solution is still optimal. More specifically,
we slice the demand as follows: for $1\leq i\leq M=\max_{1\leq t\leq T}\left\lceil a(t)\right\rceil $,
$1\leq t\leq T,$
\[
a_{i}(t)\triangleq\min\left\{ 1,\max\left\{ 0,a(t)-(i-1)\right\} \right\} .
\]
And the corresponding sub-problem $\textbf{CP}_{{\rm i}}$ is defined
as follows.
\begin{eqnarray*}
\textbf{CP}_{{\rm i}}: & \min & \sum_{t=1}^{T}\left\{ p(t)\cdot d_{t}^{i}\cdot x_{i}(t)+\beta_{s}[x_{i}(t)-x_{i}(t-1)]^{+}\right\} \\
 & \mbox{s.t.} & x_{i}(t)\geq a_{i}(t),\\
 &  & x_{i}(0)=0,\\
 & \mbox{var} & x_{i}(t)\in\{0,1\},\, t\in[1,T],
\end{eqnarray*}
where $x_{i}(t)$ indicates whether the $i$-th server is on at time
$t$ and $d_{t}^{i}\triangleq d_{t}(i)-d_{t}(i-1).$ $d_{t}^{i}$
can be interpreted as the power consumption due to the $i$-th server
at $t$.

Problem $\textbf{CP}_{{\rm i}}$ solves the capacity provisioning
problem with inputs workload $\boldsymbol{a_{i}}$, grid price $\boldsymbol{p}$
and $d_{t}^{i}$. The key reason for our decomposition is that $\textbf{CP}_{{\rm i}}$
is easier to solve, since $\boldsymbol{a_{i}}$
take values in $[0,1]$ and exactly one server is required to serve
each $\boldsymbol{a_{i}}$. Generally speaking, a divide-and-conquer
manner may suffer from optimality loss. Surprisingly, as the following
theorem states, the individual optimal solutions for problems $\textbf{CP}_{{\rm i}}$
can be put together to form an optimal solution to the original problem
\textbf{CP}. Denote ${\rm C_{CP_{i}}}(\boldsymbol{x_{i}})$ as the
cost of solution $\boldsymbol{x_{i}}$ for problem $\textbf{CP}_{{\rm i}}$
and ${\rm C_{CP}}(\boldsymbol{x})$ the cost of solution $\boldsymbol{x}$
for problem $\textbf{CP}$.

\begin{thm}
\label{thm:sub-ocp optimal} Consider problem $\mathbf{CP}$ with any $d_{t}(x(t))=g_{t}(x(t),a(t)))$ that is convex in x(t). Let $\bar{\boldsymbol{x}}_{i}$ be an
optimal solution and $\boldsymbol{x}_{i}^{on}$ an online solution
for problem $\mathbf{CP_{i}}$ with workload \textup{$\boldsymbol{a}_{i}$}\textup{\emph{,}}\emph{
}then $\sum_{i=1}^{M}\bar{\boldsymbol{x}}_{i}$ is an optimal solution
for $\mathbf{CP}$ with workload \textup{$\boldsymbol{a}$}. Furthermore,
if \textup{$\forall\boldsymbol{a}_{i},i$}\textup{\emph{, we have
}}\textup{${\rm C_{CP_{i}}}(\boldsymbol{x}_{i}^{on})\\\leq\gamma\cdot{\rm C_{CP_{i}}}(\bar{\boldsymbol{x}}_{i})$}\textup{\emph{
for a constant $\gamma\geq 1$, then}}\textup{ ${\rm C_{CP}}(\sum_{i=1}^{M}\boldsymbol{x}_{i}^{on})\leq\gamma\cdot{\rm C_{CP}}(\sum_{i=1}^{M}\bar{\boldsymbol{x}}_{i}),\:\forall\boldsymbol{a}$.}
\end{thm}
\begin{proof}
Refer to Appendix \ref{sub:Proof-of-Theorem 3}.
\end{proof}

\begin{figure}[t!]
\begin{minipage}[t]{0.48\linewidth}
\centering
\includegraphics[width=0.99\columnwidth]{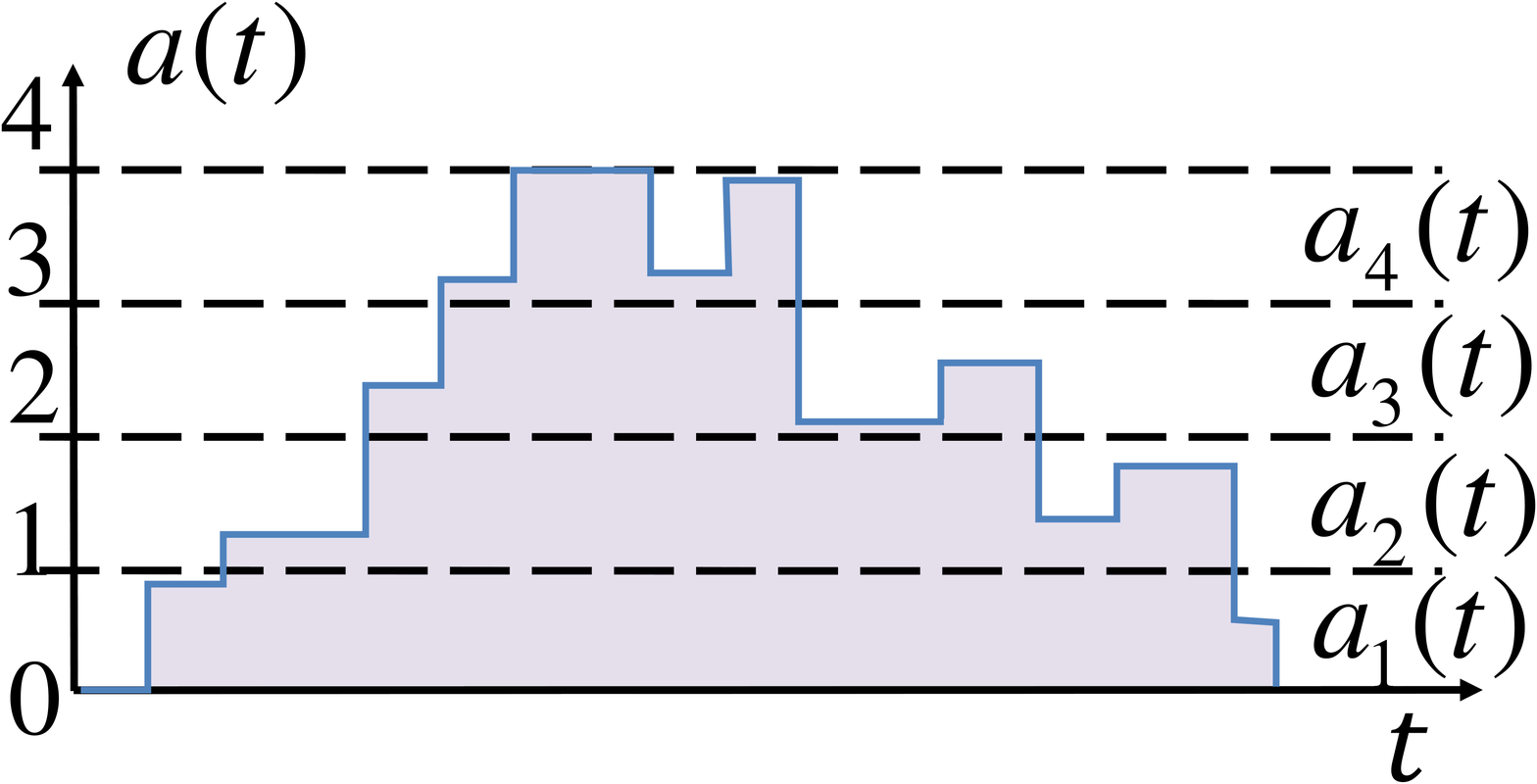}
\caption{\footnotesize\label{fig:example-of-sub-demand-a} An example of how workload $\boldsymbol{a}$ is decomposed into $4$ sub-demands.}
\end{minipage}
\hfill
\begin{minipage}[t]{0.48\linewidth}
\centering
\includegraphics[width=0.99\columnwidth]{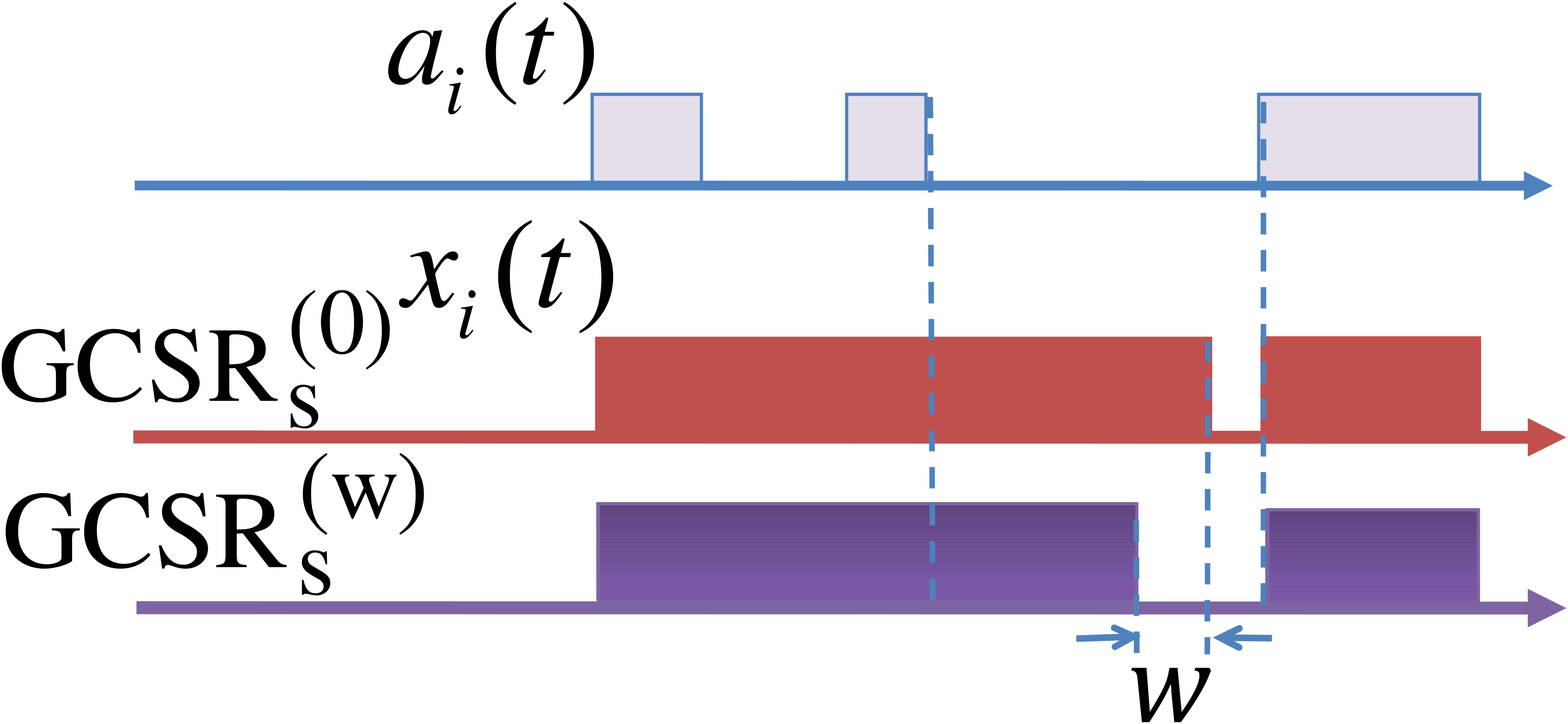}
\caption{\footnotesize\label{fig:An-example-of-OnGridOnline} An example of $a_{i}(t)$ and corresponding solution obtained by $\mathbf{GCSR_{s}^{(w)}}$.}
\end{minipage}
\end{figure}

Thus, it remains to design algorithms for each $\textbf{CP}_{{\rm i}}$.
To solve $\textbf{CP}_{{\rm i}}$ in an online fashion one need only
orchestrate one server to satisfy the workload $\boldsymbol{a}_{i}$
and minimize the total cost. When $a_{i}(t)>0$, we must keep the
server active to satisfy the workload. The challenging part is what
we should do if the server is already active but $a_{i}(t)=0$. Should
we turn off the server immediately or keep it idling for some time?
Should we distinguish the scenarios when the grid price is high versus
low?

Inspired by ``ski-rental'' \cite{borodin1998online} and \cite{labpaper}, we solve $\textbf{CP}_{{\rm i}}$
by the following ``break-even'' idea.  During the idle period,
\emph{i.e.}, $a_{i}(t)=0$, we accumulate an ``idling cost'' and when it
reaches $\beta_{s}$, we turn off the server; otherwise, we keep the
server idling. Specifically, our online algorithm $\mathbf{GCSR_{s}^{(w)}}$ (Generalized Collective Server Rental)
for $\textbf{CP}_{{\rm i}}$ has a look-ahead window $w$. At time
$t$, if there exist $\tau'\in[t,t+w]$ such that the idling cost
till $\tau'$ is at least $\beta_{s}$, we turn off the server; otherwise,
we keep it idling. More formally, we have Algorithm \ref{alg:GCSR(w)}
and its competitive analysis in Theorem \ref{thm:cr.GCSR(w)}. A simple
example of $\mathbf{GCSR_{s}^{(w)}}$ is shown in Fig. \ref{fig:An-example-of-OnGridOnline}.

Our online algorithm for \textbf{CP}, denoted as $\mathbf{GCSR^{(w)}}$,
first employs $\mathbf{GCSR_{s}^{(w)}}$ to solve each $\textbf{CP}_{{\rm i}}$
on workload $\boldsymbol{a}_{i}$, $1\leq i\leq M$, in an online
fashion to produce output $\boldsymbol{x}_{i}^{on}$ and then simply
outputs $\sum_{i=1}^{M}\boldsymbol{x}_{i}^{on}=\boldsymbol{x}^{on}$
as the output for the original problem \textbf{CP}.
\begin{algorithm}[t!]
{
\caption{\label{alg:GCSR(w)} $\mathbf{GCSR_{s}^{(w)}}$ for problem $\textbf{CP}_{{\rm i}}$}
\begin{algorithmic}[1]
\STATE $C_{i}=0$,$x_{i}(0)=0$
\STATE \textbf{at} current time $t,$ \textbf{do}
\STATE {Set $\tau'\leftarrow\min\{t'\in[t,t+w]\:|\: C_{i}+\sum_{\tau=t}^{t'}p(\tau)d_{\tau}^{i}\geq\beta_{s}\}$}
\IF{$a_{i}(t)>0$}
\STATE {$x_{i}(t)=1$ and $C_{i}=0$}
\ELSIF {$\tau'=\mbox{NULL}$ or $\exists\tau\in[t,\tau']$, $a_{i}(\tau)>0$}
\STATE {$x_{i}(t)=x_{i}(t-1)$ and $C_{i}=C_{i}+p(t)d_{t}^{i}x_{i}(t)$}
\ELSE
\STATE {$x_{i}(t)=0$ and $C_{i}=0$}
\ENDIF
\end{algorithmic}
}
\end{algorithm}
\begin{thm}
\label{thm:cr.GCSR(w)} \textup{$\mathbf{GCSR_{s}^{(w)}}$} achieves
a competitive ratio of $2-\alpha_{s}$ for $\mathbf{CP_{i}}$, where
$\alpha_{s}\triangleq\min\left(1,wd_{\min}P_{\min}/\beta_{s}\right)\in[0,1]$
is a ``normalized'' look-ahead window size and $d_{\min}\triangleq\min_{t}\{d_{t}(1)\\-d_{t}(0)\}$.
Hence, according to Theorem \ref{thm:sub-ocp optimal}, \emph{$\mathbf{GCSR^{(w)}}$
}achieves the same competitive ratio for $\mathbf{CP}$. Further, no deterministic online algorithm with a look-ahead window w can achieve a smaller competitive ratio.
\end{thm}
\begin{proof}
Refer to Appendix \ref{sub:Proof-of-Theorem 4}.
\end{proof}

A consequence of Theorem \ref{thm:cr.GCSR(w)} is that when the look-ahead
window size $w$ reaches a break-even interval $\Delta_{s}\triangleq\beta_{s}/(d_{\min}P_{\min})$,
our online algorithm has a competitive ratio of $1.$ That is, having a look-ahead
window larger than $\Delta_{s}$ will not decrease the cost any further.

\section{Online Algorithms for Hybrid Data Centers}\label{sec:Online-Algorithms-Hybrid}

Unlike on-grid data centers, hybrid data centers have on-site power
generation and therefore have to solve both capacity provisioning
(\textbf{CP}) and energy provisioning (\textbf{EP}) to solve the data
center cost minimization (\textbf{DCM}) problem. We design an online
algorithm that we call $\mathbf{DCMON}$ solving \textbf{DCM} as follows.
\begin{enumerate}

\item Run algorithm $\mathbf{GCSR}$ from Sec. \ref{sec:GCSR} to solve
\textbf{CP} that takes workload $\boldsymbol{a}$, grid price
$\boldsymbol{p}$ and time-dependent function $g_{t},\ \forall t$
as input and produces the number of active servers $\boldsymbol{x}^{on}$.

\item Run algorithm $\mathbf{CHASE}$ described in Section \ref{sub:chase}
below to solve \textbf{EP} that takes the energy demand $d_{t}(x^{on}(t))=g_{t}(x^{on}(t),a(t))$
and grid price $p(t),\ \forall t$ as input and decides when to turn
on/off on-site generators and how much power to draw from the generators
and the grid. Note that a similar problem has been studied in the
microgrid scenarios for energy generation scheduling in our previous
work \cite{chase}. In this paper, we adapt
algorithm $\mathbf{CHASE}$ developed in \cite{chase}
to our data center scenarios to solve \textbf{EP} in an online fashion.
\end{enumerate}

For the sake of completeness, we first briefly present the design
behind $\mathbf{CHASE}$ in Sec. \ref{sub:EP offline} and the algorithm
and its intuitions in Sec. \ref{sub:chase}. Then we present the combined
algorithm $\mathbf{DCMON}$ in Sec. \ref{sub:DCMON}.

\subsection{A useful structure of an offline optimal solution
of EP}\label{sub:EP offline}
We first reveal an elegant structure of an offline optimal solution
and then exploit this structure in the design of our online algorithm
$\mathbf{CHASE}$.

\subsubsection{Decompose EP into sub-problems ${EP}_{i}$s}

For the ease of presentation, we denote $e(t)=d_{t}(x^{on}(t))$.
Similar as the decomposition of workload when solving \textbf{CP}, we
decompose the energy demand $\boldsymbol{e}$ into $N$ sub-demands
and define sub-problem for each generator, then solve energy provisioning
\emph{separately} for each sub-problem, where $N$ is the number of
on-site generators. Specifically, for $1\leq i\leq N,\ 1\leq t\leq T$,
\[
e_{i}(t)\triangleq\min\left\{ L,\max\left\{ 0,e(t)-(i-1)L\right\} \right\} .
\]
The corresponding sub-problem $\textbf{EP}_{{\rm i}}$ is in the
same form as $\textbf{EP}$ except that $d_{t}(\bar{x}(t))$ is replaced
by $e_{i}(t)$ and $y(t)$ is replaced by $y_{i}(t)\in\{0,1\}$.
Using this decomposition, we can solve \textbf{EP} on input $\boldsymbol{e}$
by simultaneously solving simpler problems $\textbf{EP}_{{\rm i}}$
on input $\boldsymbol{e}_{i}$ that only involve a single generator.
Theorem \ref{thm:sub-ep optimal} shows that the decomposition incurs
no optimality loss. Denote ${\rm C_{EP_{i}}}(\boldsymbol{y_{i}})$
as the cost of solution $\boldsymbol{y_{i}}$ for problem $\textbf{EP}_{{\rm i}}$
and ${\rm C_{EP}}(\boldsymbol{y})$ the cost of solution $\boldsymbol{y}$
for problem $\textbf{EP}$.
\begin{thm}
\label{thm:sub-ep optimal} Let $\boldsymbol{\bar{y}}_{i}$ be an
optimal solution and $\boldsymbol{y}_{i}^{on}$ an online solution
for $\mathbf{EP_{i}}$ with energy demand \textup{$\boldsymbol{e}_{i}$}\textup{\emph{,}}\emph{
}then $\sum_{i=1}^{N}\boldsymbol{\bar{y}}_{i}$ is an optimal solution
for $\mathbf{EP}$ with energy demand \textup{$\boldsymbol{e}$}.
Furthermore, if \textup{$\forall\boldsymbol{e}_{i},i$}\textup{\emph{,
we have }}\textup{${\rm C_{EP_{i}}}(\boldsymbol{y}_{i}^{on})\leq\gamma\cdot{\rm C_{EP_{i}}}(\boldsymbol{\bar{y}}_{i})$}\textup{\emph{
for a constant $\gamma\geq 1$, then}}\textup{ ${\rm C_{EP}}(\sum_{i=1}^{N}\boldsymbol{y}_{i}^{on})\leq\gamma\cdot{\rm C_{EP}}(\sum_{i=1}^{N}\boldsymbol{\bar{y}}_{i}),\:\forall\boldsymbol{e}$.}
\end{thm}
\begin{proof}
Refer to Appendix \ref{sub:Proof-of-Theorem 5}.
\end{proof}

\subsubsection{Solve each sub-problem ${EP}_{i}$}

Based on Theorem \ref{thm:sub-ep optimal}, it remains to design
 algorithms for each $\textbf{EP}_{{\rm i}}$. Define
\begin{equation}
r_{i}(t)=\psi\left(0,p(t),e_{i}(t)\right)-\psi\left(1,p(t),e_{i}(t)\right).\label{eq:ri(t)}
\end{equation}
$r_{i}(t)$ can be interpreted as the one-slot cost difference between
not using and using on-site generation. Intuitively, if $r_{i}(t)>0$
(resp. $r_{i}(t)<0$), it will be desirable to turn on (resp. off)
the generator. However, due to the startup cost, we should not turn
on and off the generator too frequently. Instead, we should evaluate
whether the \emph{cumulative} gain or loss in the future can offset
the startup cost. This intuition motivates us to define the following
cumulative cost difference $R_{i}(t)$. We set initial values as $R_{i}(0)=-\beta_{g}$
and define $R_{i}(t)$ inductively:
\begin{equation}
R_{i}(t)\triangleq\min\left\{ 0,\max\left\{ -\beta_{g},R_{i}(t-1)+r_{i}(t)\right\} \right\} ,\label{eq:Regret.Definition}
\end{equation}
Note that $R_{i}(t)$ is only within the range $[-\beta_{g},0]$.
An important feature of $R_{i}(t)$ useful later in online algorithm
design is that it can be computed given the past and current inputs.
An illustrating example of $R_{i}(t)$ is shown in Fig. \ref{fig:An-example-of-CHASE}.
\begin{figure}[t!]
\begin{minipage}[t]{0.48\linewidth}
\centering
\includegraphics[width=0.99\columnwidth]{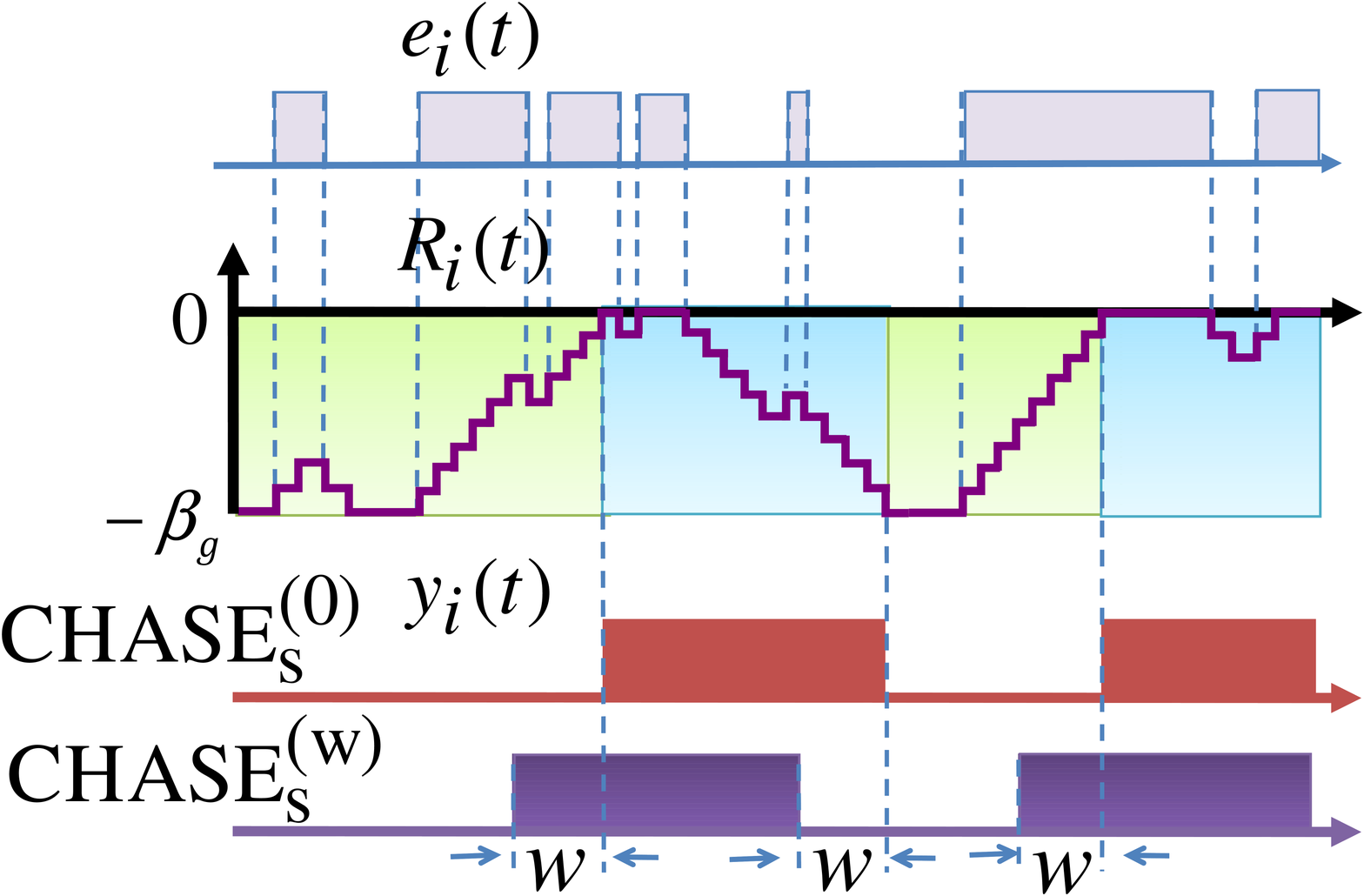}
\caption{\footnotesize\label{fig:An-example-of-CHASE} An example of $e_{i}(t)$, $R_{i}(t)$ and the corresponding solution obtained by $\mathbf{CHASE_{s}^{(w)}}$ for $\textbf{EP}_{{\rm i}}$.}
\end{minipage}
\hfill
\begin{minipage}[t]{0.48\linewidth}
\centering
\includegraphics[width=0.99\columnwidth]{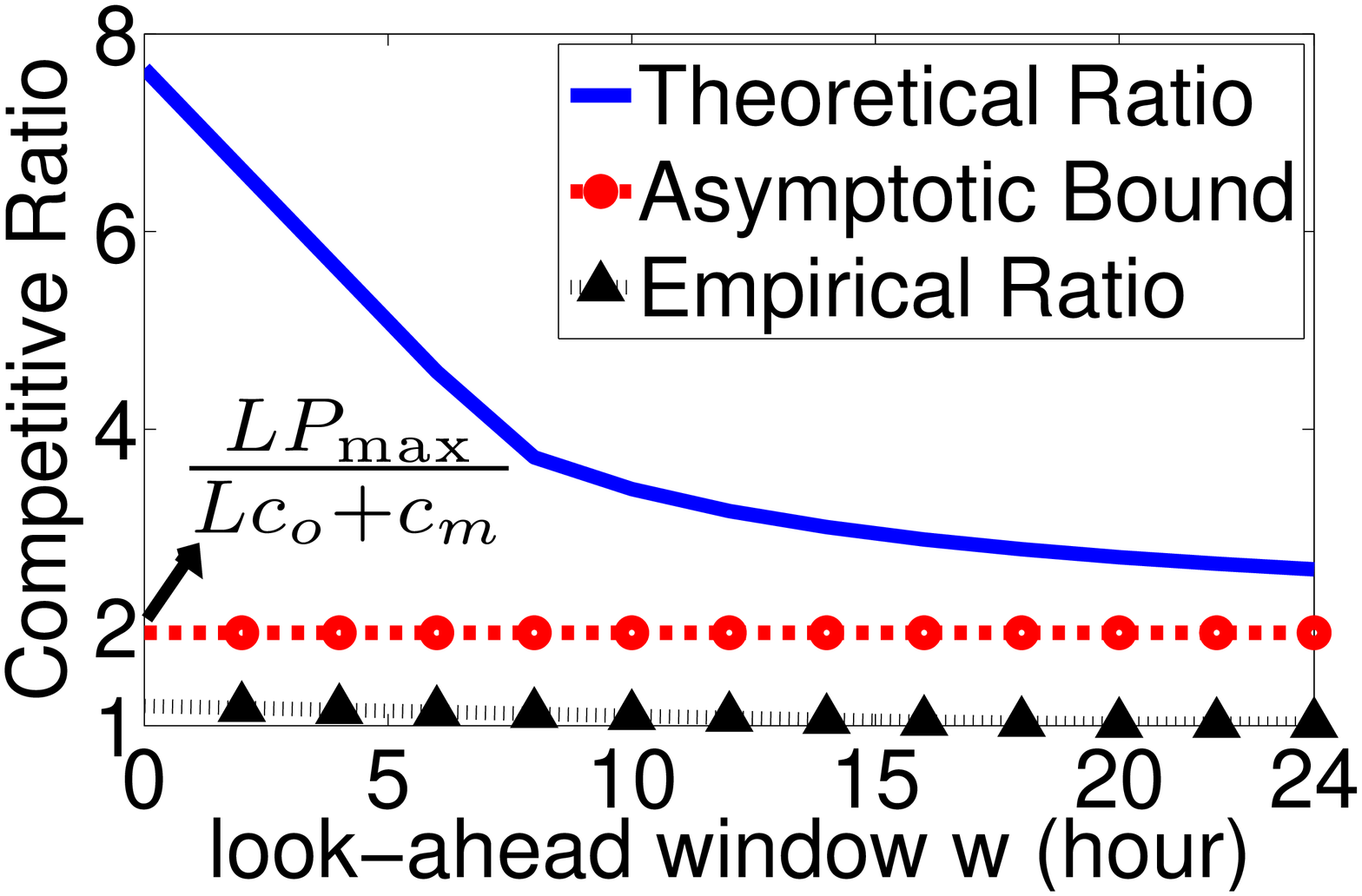}
\caption{\footnotesize\label{fig:ratio-dcmon} Theoretical and empirical ratios of algorithm $\mathbf{DCMON^{(w)}}$ vs. look-ahead window size w.}
\end{minipage}
\end{figure}


Intuitively, when $R_{i}(t)$ hits its boundary $0$, the cost difference
between not using and using on-site generation within a certain period
is at least $\beta_{g}$, which can offset the startup cost. Thus,
it makes sense to turn on the generator. Similarly, when $R_{i}(t)$
hits $-\beta_{g}$, it may be better to turn off the generator and
use the grid. The following theorem formalizes this intuition, and
shows an optimal solution $\bar{y}_{i}(t)$ for problem $\textbf{EP}_{{\rm i}}$
at the time epoch when $R_{i}(t)$ hits its boundary values $-\beta_{g}$
or $0$.

\begin{thm}
\label{thm:EP-offline optimal } There exists an offline optimal solution
for problem $\mathbf{EP_{i}}$ , denoted by $\bar{y}_{i}(t)$,
$1\leq t\leq T$, so that:
\begin{itemize}
\item if $R_{i}(t)=-\beta_{g}$, then $\bar{y}_{i}(t)=0;$
\item if $R_{i}(t)=0$, then $\bar{y}_{i}(t)=1.$
\end{itemize}
\end{thm}
\begin{proof}
Refer to Appendix \ref{sub:Proof-of-Theorem 6}.
\end{proof}

\subsection{Online algorithm CHASE}\label{sub:chase}

Our online algorithm $\mathbf{CHASE_{s}^{(w)}}$ with look-ahead window $w$
exploits the insights revealed in Theorem \ref{thm:EP-offline optimal }
to solve $\textbf{EP}_{{\rm i}}$. The idea behind $\mathbf{CHASE_{s}^{(w)}}$
is to track the offline optimal in an online fashion. In particular,
at time $0$, $R_{i}(0)=-\beta_{g}$ and we set $y_{i}(t)=0$. We
keep tracking the value of $R_{i}(t)$ at every time slot within the
look-ahead window. Once we observe that $R_{i}(t)$ hits values $-\beta_{g}$
or $0$, we set the $y_{i}(t)$ to the optimal solution as Theorem
\ref{thm:EP-offline optimal } reveals; otherwise, keep $y_{i}(t)=y_{i}(t-1)$
unchanged. More formally, we have Algorithm \ref{alg:CHASE(w)} and
its competitive analysis in Theorem \ref{thm:cr.CHASE(w)}. An example
of $\mathbf{CHASE_{s}^{(w)}}$\emph{ }is shown in Fig. \ref{fig:An-example-of-CHASE}.

The online algorithm for \textbf{EP}, denoted as $\mathbf{CHASE^{(w)}}$,
first employs $\mathbf{CHASE_{s}^{(w)}}$ to solve each $\textbf{EP}_{{\rm i}}$
on energy demand $\boldsymbol{e}_{i}$, $1\leq i\leq N$, in an online
fashion to produce output $\boldsymbol{y}_{i}^{on}$ and then simply
outputs $\sum_{i=1}^{N}\boldsymbol{y}_{i}^{on}$ as the output for
the original problem \textbf{EP}.
\begin{algorithm}[htb!]
{
\caption{\label{alg:CHASE(w)}$\mathbf{CHASE_{s}^{(w)}}$ for problem $\textbf{EP}_{{\rm i}}$}
\begin{algorithmic}[1]
\STATE \textbf{at} current time $t,$ \textbf{do}
\STATE Obtain $\left(R_{i}(\tau)\right)_{\tau=t}^{t+w}$
\STATE Set $\tau'\leftarrow\min\{\tau\in[t,t+w]\:|\: R_{i}(\tau)=0\ \mbox{or}\ -\beta_{g}\}$
\IF{$\tau'=\mbox{NULL}$}
\STATE {$y_{i}(t)=y_{i}(t-1)$}
\ELSIF {$R_{i}(\tau')=0$}
\STATE {$y_{i}(t)=1$}
\ELSE
\STATE {$y_{i}(t)=0$}
\ENDIF
\end{algorithmic}
}
\end{algorithm}

\begin{thm}
\textbf{\label{thm:cr.CHASE(w)}}$\mathbf{CHASE_{s}^{(w)}}$ for problem
$\mathbf{EP_{i}}$ with a look-ahead window $w$ has a competitive
ratio of
\[
1+\frac{2\beta_{g}\left(LP{}_{\max}-Lc_{o}-c_{m}\right)}{\beta_{g}LP{}_{\max}+wc_{m}P{}_{\max}\left(L-\frac{c_{m}}{P{}_{\max}-c_{o}}\right)}.
\]
Hence, according to Theorem \ref{thm:sub-ep optimal},
$\mathbf{CHASE^{(w)}}$ achieves the same competitive ratio
for problem $\mathbf{EP}$.
\end{thm}
\begin{proof}
Refer to Appendix \ref{sub:Proof-of-Theorem 7}.
\end{proof}

\subsection{Combining GCSR and CHASE}\label{sub:DCMON}

Our algorithm $\mathbf{DCMON^{(w)}}$ for solving problem \textbf{DCM}
with a look-ahead window of $w\geq0$, \emph{i.e.}, knowing grid prices $p(\tau)$,
workload $a(\tau)$ and the function $g_{\tau},1\leq\tau\leq t+w,$
at time $t$, first uses $\mathbf{GCSR}$ from Sec. \ref{sec:GCSR}
to solve problem \textbf{CP} and then uses $\mathbf{CHASE}$ in Sec.
\ref{sub:chase} to solve problem \textbf{EP}. An important observation
is that the available look-ahead window size for $\mathbf{GCSR}$ to
solve \textbf{CP} is $w$, \emph{i.e.}, knows $p(\tau)$, $a(\tau)$ and $g_{\tau}$, $1\leq\tau\leq t+w,$ at time $t$; however, the available look-ahead window size for $\mathbf{CHASE}$
to solve \textbf{EP} is only $\left[w-\Delta_{s}\right]^{+}$, \emph{i.e.}, knows $p(\tau)$ and $e(\tau)=d_{\tau}(x^{on}(\tau))$, $1\leq\tau\leq t+\left[w-\Delta_{s}\right]^{+},$ at time $t$ ($\Delta_{s}$ is the break-even interval defined in Sec. \ref{sec:GCSR}).

This is because at time $t,$ $\mathbf{CHASE^{(w)}}$ knows grid prices
$p(\tau)$, workload $a(\tau)$ and the function $g_{\tau},1\leq\tau\leq t+w.$
However, not all the energy demands $\left(e(\tau)\right)_{\tau=1}^{t+w}$
are known by $\mathbf{CHASE^{(w)}}$. Because we derive the server
state $\boldsymbol{x^{on}}$ by solving problem \textbf{CP} using
our online algorithm $\mathbf{GCSR^{(w)}}$ using $p(\tau),\ a(\tau),\ g_{\tau},\ 1\leq\tau\leq t+w.$
A key observation is that at time $t$ it is not possible to compute
$\boldsymbol{x^{on}}$ for the full look-ahead window of $t+w,$ since
$x^{on}(t+1),\ldots,x^{on}(t+w)$ may depend on inputs $p(\tau),\ a(\tau),\ g_{\tau},\tau>t+w$
that our algorithm does not yet know. Fortunately, for $w\geq\Delta_{s}$
we can determine all $x^{on}(\tau),1\leq\tau\leq t+\left[w-\Delta_{s}\right]^{+}$
given inputs within the full look-ahead window. That is, while we
knows the grid prices \textbf{$p$}, the workload $a$ and the function
$g_{t}$ for the full look-ahead window $w$, the server state \textbf{$x^{on}$
}is known only for a smaller window of $\left[w-\Delta_{s}\right]^{+}$.
Thus, the energy demand $e(\tau)=d_{\tau}(x^{on}(\tau))=g_{\tau}(x^{on}(\tau),a(\tau)),\ 1\leq\tau\leq t+\left[w-\Delta_{s}\right]^{+}$
is available for $\mathbf{CHASE^{(w)}}$ at time $t$.

Thus, a bound on the competitive ratio of $\mathbf{DCMON^{(w)}}$
 is the product of competitive ratios for $\mathbf{GCSR^{(w)}}$ and
\\$\mathbf{CHASE^{(\left[w-\Delta_{s}\right]^{+})}}$ from Theorems
\ref{thm:cr.GCSR(w)} and \ref{thm:cr.CHASE(w)}, respectively, and
the optimality loss ratio $LP_{\max}/\left(Lc_{o}+c_{m}\right)$ due
to the offline-decomposition stated in Sec. \ref{sec:decomposition},
which is given in the following Theorem.

\begin{thm}
\textbf{\label{thm:cr.DCMON(w)}}$\mathbf{DCMON^{(w)}}$ for problem
\textbf{DCM} has a competitive ratio of\\
{
\begin{equation}
\frac{P_{\max}\left(2-\alpha_{s}\right)}{c_{o}+c_{m}/L}\left[1+\frac{2\left(LP{}_{\max}-Lc_{o}-c_{m}\right)}{LP{}_{\max}+\alpha_{g}P{}_{\max}\left(L-\frac{c_{m}}{P{}_{\max}-c_{o}}\right)}\right].\label{eq:cr.JSGR(w)}
\end{equation}
}The ratio is also upper-bounded by {
\[
\frac{P_{\max}\left(2-\alpha_{s}\right)}{c_{o}+c_{m}/L}\left[1+2\frac{P_{\max}-c_{o}}{P_{\max}}\cdot\frac{1}{1+\alpha_{g}}\right],
\]
}where $\alpha_{s}=\min\left(1,w/\Delta_{s}\right)\in[0,1]$
and $\alpha_{g}\triangleq\frac{c_{m}}{\beta_{g}}\left[w-\Delta_{s}\right]^{+}\\
\in[0,+\infty)$ are {}``normalized'' look-ahead window sizes.
\end{thm}
\begin{proof}
Refer to Appendix \ref{sub:Proof-of-Theorem 8}.
\end{proof}

As the look-ahead window size $w$ increases, the competitive ratio
in Theorem \ref{thm:cr.DCMON(w)} decreases to $LP_{\max}/\left(Lc_{o}+c_{m}\right)$
(c.f. Fig. \ref{fig:ratio-dcmon}), the inherent approximation ratio
introduced by our offline decomposition approach discussed in Section
\ref{sec:decomposition}. However, the real trace based empirical
performance of $\mathbf{DCMON^{(w)}}$ without look-ahead
is already close to the offline optimal, \emph{i.e.}, ratio close to 1 (c.f. Fig. \ref{fig:ratio-dcmon}).

\section{Empirical Evaluation}\label{sec:Numerical-Experiments}

We evaluate the performance of our algorithms by simulations based
on real-world traces with the aim of (i) corroborating the empirical performance
of our online algorithms under various realistic settings and the
impact of having look-ahead information, (ii) understanding the benefit
of opportunistically procuring energy from both on-site generators
and the grid, as compared to the current practice of purchasing from
the grid alone, (iii) studying how much on-site energy is needed for
substantial cost benefits.

\subsection{Parameters and Settings}

\emph{Workload trace}: We use the workload traces from the Akamai
network \cite{Akamai,NygrenSS10} that is the currently the world's
largest content delivery network. 
 The traces
measure the workload of Akamai servers serving web content to actual
end-users. Note that our workload is of the ``request-and-response''
type that we model in our paper. We use traces from the Akamai servers
deployed in the New York and San Jose data centers that record the
hourly average load served by each deployed server over 22 days from Dec. 21, 2008 to Jan. 11, 2009. The New
York trace represents 2.5K servers that served about $1.4\times10^{10}$
requests and $1.7\times10^{13}$ bytes of content to end-users during
our measurement period. The San Jose trace represents 1.5K servers
that served about $5.5\times10^{9}$ requests and $8\times10^{12}$
bytes of content. We show the workload in Fig. \ref{fig:trace}, in
which we normalize the load by the server's service capacity. The
workload is quite characteristic in that it shows daily variations
(peak versus off-peak) and weekly variations (weekday versus weekend).

\emph{Grid} \emph{price}: We use traces of hourly grid power prices
in New York \cite{nationalgridus} and San Jose \cite{PG&E} for the
same time period, so that it can be matched up with the workload traces
(c.f. Fig. \ref{fig:trace}).
Both workload and grid price traces show strong diurnal properties: in the daytime, the workload
and the grid price are relatively high; at night, on the contrary,
both are low. This indicates
the feasibility of reducing the data center cost by using the energy
from the on-site generators during the daytime and use the grid at night.

\emph{Server model}: As mentioned in Sec. \ref{sec:problem.formulation},
we assume the data center has a sufficient number of homogeneous servers
to serve the incoming workload at any given time. Similar to a typical
setting in \cite{palasamudram2012using}, we use the standard linear
server power consumption model. We assume that each server consumes
0.25KWh power per hour at full capacity and has a power proportional
factor (PPF=$(c_{peak}-c_{idle})/c_{peak}$) of 0.6, which gives us
$c_{idle}=0.1KW$, $c_{peak}=0.25KW$. In addition, we assume the server switching cost equals the energy
cost of running a server for 3 hours. If we assume an average
grid price as the price of energy, we get about $\beta_{s}=\$0.08$.

\emph{Cooling and power conditioning system model}: We consider a water chiller cooling system.
According to \cite{weather}, during this 22-day winter period the average high and low temperatures of New York are $41^{\circ}F$ and $29^{\circ}F$, respectively. Those of San Jose are $58^{\circ}F$ and $41^{\circ}F$, respectively. Without loss of generality, we take the high temperature as the daytime temperature and the low temperature as the nighttime temperature. Thus, according to \cite{pelley2009understanding}, the power consumed by water chiller cooling systems of the New York and San Jose data centers are about
\[
f_{c,NY}^{t}(b)=\begin{cases}
(0.041b^{2}+0.144b+0.047)b_{\max}, & \mbox{at daytime},\\
(0.03b^{2}+0.136b+0.042)b_{\max}, & \mbox{at nighttime},
\end{cases}
\]
and
\[
f_{c,SJ}^{t}(b)=\begin{cases}
(0.06b^{2}+0.16b+0.054)b_{\max}, & \mbox{at daytime},\\
(0.041b^{2}+0.144b+0.047)b_{\max}, & \mbox{at nighttime},
\end{cases}
\]
where $b_{\max}$ is the maximum server power consumption and $b$
is the server power consumption normalized by $b_{\max}$. The maximum
server power consumption of the New York and San Jose data centers
are $b_{\max}^{NY}=2500\times0.25=625KW$ and $b_{\max}^{SJ}=1500\times0.25=375KW$.
Besides, the power consumed by the power conditioning system, including
PDUs and UPSs, is $f_{p}(b)=(0.012b^{2}+0.046b+0.056)b_{\max}$ \cite{pelley2009understanding}.

\emph{Generator model}: We adopt generators with specifications the
same as the one in \cite{tecogen}. The maximum output of the generator
is 60KW, \emph{i.e.}, $L=60KW$. The incremental cost to generate
an additional unit of energy $c_{o}$ is set to be \$0.08/KWh, which
is calculated according to the gas price \cite{nationalgridus} and the generator efficiency
\cite{tecogen}. Similar to \cite{stadlerdistributed}, we set the sunk cost of running the generator
for unit time $c_{m}=\$1.2$ and the startup cost $\beta_{g}$
equivalent to the amortized capital cost, which gives $\beta_{g}=\$24$.
Besides, we assume the number of generators $N=10$, which is enough
to satisfy all the energy demand for this trace and model we use.

\emph{Cost benchmark}: Current data centers usually do not use dynamic
capacity provisioning and on-site generators. Thus, we use the cost
incurred by static capacity provisioning with grid power as the benchmark
using which we evaluate the cost reduction due to our algorithms.
Static capacity provisioning runs a fixed number of servers at all
times to serve the workload, without dynamically turning on/off the
servers. For our benchmark, we assume that the data center has complete
workload information ahead of time and provisions exactly to satisfy
the peak workload and uses only grid power. Using such a benchmark
gives us a conservative evaluation of the cost saving from our algorithms.

\emph{Comparisons of Algorithms}: We compare four algorithms:
our online and offline optimal algorithms in on-grid scenarios, \emph{i.e.}, $\mathbf{GCSR}$ and $\mathbf{CPOFF}$,
and hybrid scenarios, \emph{i.e.}, $\mathbf{DCMON}$ and $\mathbf{DCMOFF}$.
\begin{figure}[t!]
\centering{}\subfloat[\label{fig:trace-ny}New York]{\begin{centering}
\includegraphics[width=0.485\columnwidth]{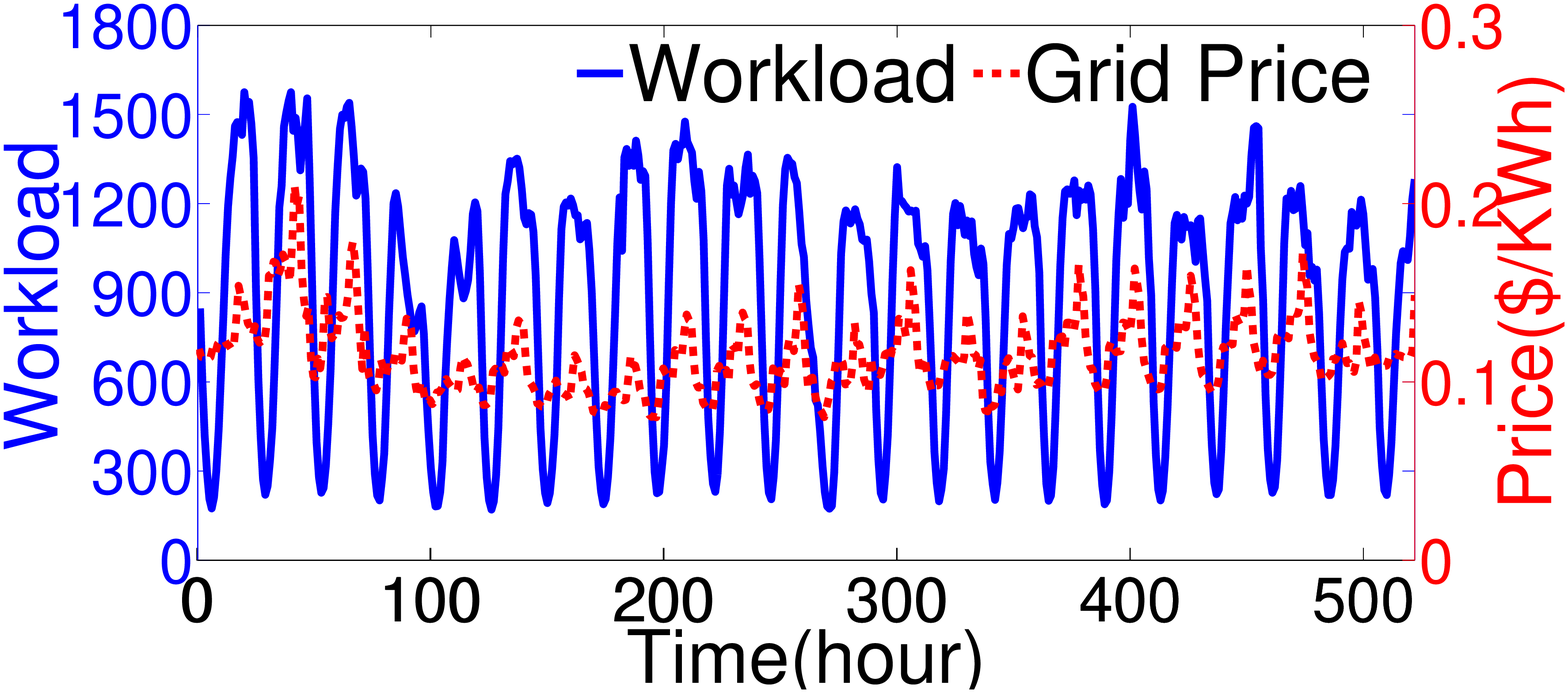}
\par\end{centering}
}\subfloat[\label{fig:trace-sj}San Jose]{\begin{centering}
\includegraphics[width=0.485\columnwidth]{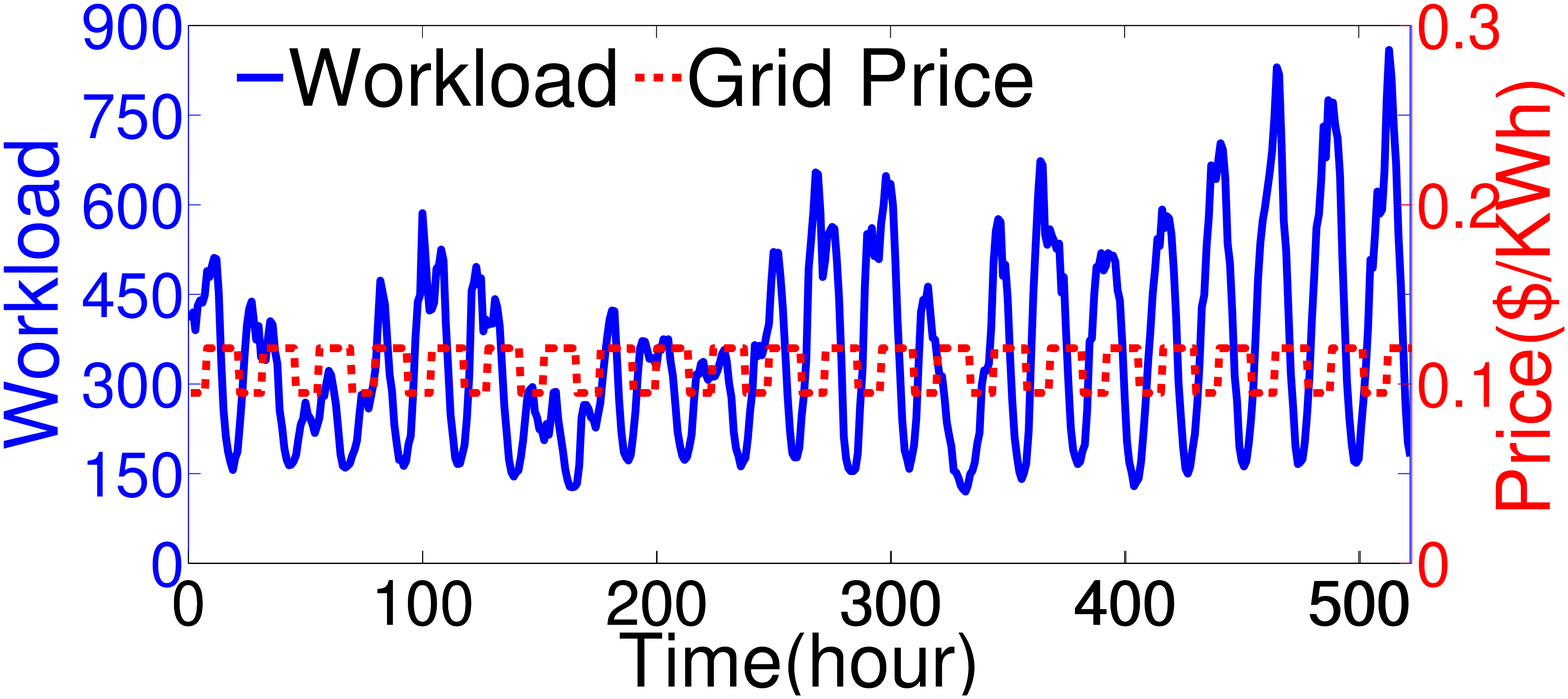}
\par\end{centering}
}$\;\;$\caption{\label{fig:trace}Real-world workload from Akamai and the grid power
price.}
\end{figure}

\subsection{Impact of Model Parameters on Cost Reduction}

We study the cost reduction provided by our offline and online algorithms
for both on-grid and hybrid data centers using the New York trace
unless specified otherwise. We assume no look-ahead information is
available when running the online algorithms. We compute the cost
reduction (in percentage) as compared to the cost benchmark which
we described earlier. When all parameters take their default values,
our offline (resp. online) algorithms provide up to 12.3\% (resp.,
7.3\%) cost reduction for on-grid and 25.8\% (resp., 20.7\%) cost
reduction for hybrid data centers (c.f. Fig. \ref{fig:cost.reduction.as.a.func.of.parameters}. The default value of $c_{o}$ is \$0.08/KWh.).
Note that the online algorithms provide cost reduction that are 5\%
smaller than offline algorithms on account of their lack of knowledge
of future inputs. Further, note that cost reduction of a hybrid data
center is larger than that of a on-grid data center, since hybrid
data center has the ability to generate energy on-site to avoid higher
grid prices. Nevertheless, the extent of cost reduction in all cases
is high providing strong evidence for the need to perform energy and
server capacity optimizations.

Data centers may deploy different types of servers and generators
with different model parameters. It is then important to understand
the impact on cost reduction due to these parameters. We first study
the impact of varying $c_{o}$ (c.f. Fig. \ref{fig:cost.reduction.as.a.func.of.parameters}).
For a hybrid data center, as $c_{o}$ increases the cost of on-site
generation increases making it less effective for cost reduction (c.f
Fig. \ref{fig:cost reduction co}). For the same reason, the cost
reduction of a hybrid data center tends to that of the on-grid data
center with increasing $c_{o}$ as on-site generation becomes less
economical.

We then study the impact of power proportional factor (PPF). More
specifically, we fix $c_{peak}=0.25KW$, and vary PPF from 0 to 1
(c.f. Fig. \ref{fig:cost reduction ppf}). As PPF increases, the server idle power decreases, thus dynamic provisioning has lesser impact on the cost reduction. This explains why \textbf{CP} achieves
no cost reduction when PPF=1. Since \textbf{DCM} also solves \textbf{CP} problem, its
performance degrades with increasing PPF as well.
\begin{figure}[t£¡]
\subfloat[\label{fig:cost reduction co} Cost Reduction vs. $c_{o}$]{\begin{raggedright}
\includegraphics[width=0.5\columnwidth]{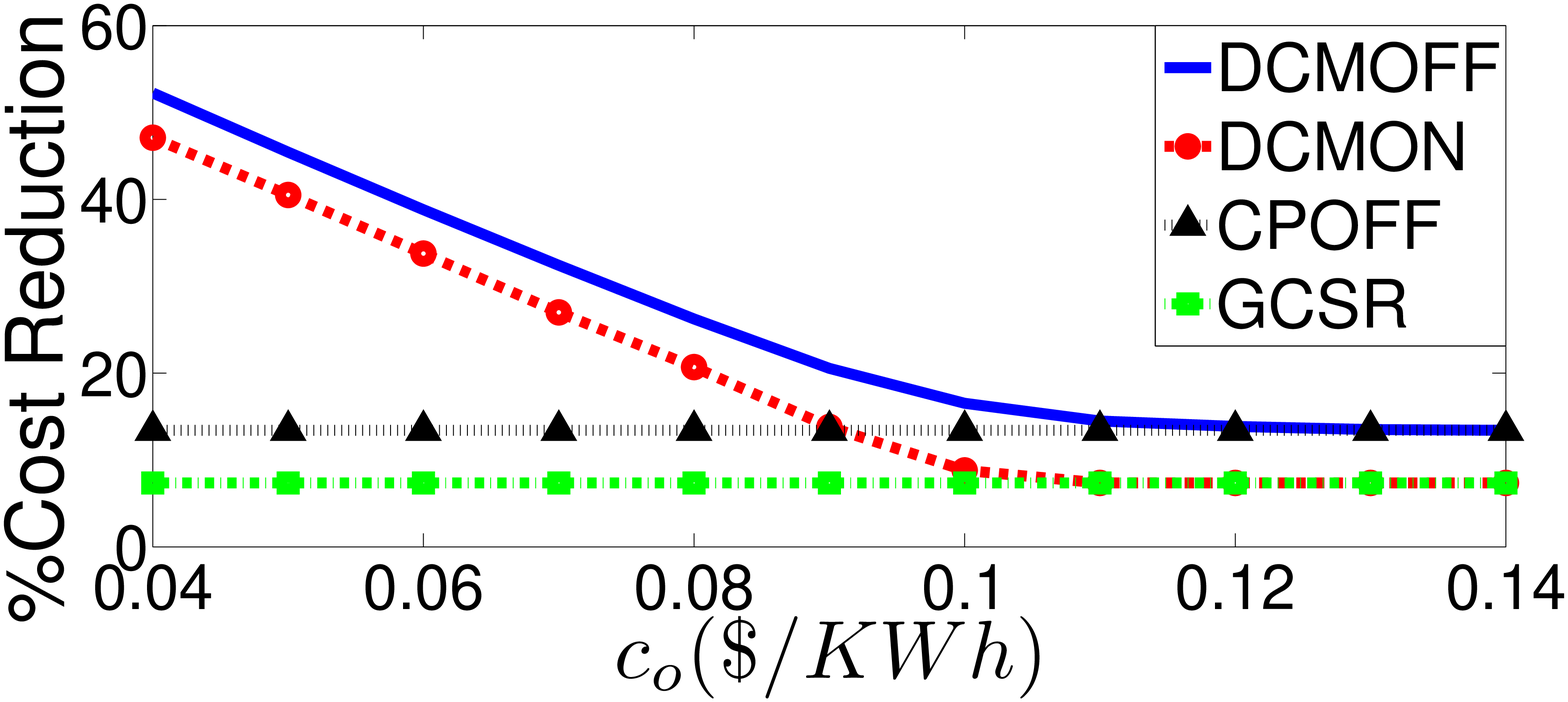}
\par\end{raggedright}
}\subfloat[\label{fig:cost reduction ppf} Cost Reduction vs. PPF ]{\includegraphics[width=0.5\columnwidth]{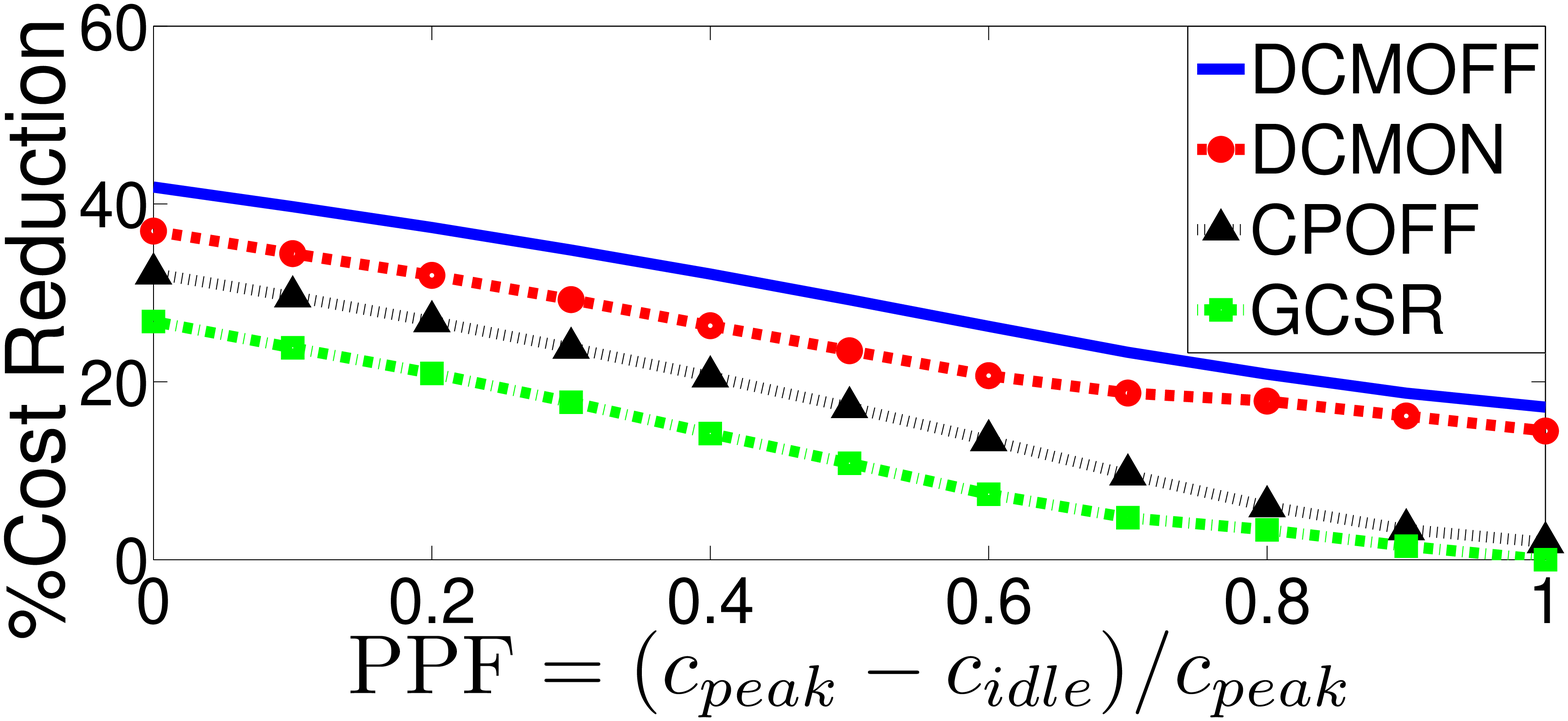}}\caption{\label{fig:cost.reduction.as.a.func.of.parameters}Variation of cost reduction with model parameters. }
\end{figure}

\subsection{The Relative Value of Energy versus Capacity Provisioning}

In this subsection, we use both New York and San Jose traces. For
a hybrid data center, we ask which optimization provides a larger
cost reduction: energy provisioning (\textbf{EP}) or server capacity
provisioning (\textbf{CP}) in comparison with the joint optimization
of doing both (\textbf{DCM}). The cost reductions of different optimization
are shown in Fig. \ref{fig:Relative-values}.

\begin{figure}[t£¡]
\subfloat[\label{fig:Relative-values-ny}New York]{\begin{centering}
\includegraphics[width=0.48\columnwidth]{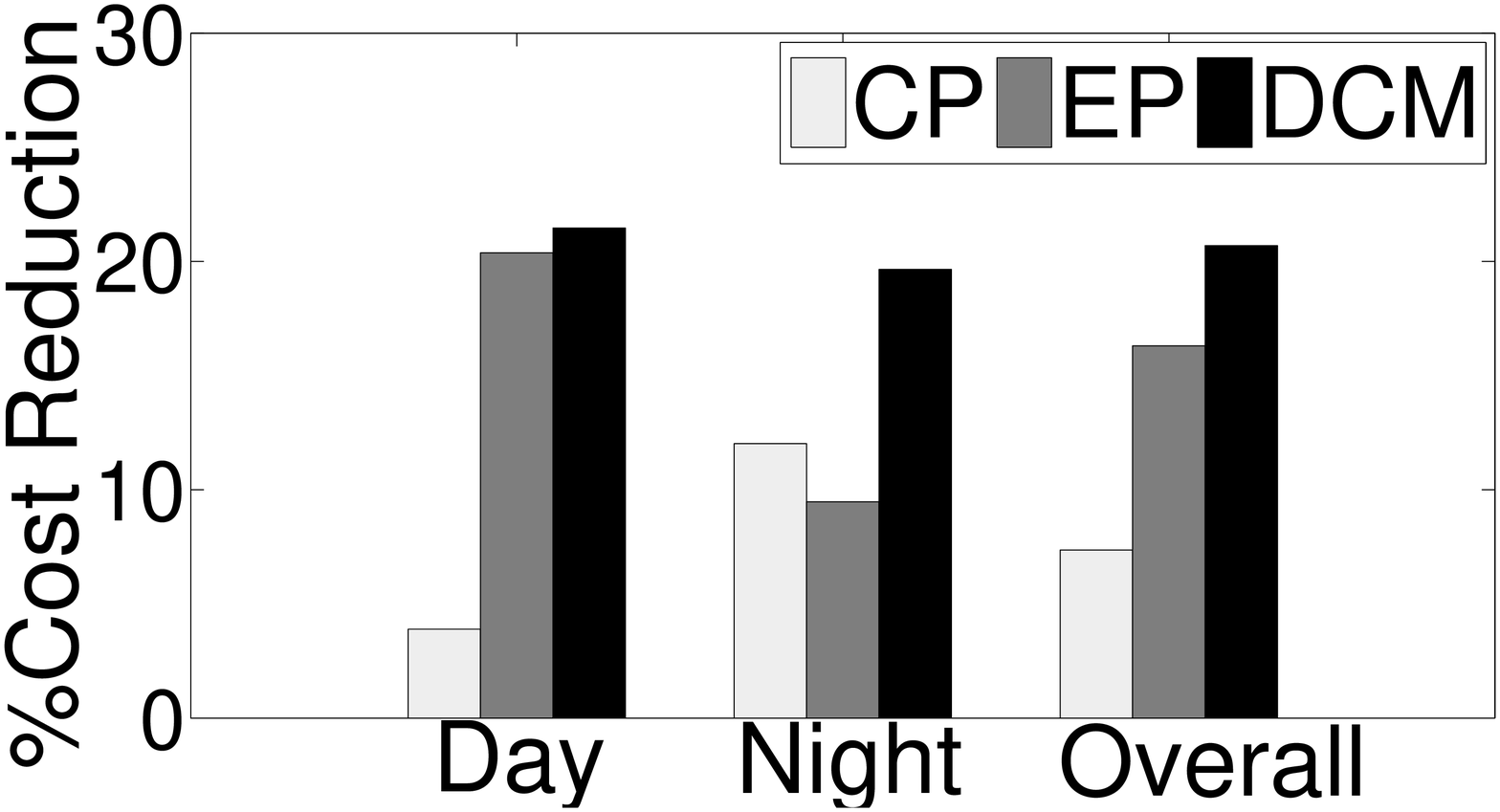}
\par\end{centering}
}\subfloat[\label{fig:Relative-values-sj}San Jose]{\begin{centering}
\includegraphics[width=0.48\columnwidth]{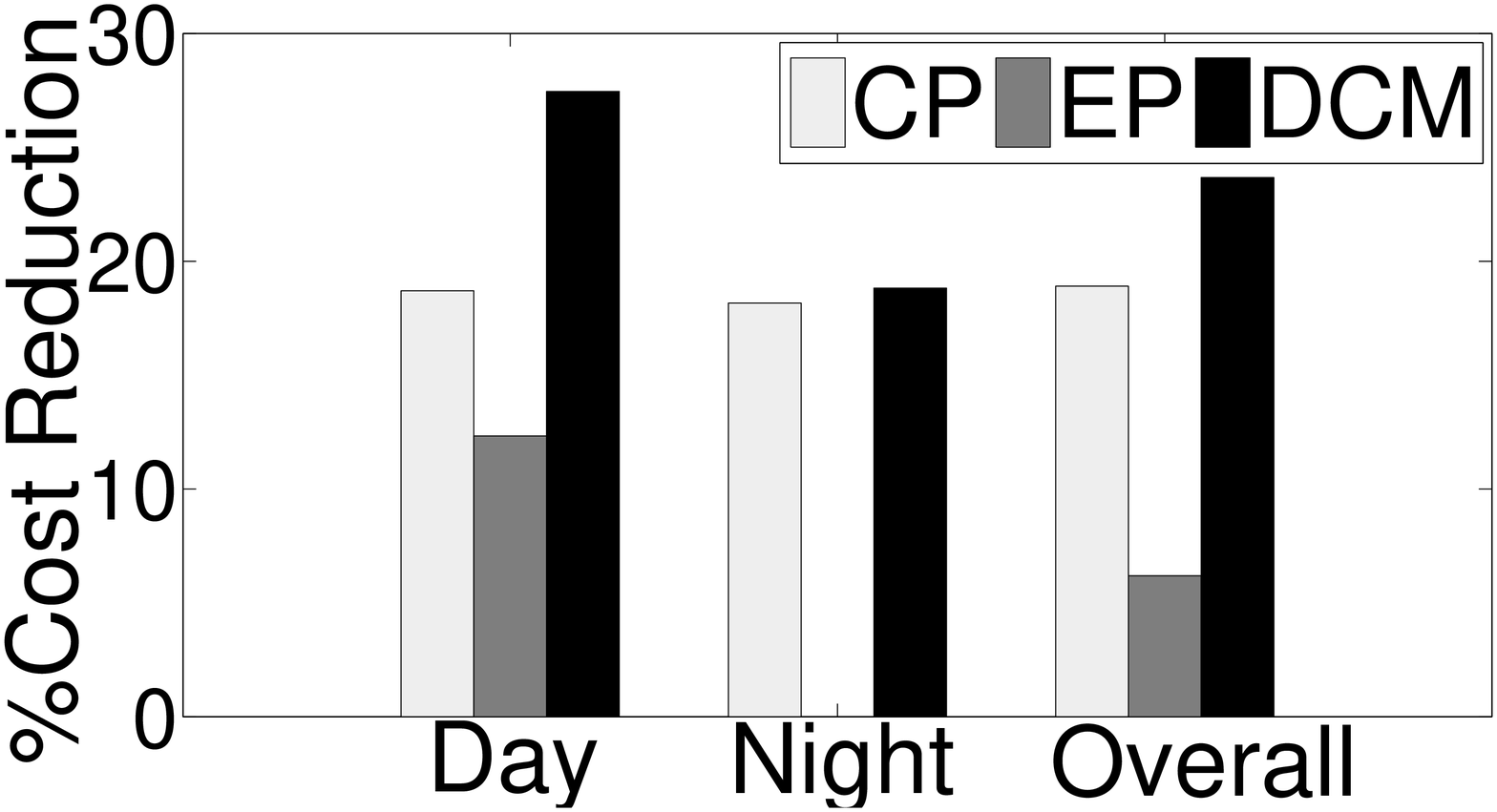}
\par\end{centering}

}\caption{\label{fig:Relative-values}Relative values of CP, EP, and DCM.}
\end{figure}

For the New York scenario in Fig. \ref{fig:Relative-values-ny}, overall,
we see that \textbf{EP}, \textbf{CP}, and \textbf{DCM} provide cost
reductions of 16.3\%, 7.3\%, and 20.7\%, respectively. However, note
that during the day doing \textbf{EP} alone provides almost as much
cost reduction as the joint optimization \textbf{DCM}. The reason
is that during the high traffic hours in the day, solving \textbf{EP}
to avoid higher grid prices provides a larger benefit than optimizing
the energy consumption by server shutdown. The opposite is true during
the night where \textbf{CP} is more critical than \textbf{EP}, since
minimizing the energy consumption by shutting down idle servers yields
more benefit.

For the San Jose scenario in Fig. \ref{fig:Relative-values-sj}, overall,
\textbf{EP}, \textbf{CP}, and \textbf{DCM} provide cost reductions
of 6.1\%, 19\%, and 23.7\%, respectively. Compared to the New York
scenario, the reason why \textbf{EP} achieves so little cost reduction
is that the grid power is cheaper and thus on-site generation is not
that economical. Meanwhile, \textbf{CP} performs closer to \textbf{DCM},
which is because the workload curve is highly skew (shown in Fig.
\ref{fig:trace-sj}) and dynamic provisioning for the server capacity
saves a lot of server idling cost as well as cooling and power conditioning
cost.

In a nutshell, \textbf{EP} favors high grid power price while workload
with less regular pattern makes \textbf{CP} more competitive.

\subsection{Benefit of Looking Ahead}

We evaluate the cost reduction benefit of increasing the look-ahead
window. From Fig. \ref{fig:Cost-Reduction-look-ahead}, we observe
that while the performance of our online algorithms are already good
when there is no look-ahead information, they quickly
improve to the offline optimal when a small amount of look-ahead, \emph{e.g.,} 6 hours, is available,
indicating the value of short-term prediction of inputs. Note that while the competitive ratio analysis in Theorem \ref{thm:cr.DCMON(w)}
is for the worst case inputs, our online algorithms perform much closer
to the offline optimal for realistic inputs.

\subsection{How Much On-site Power Production is Enough}

Thus far, in our experiments, we assumed that a hybrid data center
had the ability to supply all its energy from on-site power generation
($N=10$). However, an important question is how much investment should
a data center operator make in installing on-site generator capacity
to obtain largest cost reduction.

\begin{figure}[t!]
\subfloat[\label{fig:Cost-Reduction-look-ahead} Cost Reduction vs. look-ahead window size w]{\includegraphics[width=0.5\columnwidth]{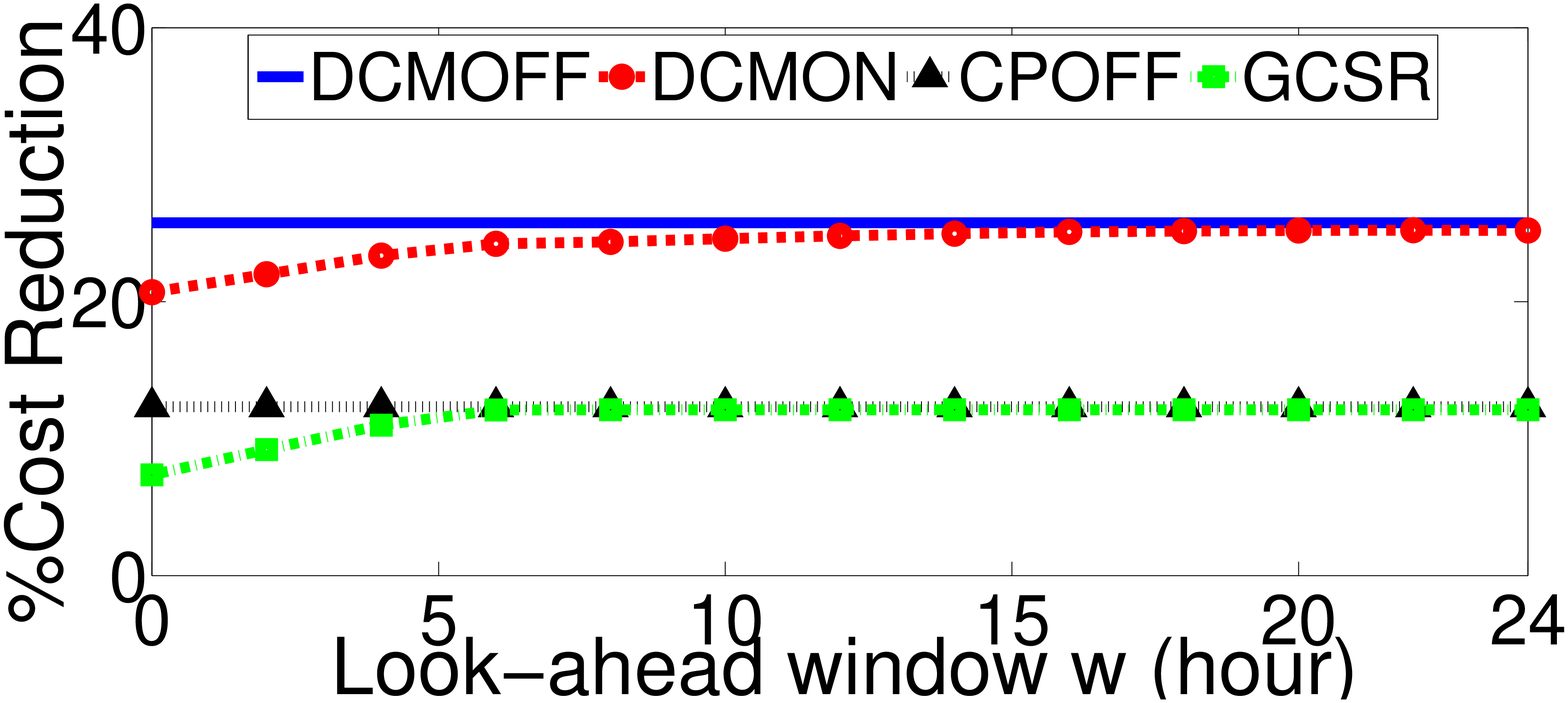}}
\subfloat[\label{fig:insufficient} Cost Reduction vs. percentage of on-site power production capacity]{\includegraphics[width=0.5\columnwidth]{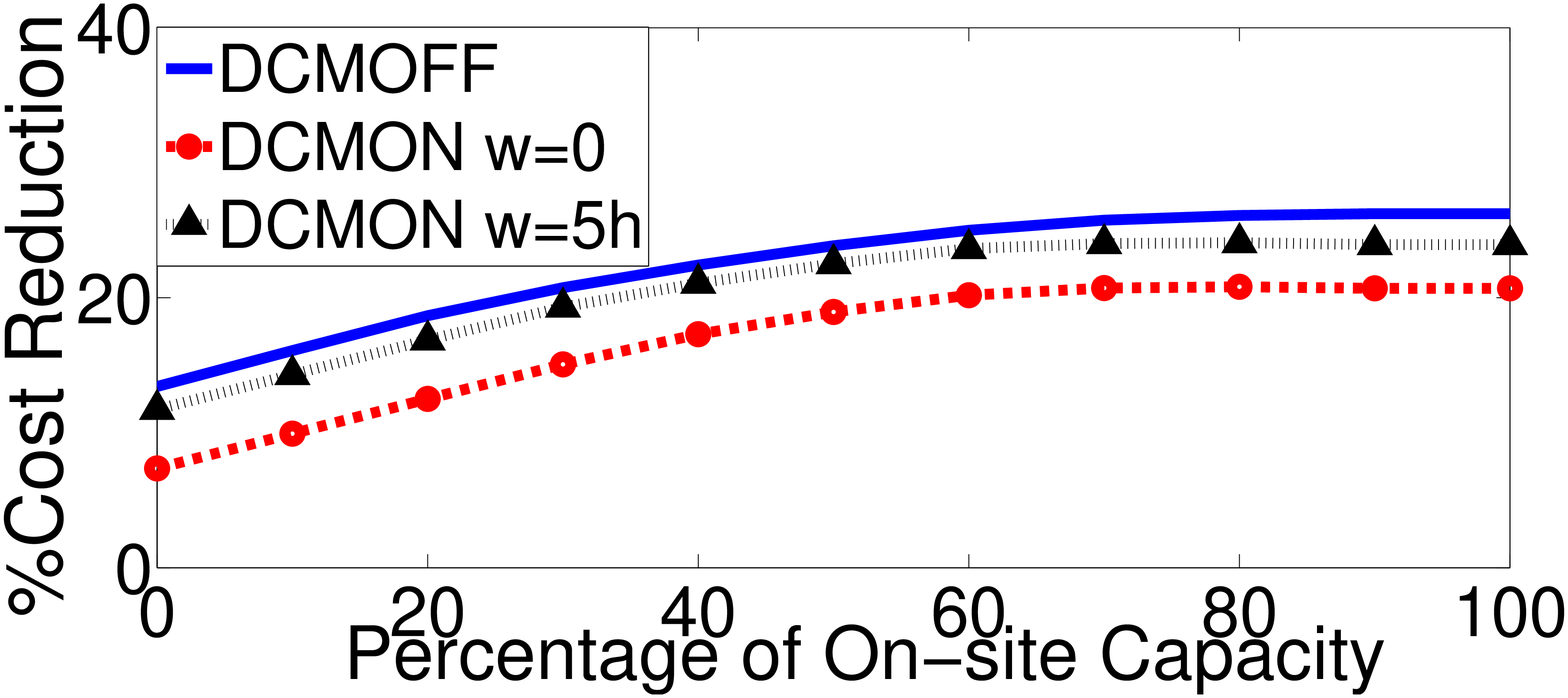}}
\caption{ Variation of cost reduction with look-ahead and on-site capacity.}
\end{figure}

More specifically, we vary the number of on-site generators $N$ from 0 to 10
and show the corresponding performances of our algorithms. Interestingly,
in Fig. \ref{fig:insufficient}, our results show that provisioning
on-site generators to produce 80\% of the peak power demand of the
data center is sufficient to obtain all of the cost reduction benefits.
Further, with just 60\% on-site power generation capacity we can achieve
95\% of the maximum cost reduction. The intuitive reason is that
most of time the demands of the data center are significantly lower
than their peaks.

\section{Related Work }

Our study is among a series of work on dynamic provisioning in data
centers and power systems \cite{stanojevic2010distributed,linonline2012,qureshi2009cutting}.

In particular,
for the capacity provisioning problem, \cite{lin2011dynamic} and \cite{labpaper} propose online algorithms with performance guarantee to reduce servers operating cost under convex and linear mixed integer optimization scenarios, respectively. Different from these two, our work designs online algorithm under non-linear mixed integer optimization scenario and we take into account the operating cost of servers as well as power conditioning and cooling systems. \cite{liu2012renewable,xu2013temperature} also model cooling systems, but focus on offline optimization of the operating cost.

Energy provisioning for power systems is characterized by unit-commitment
problem (UC) \cite{baldwin1959study,1295033}, including a mixed-integer
programming approach \cite{muckstadt1968application} approach and
a stochastic control approach \cite{shiina2004stochastic}. All these
approaches assume the demand (or its distribution) in the entire horizon
is known \emph{a priori}, thus they are applicable only when future
input information can be predicted with certain level of accuracy. In contrast, in this paper we consider an online setting where the algorithms may utilize only information in the current time slot.

In addition to the difference of our work and existing works in the two problems (\emph{i.e.}, capacity provisioning and energy provisioning), our work is also unique in that we jointly optimize both problems while existing works focus on only one of them.

\section{Conclusions\label{sec:conclusion} }

Our work focuses on the cost minimization of data centers achieved
by jointly optimizing \emph{both} the supply of energy from on-site
power generators and the grid, and the demand for energy from its deployed
servers as well as power conditioning and cooling systems.
We show that such an integrated approach is not only possible
in next-generation data centers but also desirable for achieving significant
cost reductions. Our offline optimal algorithm and our online algorithms
with provably good competitive ratios provide key ideas on how to
coordinate energy procurement and production with the energy consumption.
Our empirical work answers several of the important
questions relevant to data center operators focusing on minimizing
their operating costs. We show that a hybrid (resp., on-grid) data
center can achieve a cost reduction between 20.7\% to 25.8\% (resp., 7.3\%
to 12.3\%) by employing our joint optimization framework. We also show
that on-site power generation can provide an additional cost reduction
of about 13\%, and that most of the additional benefit is obtained
by a partial on-site generation capacity of 60\% of the peak power
requirement of the data center.

This work can be extended in several directions.
First, it is interesting to study how energy storage devices can be used to further
reduce the data center operating cost. Second, another interesting direction is to generalize our analysis to take into account deferable workloads.
Third, extension from homogeneous servers and generators to heterogeneous setting is also of great interest.

\section{Acknowledgments}
The work described in this paper was partially supported
by China National 973 projects (No. 2012CB315904 and
2013CB336700), several grants from the University Grants
Committee of the Hong Kong Special Administrative Region,
China (Area of Excellence Project No. AoE/E-02/08
and General Research Fund Project No. 411010 and 411011), and
two gift grants from Microsoft and Cisco.

\bibliographystyle{plain}
{
\bibliography{ref}
}

\appendix
\section{Proof of Theorem 3}\label{sub:Proof-of-Theorem 3}
First, we show that the combined solution $\sum_{i=1}^{M}\bar{\boldsymbol{x}}_{i}$
is optimal to $\textbf{CP}$.

Denote ${\rm C_{CP}}(\boldsymbol{x})$ to be cost of \textbf{CP }of
solution $\boldsymbol{x}$. Suppose that $\tilde{\boldsymbol{x}}$
is an optimal solution for \textbf{CP}. We will show that we can construct
a new feasible solution $\sum_{i=1}^{M}\boldsymbol{\hat{x}}_{i}$
for $\textbf{CP}$, and a new feasible solution $\boldsymbol{\hat{x}}_{i}$
for each $\textbf{CP}_{{\rm i}}$, such that
\begin{equation}
{\rm C_{CP}}(\tilde{\boldsymbol{x}})={\rm C_{CP}}(\sum_{i=1}^{M}\boldsymbol{\hat{x}}_{i})=\sum_{i=1}^{M}{\rm C_{CP_{i}}}(\boldsymbol{\hat{x}}_{i})+\sum_{t=1}^{T}p(t)d_{t}(0).\label{eqn:thm4-eq1}
\end{equation}

$\bar{\boldsymbol{x}}_{i}$ is an optimal solution for each $\textbf{CP}_{{\rm i}}$.
Hence, ${\rm C_{CP_{i}}}(\boldsymbol{\hat{x}}_{i})\ge{\rm C_{CP_{i}}}(\bar{\boldsymbol{x}}_{i})$
for each $i$. Thus,
\begin{eqnarray}
{\rm C_{CP}}(\tilde{\boldsymbol{x}}) & = & \sum_{i=1}^{M}{\rm C_{CP_{i}}}(\boldsymbol{\hat{x}}_{i})+\sum_{t=1}^{T}p(t)d_{t}(0)\nonumber \\
 & \ge & \sum_{i=1}^{M}{\rm {\rm C_{CP_{i}}}}(\bar{\boldsymbol{x}}_{i})+\sum_{t=1}^{T}p(t)d_{t}(0).\label{eq:thm4-eq1-1}
\end{eqnarray}

Besides, we also can prove that
\begin{eqnarray}
\sum_{i=1}^{M}{\rm {\rm C_{CP_{i}}}}(\bar{\boldsymbol{x}}_{i})+\sum_{t=1}^{T}p(t)d_{t}(0) & \ge & {\rm C_{CP}}(\sum_{i=1}^{M}\bar{\boldsymbol{x}}_{i}).\label{eq:thm4-eq2}
\end{eqnarray}

Hence, ${\rm C_{CP}}(\tilde{\boldsymbol{x}})={\rm C_{CP}}(\sum_{i=1}^{M}\bar{\boldsymbol{x}}_{i})$,
\emph{i.e.}, $\sum_{i=1}^{M}\bar{\boldsymbol{x}}_{i}$ is an optimal
solution for \textbf{CP}.

Then, we show ${\rm C_{CP}}(\sum_{i=1}^{M}\boldsymbol{x}_{i}^{on})\leq\gamma\cdot{\rm C_{CP}}(\sum_{i=1}^{M}\bar{\boldsymbol{x}}_{i})$.

Because ${\rm C_{CP_{i}}}(\boldsymbol{x}_{i}^{on})\leq\gamma\cdot{\rm C_{CP_{i}}}(\bar{\boldsymbol{x}}_{i})$
and $\bar{\boldsymbol{x}}_{i}$ is for $\textbf{CP}_{{\rm i}}$, we
have $\gamma\geq1$. According to Eqn. (\ref{eq:thm4-eq1-1}),
we obtain
\begin{eqnarray*}
\gamma\cdot{\rm C_{CP}}(\tilde{\boldsymbol{x}}) & \geq & \sum_{i=1}^{M}{\rm C_{CP_{i}}}(\boldsymbol{x}_{i}^{on})+\sum_{t=1}^{T}p(t)d_{t}(0).
\end{eqnarray*}

Besides, we also can prove that
\begin{eqnarray}
\sum_{i=1}^{M}{\rm {\rm C_{CP_{i}}}}(\boldsymbol{x}_{i}^{on})+\sum_{t=1}^{T}p(t)d_{t}(0) & \ge & {\rm C_{CP}}(\sum_{i=1}^{M}\boldsymbol{x}_{i}^{on}).\label{eq:thm4-eq2-1}
\end{eqnarray}

Hence, ${\rm C_{CP}}(\sum_{i=1}^{M}\boldsymbol{x}_{i}^{on})\leq\gamma\cdot{\rm C_{CP}}(\sum_{i=1}^{M}\bar{\boldsymbol{x}}_{i})$.

It remains to prove Eqns. (\ref{eqn:thm4-eq1}), (\ref{eq:thm4-eq2})
and (\ref{eq:thm4-eq2-1}), which we show in Lemmas~\ref{lem:thm4-lem1}
and \ref{lem:thm4-lem2}.
\begin{lem}
\textup{\label{lem:thm4-lem1}${\rm C_{CP}}(\tilde{\boldsymbol{x}})={\rm C_{CP}}(\sum_{i=1}^{M}\boldsymbol{\hat{x}}_{i})=\sum_{i=1}^{M}{\rm C_{CP_{i}}}(\boldsymbol{\hat{x}}_{i})+\sum_{t=1}^{T}p(t)d_{t}(0).$}\end{lem}
\begin{proof}
Define $\boldsymbol{\hat{x}}_{i}$ based on $\tilde{\boldsymbol{x}}$
by:
\[
\widehat{x}_{i}(t)=\begin{cases}
1, & \mbox{if\ }i\le\tilde{x}(t)\\
0, & \mbox{otherwise.}
\end{cases}
\]

It is straightforward to see that
\[
\tilde{x}(t)=\sum_{i=1}^{M}\hat{x}_{i}(t)
\]
and $\boldsymbol{\hat{x}}_{i}$ is a feasible solution for $\textbf{CP}_{{\rm i}}$,
i.e., $\boldsymbol{\hat{x}}_{i}\geq\boldsymbol{a}_{i}$.

So we have ${\rm C_{CP}}(\tilde{\boldsymbol{x}})={\rm C_{CP}}(\sum_{i=1}^{M}\boldsymbol{\hat{x}}_{i}).$

Note that $\widehat{x}_{1}(t)\ge...\ge\widehat{x}_{M}(t)$ is a decreasing
sequence. Because $\widehat{x}_{i}(t)\in\{0,1\},\ \forall i,t$, we
obtain
\begin{eqnarray}
 &  & \sum_{i=1}^{M}[\hat{x}_{i}(t)-\hat{x}_{i}(t-1)]^{+}\nonumber \\
 & = & \begin{cases}
0,\mbox{\qquad\quad if\ }\sum_{i=1}^{M}\hat{x}_{i}(t)\leq\sum_{i=1}^{M}\hat{x}_{i}(t-1)\\
\sum_{i=1}^{M}\hat{x}_{i}(t)-\sum_{i=1}^{M}\hat{x}_{i}(t-1),\mbox{\ \ \ otherwise }
\end{cases}\nonumber \\
 & = & \Big[\sum_{i=1}^{M}\hat{x}_{i}(t)-\sum_{i=1}^{M}\hat{x}_{i}(t-1)\Big]^{+},\label{eq:thm4-eq3}
\end{eqnarray}
and
\begin{eqnarray}
\sum_{i=1}^{M}d_{t}^{i}\cdot\hat{x}_{i}(t)+d_{t}(0) & = & \sum_{i=1}^{\tilde{x}(t)}d_{t}^{i}\cdot1+d_{t}(0)\nonumber \\
 & = & \sum_{i=1}^{\tilde{x}(t)}\left[d_{t}(i)-d_{t}(i-1)\right]+d_{t}(0)\nonumber \\
 & = & d_{t}(\tilde{x}(t))-d_{t}(0)+d_{t}(0)\nonumber \\
 & = & d_{t}(\tilde{x}(t))=d_{t}(\sum_{i=1}^{M}\hat{x}_{i}(t)).\label{eq:thm4-eq4}
\end{eqnarray}

By Eqns. (\ref{eq:thm4-eq3}) and (\ref{eq:thm4-eq4}),
\[
{\rm C_{CP}}(\sum_{i=1}^{M}\boldsymbol{\hat{x}}_{i})=\sum_{i=1}^{M}{\rm C_{CP_{i}}}(\boldsymbol{\hat{x}}_{i})+\sum_{t=1}^{T}p(t)d_{t}(0).
\]

This completes the proof of this lemma.\end{proof}
\begin{lem}
\textup{\label{lem:thm4-lem2}$\sum_{i=1}^{M}{\rm {\rm C_{CP_{i}}}}(\boldsymbol{x}_{i})+\sum_{t=1}^{T}p(t)d_{t}(0)\ge{\rm C_{CP}}(\sum_{i=1}^{M}\boldsymbol{x}_{i})$,
where $\boldsymbol{x}_{i}$ is any feasible solution for problem }$\textbf{CP}_{{\rm i}}$.\end{lem}
\begin{proof}
First, it is straightforward that
\begin{equation}
\sum_{i=1}^{M}[x_{i}(t)-x_{i}(t-1)]^{+}\geq\Big[\sum_{i=1}^{M}x_{i}(t)-\sum_{i=1}^{M}x_{i}(t-1)\Big]^{+}.\label{eq:thm4-eq5}
\end{equation}

Denote $x(t)=\sum_{i=1}^{M}x_{i}(t)$. Then, $\forall t$,
\begin{eqnarray}
\sum_{i=1}^{M}d_{t}^{i}\cdot x_{i}(t)+d_{t}(0) & \geq & \sum_{i=1}^{x(t)}d_{t}^{i}+d_{t}(0)\nonumber \\
 & = & d_{t}(x(t))-d_{t}(0)+d_{t}(0)\nonumber \\
 & = & d_{t}(x(t))=d_{t}(\sum_{i=1}^{M}x_{i}(t)),\label{eq:thm4-eq6-1}
\end{eqnarray}
where the first inequality comes from $x_{i}(t)\in\{0,1\}$ and $d_{t}^{1}\leq d_{t}^{2}\leq\cdots\leq d_{t}^{M}$.
This is because $d_{t}^{i}=d_{t}(i)-d_{t}(i-1)$ and $d_{t}(x)$ is convex in $x$.

This lemma follows from Eqns. (\ref{eq:thm4-eq5}) and (\ref{eq:thm4-eq6-1}).
\end{proof}

\section{Proof of Theorem 5}\label{sub:Proof-of-Theorem 5}

First, we show that the combined solution $\sum_{i=1}^{N}\bar{\boldsymbol{y}}_{i}$
is optimal to $\textbf{EP}$.

Denote ${\rm C_{EP}}(\boldsymbol{y})$ to be cost of \textbf{EP} of
solution $\boldsymbol{y}$. Suppose that $\tilde{\boldsymbol{y}}$
is an optimal solution for \textbf{EP}. We will show that we can construct
a new feasible solution $\sum_{i=1}^{N}\boldsymbol{\hat{y}}_{i}$
for $\textbf{EP}$, and a new feasible solution $\boldsymbol{\hat{y}}_{i}$
for each $\textbf{EP}_{{\rm i}}$, such that
\begin{equation}
{\rm C_{EP}}(\tilde{\boldsymbol{y}})={\rm C_{EP}}(\sum_{i=1}^{N}\boldsymbol{\hat{y}}_{i})=\sum_{i=1}^{N}{\rm C_{EP_{i}}}(\boldsymbol{\hat{y}}_{i})+\sum_{t=1}^{T}p(t)\left[e(t)-NL\right]^{+}.\label{eqn:thm6-eq1}
\end{equation}

$\boldsymbol{\bar{y}}_{i}$ is an optimal solution for each $\textbf{EP}_{{\rm i}}$.
Hence, ${\rm C_{EP_{i}}}(\boldsymbol{\hat{y}}_{i})\ge{\rm C_{EP_{i}}}(\boldsymbol{\bar{y}}_{i})$
for each $i$. Thus,
\begin{eqnarray}
{\rm C_{EP}}(\tilde{\boldsymbol{y}}) & = & \sum_{i=1}^{N}{\rm C_{EP_{i}}}(\boldsymbol{\hat{y}}_{i})+\sum_{t=1}^{T}p(t)\left[e(t)-NL\right]^{+}\nonumber \\
 & \ge & \sum_{i=1}^{N}{\rm {\rm C_{EP_{i}}}}(\boldsymbol{\bar{y}}_{i})+\sum_{t=1}^{T}p(t)\left[e(t)-NL\right]^{+}.\label{eq:thm6-eq1-1}
\end{eqnarray}

Besides, we also can prove that
\begin{equation}
\sum_{i=1}^{N}{\rm {\rm C_{EP_{i}}}}(\boldsymbol{\bar{y}}_{i})+\sum_{t=1}^{T}p(t)\left[e(t)-NL\right]^{+}\ge{\rm C_{EP}}(\sum_{i=1}^{N}\boldsymbol{\bar{y}}_{i}).\label{eq:thm6-eq2}
\end{equation}

Hence, ${\rm C_{EP}}(\tilde{\boldsymbol{y}})={\rm C_{EP}}(\sum_{i=1}^{N}\boldsymbol{\bar{y}}_{i})$,
\emph{i.e.}, $\sum_{i=1}^{N}\boldsymbol{\bar{y}}_{i}$ is an optimal
solution for $\textbf{EP}$.

Then, we show ${\rm C_{EP}}(\sum_{i=1}^{N}\boldsymbol{y}_{i}^{on})\leq\gamma\cdot{\rm C_{EP}}(\sum_{i=1}^{N}\bar{\boldsymbol{y}}_{i})$.

Because ${\rm C_{EP_{i}}}(\boldsymbol{y}_{i}^{on})\leq\gamma\cdot{\rm C_{EP_{i}}}(\bar{\boldsymbol{y}}_{i})$
and $\bar{\boldsymbol{y}}_{i}$ is optimal for $\textbf{EP}_{{\rm i}}$,
we have $\gamma\geq1$. According to Eqn. (\ref{eq:thm6-eq1-1}),
we have
\begin{eqnarray*}
\gamma\cdot{\rm C_{EP}}(\tilde{\boldsymbol{y}}) & \geq & \sum_{i=1}^{N}{\rm {\rm C_{EP_{i}}}}(\boldsymbol{y}_{i}^{on})+\sum_{t=1}^{T}p(t)\left[e(t)-NL\right]^{+}.
\end{eqnarray*}

Besides, we also can prove that
\begin{equation}
\sum_{i=1}^{N}{\rm {\rm C_{EP_{i}}}}(\boldsymbol{y}_{i}^{on})+\sum_{t=1}^{T}p(t)\left[e(t)-NL\right]^{+}\ge{\rm C_{EP}}(\sum_{i=1}^{N}\boldsymbol{y}_{i}^{on}).\label{eq:thm6-eq2-1}
\end{equation}

Hence, ${\rm C_{EP}}(\sum_{i=1}^{N}\boldsymbol{y}_{i}^{on})\leq\gamma\cdot{\rm C_{EP}}(\sum_{i=1}^{N}\bar{\boldsymbol{y}}_{i})$.

It remains to prove Eqn. (\ref{eqn:thm6-eq1}), (\ref{eq:thm6-eq2})
and (\ref{eq:thm6-eq2-1}), which we show in Lemmas~\ref{lem:thm6-lem1}
and \ref{lem:thm6-lem2}.
\begin{lem}
\textup{\label{lem:thm6-lem1}${\rm C_{EP}}(\tilde{\boldsymbol{y}})={\rm C_{EP}}(\sum_{i=1}^{N}\boldsymbol{\hat{y}}_{i})=\sum_{i=1}^{N}{\rm C_{EP_{i}}}(\boldsymbol{\hat{y}}_{i})+\sum_{t=1}^{T}p(t)\left[e(t)-NL\right]^{+}$.}\end{lem}
\begin{proof}
Define $\boldsymbol{\hat{y}}_{i}$ based on $\tilde{\boldsymbol{y}}$
by:
\begin{equation}
\widehat{y}_{i}(t)=\begin{cases}
1, & \mbox{if\ }i\le\tilde{y}(t)\\
0, & \mbox{otherwise.}
\end{cases}
\end{equation}

It is straightforward to see that
\begin{equation}
\tilde{y}(t)=\sum_{i=1}^{N}\hat{y}_{i}(t).
\end{equation}

So we have ${\rm C_{EP}}(\tilde{\boldsymbol{y}})={\rm C_{EP}}(\sum_{i=1}^{N}\boldsymbol{\hat{y}}_{i}).$

According to $\textbf{EP}$,
\begin{eqnarray*}
{\rm C_{EP}}(\sum_{i=1}^{N}\boldsymbol{\hat{y}}_{i}) & = & \sum_{t=1}^{T}\left\{ \psi\left(\sum_{i=1}^{N}\hat{y}_{i}(t),p(t),e(t)\right)\right.\\
 &  & \;\ \left.+\beta_{g}[\sum_{i=1}^{N}\hat{y}_{i}(t)-\sum_{i=1}^{N}\hat{y}_{i}(t-1)]^{+}\right\} ,
\end{eqnarray*}
and
\begin{eqnarray*}
\sum_{i=1}^{N}{\rm C_{EP_{i}}}(\boldsymbol{\hat{y}}_{i}) & = & \sum_{t=1}^{T}\left\{ \sum_{i=1}^{N}\psi\left(\hat{y}_{i}(t),p(t),e_{i}(t)\right)\right.\\
 &  & \;\ \left.+\beta_{g}\sum_{i=1}^{N}[\hat{y}_{i}(t)-\hat{y}_{i}(t-1)]^{+}\right\} .
\end{eqnarray*}

Note that $\widehat{y}_{1}(t)\ge...\ge\widehat{y}_{N}(t)$ is a decreasing
sequence. Because $\widehat{y}_{i}(t)\in\{0,1\},\ \forall i,t$, we
obtain
\begin{eqnarray}
 &  & \sum_{i=1}^{N}[\hat{y}_{i}(t)-\hat{y}_{i}(t-1)]^{+}\nonumber \\
 & = & \begin{cases}
0,\mbox{\qquad\quad if\ }\sum_{i=1}^{N}\hat{y}_{i}(t)\leq\sum_{i=1}^{N}\hat{y}_{i}(t-1)\\
\sum_{i=1}^{N}\hat{y}_{i}(t)-\sum_{i=1}^{N}\hat{y}_{i}(t-1),\mbox{\ \ \ otherwise }
\end{cases}\nonumber \\
 & = & \Big[\sum_{i=1}^{N}\hat{y}_{i}(t)-\sum_{i=1}^{N}\hat{y}_{i}(t-1)\Big]^{+}.\label{eq:thm6-eq3}
\end{eqnarray}

Also, according to Eqn. (\ref{eq:the original cost function}),
$\psi\left(y(t),p(t),e(t)\right)$ can be rewritten as:
\begin{eqnarray}
 &  & \psi\left(y(t),p(t),e(t)\right)\nonumber \\
 & \triangleq & \begin{cases}
c_{m}y(t)+p(t)e(t), & \mbox{if }p(t)\leq c_{o},\\
c_{m}y(t)+p(t)e(t)+ & \mbox{else.}\\
\left[c_{o}-p(t)\right]\min\{e(t),Ly(t)\}
\end{cases}\label{eq:psi-rewritten}
\end{eqnarray}

Next, we distinguish two cases:

{\em Case 1}: $e(t)<NL$. In this case, $\sum_{i=1}^{N}e_{i}(t)=e(t)$ and\\
$\left[e(t)-NL\right]^{+}=0$. According to the definition of $e_{i}(t)$,
denoting $\bar{N}=\left\lfloor e(t)/L\right\rfloor <N$, we have
\[
e_{i}(t)=\begin{cases}
L, & \mbox{if }i\leq\bar{N},\\
e(t)-\bar{N}L, & \mbox{if }i=\bar{N}+1,\\
0, & \mbox{else. }
\end{cases}
\]

Because $\widehat{y}_{1}(t)\ge...\ge\widehat{y}_{N}(t)$ is a decreasing
sequence and $\widehat{y}_{i}(t)\in\{0,1\},\ \forall t$, we have
\begin{eqnarray*}
\sum_{i=1}^{N}\min\{e_{i}(t),L\hat{y}_{i}(t)\} & = & \begin{cases}
L\sum_{i=1}^{N}\hat{y}_{i}(t), & \mbox{if }\sum_{i=1}^{N}\hat{y}_{i}(t)\leq\bar{N},\\
e(t) & \mbox{else. }
\end{cases}\\
 & = & \min\{e(t),L\sum_{i=1}^{N}\hat{y}_{i}(t)\}.
\end{eqnarray*}

Thus, by Eqn. (\ref{eq:psi-rewritten}), we have
\begin{eqnarray}
\psi\left(\sum_{i=1}^{N}\hat{y}_{i}(t),p(t),e(t)\right) & = & \sum_{i=1}^{N}\psi\left(\hat{y}_{i}(t),p(t),e_{i}(t)\right)\nonumber \\
 &  & +p(t)\left[e(t)-NL\right]^{+}.\label{eq:thm6-eq4}
\end{eqnarray}

{\em Case 2}: $e(t)\geq NL$. In this case, $e_{i}(t)=L,\:\forall i\in[1,N]$, we
have
\[
\sum_{i=1}^{N}\min\{e_{i}(t),L\hat{y}_{i}(t)\}=L\sum_{i=1}^{N}\hat{y}_{i}(t)=\min\{e(t),L\sum_{i=1}^{N}\hat{y}_{i}(t)\}.
\]

Thus, by Eqn. (\ref{eq:psi-rewritten}), we have
\begin{eqnarray}
\psi\left(\sum_{i=1}^{N}\hat{y}_{i}(t),p(t),e(t)\right) & = & \sum_{i=1}^{N}\psi\left(\hat{y}_{i}(t),p(t),e_{i}(t)\right)\nonumber \\
 &  & +p(t)\left[e(t)-NL\right]^{+}.\label{eq:thm6-eq4-1}
\end{eqnarray}

By Eqns. (\ref{eq:thm6-eq3}), (\ref{eq:thm6-eq4}) and (\ref{eq:thm6-eq4-1}),
we have ${\rm C_{EP}}(\sum_{i=1}^{N}\boldsymbol{\hat{y}}_{i})=\sum_{i=1}^{N}{\rm C_{EP_{i}}}(\boldsymbol{\hat{y}}_{i})+\sum_{t=1}^{T}p(t)\left[e(t)-NL\right]^{+}.$

This completes the proof of this lemma.\end{proof}
\begin{lem}
\textup{\label{lem:thm6-lem2}$\sum_{i=1}^{N}{\rm {\rm C_{EP_{i}}}}(\boldsymbol{y}_{i})+\sum_{t=1}^{T}p(t)\left[e(t)-NL\right]^{+}\ge\\{\rm C_{EP}}(\sum_{i=1}^{N}\boldsymbol{y}_{i}),$
where $\boldsymbol{y}_{i}$ is any feasible solution for problem }$\textbf{EP}_{{\rm i}}$\end{lem}
\begin{proof}
First, it is straightforward that
\begin{equation}
\sum_{i=1}^{N}[y_{i}(t)-y_{i}(t-1)]^{+}\geq\Big[\sum_{i=1}^{N}y_{i}(t)-\sum_{i=1}^{N}y_{i}(t-1)\Big]^{+}.\label{eq:thm6-eq5}
\end{equation}

Then by Eqn. (\ref{eq:psi-rewritten}) and the fact that $\sum_{i=1}^{N}e_{i}(t)=\min\{e(t),NL\}$
and
\begin{eqnarray*}
\sum_{i=1}^{N}\min\{e_{i}(t),Ly_{i}(t)\} & \leq & \min\{\sum_{i=1}^{N}e_{i}(t),L\sum_{i=1}^{N}y_{i}(t)\}\\
 & \leq & \min\{e(t),L\sum_{i=1}^{N}y_{i}(t)\},
\end{eqnarray*}

we have
\begin{eqnarray}
\psi\left(\sum_{i=1}^{N}\bar{y}_{i}(t),p(t),e(t)\right) & \leq & \sum_{i=1}^{N}\psi\left(\bar{y}_{i}(t),p(t),e_{i}(t)\right)\nonumber \\
 &  & +p(t)\left[e(t)-NL\right]^{+}.\label{eq:thm6-eq6}
\end{eqnarray}

This lemma follows from Eqns. (\ref{eq:thm6-eq5}) and (\ref{eq:thm6-eq6}).\end{proof}

\section{Proof of Theorem 4}\label{sub:Proof-of-Theorem 4}
First, we will characterize an offline optimal algorithm for $\textbf{CP}_{{\rm i}}$.

Then, based on the optimal algorithm, we prove the competitive ratio
of our future-aware online algorithm $\mathbf{GCSR_{s}^{(w)}}$.

Finally, we prove the lower bound of competitive ratio of any deterministic
online algorithm.

In $\textbf{CP}_{{\rm i}}$, the workload input $\boldsymbol{a_{i}}$
takes value in $[0,1]$ and exactly one server is required to serve
each $\boldsymbol{a_{i}}$. When $a_{i}(t)>0$, we must keep $x_{i}(t)=1$
to satisfy the feasibility condition. The problem is what we should
do if the server is already active but there is no workload, i.e.,
$a_{i}(t)=0$.

To illustrate the problem better, we define \emph{idling interval
$I_{1}$} as follows: $I_{1}\triangleq[t_{1},t_{2}]$, such that (i)
$a_{i}(t_{1}-1)>0$; (ii) $a_{i}(t_{2}+1)>0$; (iii) $\forall\tau\in[t_{1},t_{2}]$,
$a_{i}(\tau)=0$. Similarly, define the \emph{working interval $I_{2}$:}
$I_{2}\triangleq[t_{1},t_{2}]$, such that (i) $a_{i}(t_{1}-1)=0$;
(ii) $a_{i}(t_{2}+1)=0$; (iii) $\forall\tau\in[t_{1},t_{2}]$, $a_{i}(\tau)>0$.
Define the \emph{starting interval $Is$:} $I_{s}\triangleq[0,t_{2}]$,
such that (i) $a_{i}(t_{2}+1)>0$; (ii) $\forall\tau\in[0,t_{2}]$,
$a_{i}(\tau)=0$. Define the \emph{ending interval $I_{e}$:} $I_{e}\triangleq[t_{1},T+1]$,
such that (i) $a_{i}(t_{1}-1)>0$; (ii) $\forall\tau\in[t_{1},T+1]$,
$a_{i}(\tau)=0$.

Based on the above definitions, we have the following offline optimal
algorithm $\mathbf{CPOFF_{s}}$\emph{ }for problem $\textbf{CP}_{{\rm i}}$.
\begin{algorithm}[htb!]
{
\caption{\label{alg:CPOFF} An offline optimal Algorithm $\mathbf{CPOFF_{s}}$
for $\textbf{CP}_{{\rm i}}$}
\begin{algorithmic}[1]
\STATE According to $\boldsymbol{a_{i}}$, find $I_{s}$, $I_{e}$ and all the $I_{1}$ and $I_{2}$.
\STATE During $I_{s}$ and $I_{e}$ , set $x_{i}=0$.
\STATE During each $I_{2}$, set $x_{i}=1$.
\STATE During each $I_{1}$,
\IF {$\sum_{t\in I_{1}}p(t)d_{t}^{i}\geq\beta_{s}$}
\STATE {set $x_{i}(\tau)=0,\forall\tau\in I_{1}$.}
\ELSE
\STATE {set $x_{i}(\tau)=1,\forall\tau\in I_{1}$.}
\ENDIF
\end{algorithmic}
}
\end{algorithm}

\begin{lem}
\label{lem:CPOFF}\textup{$\mathbf{CPOFF_{s}}$} is an offline optimal
algorithm to problem $\textbf{CP}_{{\rm i}}$. \end{lem}
\begin{proof}
It is easy to see that it is optimal to set $x_{i}=0$ during $I_{s}$
and $I_{e}$ and set $x_{i}=1$ during each $I_{2}$.

During an $I_{1}$, an offline optimal solution must set either $x_{i}(\tau)=0$
or $x_{i}(\tau)=1,\forall\tau\in I_{1}$; otherwise, it will incur
unnecessary switching cost and can not be optimal. The cost of setting
$x_{i}=1$ during an $I_{1}$ is $\sum_{t\in I_{1}}d_{t}^{i}p(t)$.
The cost of setting $x_{i}=0$ during $I_{1}$ is $\beta_{s}$, because
we must pay a turn-on cost $\beta_{s}$ after this $I_{1}$. Thus
the above algorithm $\mathbf{CPOFF_{s}}$\emph{ }is an offline optimal
algorithm to $\textbf{CP}_{{\rm i}}$.\end{proof}
\begin{lem}
\label{lem:GCSR ratio}\textup{$\mathbf{GCSR_{s}^{(w)}}$} is $\left(2-a_{s}\right)$-competitive
for problem $\textbf{CP}_{{\rm i}}$, where $\alpha_{s}\triangleq\min\left(1,wd_{\min}P_{\min}/\beta_{s}\right)\in[0,1]$
and $d_{\min}\triangleq\min_{t}\{d_{t}(1)-d_{t}(0)\}\geq0$.\end{lem}
\begin{proof}
We compare our online algorithm $\mathbf{GCSR_{s}^{(w)}}$ and the
offline optimal algorithm $\mathbf{CPOFF_{s}}$\emph{ }described above
for problem $\textbf{CP}_{{\rm i}}$ and prove the competitive ratio.
Let $\boldsymbol{x_{i}^{on}}$ and $\bar{\boldsymbol{x}}_{i}$ be
the solutions obtained by $\mathbf{GCSR_{s}^{(w)}}$\emph{ }and $\mathbf{CPOFF_{s}}$
for problem $\textbf{CP}_{{\rm i}}$, respectively.

Since $d_{t}(x(t))$ is increasing and convex in $x(t)$ , we have
\begin{eqnarray}
d_{t}^{i} & = & d_{t}(i)-d_{t}(i-1)\nonumber \\
 & \geq & d_{t}(i-1)-d_{t}(i-2)\nonumber \\
 & \vdots\nonumber \\
 & \geq & d_{t}(1)-d_{t}(0)\nonumber \\
 & \geq & \min_{t}\{d_{t}(1)-d_{t}(0)\}=d_{\min}\ge 0.\label{eq:dmin}
\end{eqnarray}

It is easy to see that during $I_{s}$ and $I_{2}$, $\mathbf{GCSR_{s}^{(w)}}$\emph{
}and $\mathbf{CPOFF_{s}}$\emph{ }have the same actions.\emph{ }Since
the adversary can choose the $T$ to be large enough, we can omit
the cost incurred during $I_{e}$ when doing competitive analysis.
Thus, we only need to consider the cost incurred by the $\mathbf{GCSR_{s}^{(w)}}$\emph{
}and $\mathbf{CPOFF_{s}}$ during each $I_{1}$. Notice that at the
beginning of an $I_{2}$, both algorithm may incur switching cost.
However, there must be an $I_{1}$ before an $I_{2}$. So this switching
cost will be taken into account when we analyze the cost incurred
during $I_{1}$. More formally, for a certain $I_{1}$,denoted as
$[t_{1},t_{2}]$,
\begin{eqnarray}
 &  & Cost_{I_{1}}(\boldsymbol{x_{i}})\nonumber \\
 & = & \sum_{t=t_{1}}^{t_{2}}p(t)d_{t}^{i}\left(x_{i}(t)-\left\lceil a_{i}(t)\right\rceil \right)+\beta_{s}\sum_{t=t_{1}}^{t_{2}+1}\left[x_{i}(t)-x_{i}(t-1)\right]^{+}\nonumber \\
 & = & \sum_{t=t_{1}}^{t_{2}}p(t)d_{t}^{i}x_{i}(t)+\beta_{s}\sum_{t=t_{1}}^{t_{2}+1}\left[x_{i}(t)-x_{i}(t-1)\right]^{+}.\label{eq:thm5-eq1}
\end{eqnarray}

$\mathbf{GCSR_{s}^{(w)}}$\emph{ }performs as follows: it accumulates
an {}``idling cost'' and when it reaches $\beta_{s}$, it turns
off the server; otherwise, it keeps the server idle. Specifically,
at time $t$, if there exists $\tau\in[t,t+w]$ such that the idling
cost till $\tau$ is at least $\beta_{s}$, it turns off the server;
otherwise, it keeps it idle. We distinguish two cases:

{\em Case 1}: $w\geq\beta_{s}/(d_{\min}P_{\min})$. In this case,
$\mathbf{GCSR_{s}^{(w)}}$\emph{ }performs the same as $\mathbf{CPOFF_{s}}$
\emph{. }Because

If $\sum_{t\in I_{1}}d_{t}^{i}p(t)\geq\beta_{s}$, $\mathbf{CPOFF_{s}}$\emph{
}turns off the server at the beginning of the $I_{1}$, i.e., at $t_{1}$.
Since $w\geq\beta_{s}/(d_{\min}P_{\min})$ and $d_{t}^{i}\geq d_{\min}$
according to Eqn. (\ref{eq:dmin}), at $t_{1}$ $\mathbf{GCSR_{s}^{(w)}}$\emph{
}can find a $\tau\in[t_{1},t_{1}+w]$ such that the idling cost till
$\tau$ is at least $\beta_{s}$, as a consequence of which it also
turns off the server at the beginning of the $I_{1}$. Both algorithms
turn on the server at the beginning of the following $I_{2}$. Thus,
we obtain
\begin{equation}
Cost_{I_{1}}(\boldsymbol{x_{i}^{on}})=Cost_{I_{1}}(\bar{\boldsymbol{x}}_{i})=\beta_{s}.\label{eq:thm5-eq2}
\end{equation}

If $\sum_{t\in I_{1}}d_{t}^{i}p(t)<\beta_{s}$,$\mathbf{CPOFF_{s}}$
keeps the server idling during the whole $I_{1}$. $\mathbf{GCSR_{s}^{(w)}}$
finds that the accumulate idling cost till the end of the $I_{1}$
will not reach $\beta_{s}$, so it also keeps the server idling during
the whole $I_{1}$. Thus, we have

\[
Cost_{I_{1}}(\boldsymbol{x_{i}^{on}})=Cost_{I_{1}}(\bar{\boldsymbol{x}}_{i})=\sum_{t\in I_{1}}d_{t}^{i}p(t).
\]

{\em Case 2}: $w<\beta_{s}/(d_{\min}P_{\min})$. In this case,
to beat $\mathbf{GCSR_{s}^{(w)}}$, the adversary will choose $p(t)$,
$a_{i}(t)$ and $d_{t}^{i}$ so that $\mathbf{GCSR_{s}^{(w)}}$ will
keep the server idling for some time and then turn it off, but $\mathbf{CPOFF_{s}}$
will turn off the server at the beginning of the $I_{1}$. Suppose
$\mathbf{GCSR_{s}^{(w)}}$ keeps the server idling for $\delta$ slots
given no workload within the look-ahead window and then turn it off.
Then according to Algorithm \ref{alg:GCSR(w)}, we must have $\sum_{\delta+w}d_{t}^{i}p(t)<\beta_{s}$
and $\sum_{\delta+w+1}d_{t}^{i}p(t)\geq\beta_{s}$. In this case,
$Cost_{I_{1}}(\bar{\boldsymbol{x}}_{i})=\beta_{s}$ and
\begin{eqnarray*}
Cost_{I_{1}}(\boldsymbol{x_{i}^{on}}) & = & \sum_{\delta}d_{t}^{i}p(t)+\beta_{s}\\
 & = & \sum_{\delta+w}d_{t}^{i}p(t)-\sum_{w}d_{t}^{i}p(t)+\beta_{s}\\
 & \leq & \beta_{s}-d_{\min}P_{\min}w+\beta_{s}\\
 & = & \beta_{s}(2-\frac{d_{\min}P_{\min}}{\beta_{s}}w).
\end{eqnarray*}

So
\begin{eqnarray*}
\frac{{\rm C_{CP_{i}}}(\boldsymbol{x_{i}^{on}})}{{\rm C_{CP_{i}}}(\bar{\boldsymbol{x}}_{i})} & \leq & \frac{Cost_{I_{1}}(\boldsymbol{x_{i}^{on}})}{Cost_{I_{1}}(\bar{\boldsymbol{x}}_{i})}\\
 & \leq & 2-\frac{d_{\min}P_{\min}}{\beta_{s}}w.
\end{eqnarray*}

Combining the above two cases establishes this lemma.

Furthermore, we have some important observations on $\boldsymbol{x_{i}^{on}}$
and $\bar{\boldsymbol{x}}_{i}$, which will be used in later proofs.
\begin{equation}
\sum_{t=1}^{T}\left[x_{i}^{on}(t)-x_{i}^{on}(t-1)\right]^{+}=\sum_{t=1}^{T}\left[\bar{x}_{i}(t)-\bar{x}_{i}(t-1)\right]^{+}.\label{eq:cp-ob1}
\end{equation}
This is because during an $I_{1}$ with $\sum_{t\in I_{1}}d_{t}^{i}p(t)\geq\beta_{s}$,
$\boldsymbol{x_{i}^{on}}$ keeps the server idling for some time and
then turn it off. $\bar{\boldsymbol{x}}_{i}$ turns off the server
at the beginning of the $I_{1}$. Both $\boldsymbol{x_{i}^{on}}$
and $\bar{\boldsymbol{x}}_{i}$ turn on the server at the beginning
of the following $I_{2}$. During an $I_{1}$ with $\sum_{t\in I_{1}}d_{t}^{i}p(t)<\beta_{s}$,
both $\boldsymbol{x_{i}^{on}}$ and $\bar{\boldsymbol{x}}_{i}$ keep
the server idling till the following $I_{2}$. Thus, $\boldsymbol{x_{i}^{on}}$
and $\bar{\boldsymbol{x}}_{i}$ incur the same server switching cost.
Besides, in both above cases, $x_{i}^{on}(t)$ is no less than $\bar{x}_{i}(t)$,
we have
\begin{equation}
\boldsymbol{x_{i}^{on}}\geq\bar{\boldsymbol{x}}_{i}.\label{eq:cp-ob2}
\end{equation}

We also observe that
\begin{eqnarray}
 &  & \sum_{t=1}^{T}d_{t}^{i}p(t)\left(x_{i}^{on}(t)-\left\lceil a_{i}(t)\right\rceil \right)\nonumber \\
 & \leq & \sum_{t=1}^{T}d_{t}^{i}p(t)\left(\bar{x}_{i}(t)-\left\lceil a_{i}(t)\right\rceil \right)+\nonumber \\
 &  & (1-\alpha_{s})\sum_{t=1}^{T}\left[\bar{x}_{i}(t)-\bar{x}_{i}(t-1)\right]^{+}.\label{eq:cp-ob3}
\end{eqnarray}
By rearranging the terms, we obtain
\begin{equation}
\sum_{t=1}^{T}d_{t}^{i}p(t)\left(x_{i}^{on}(t)-\bar{x}_{i}(t)\right)\leq(1-\alpha_{s})\sum_{t=1}^{T}\left[\bar{x}_{i}(t)-\bar{x}_{i}(t-1)\right]^{+}.\label{eq:cp-ob4}
\end{equation}
Notice that $\sum_{t=1}^{T}d_{t}^{i}p(t)\left(x_{i}(t)-\left\lceil a_{i}(t)\right\rceil \right)$
can be seen as the total server idling cost incurred by solution $\boldsymbol{x}_{i}$.
Since idling only happens in $I_{1}$, Eqn. (\ref{eq:cp-ob3})
follows from the cases discussed above.\end{proof}
\begin{lem}
\label{lem:cr lower bound}$\left(2-a_{s}\right)$ is the lower bound
of competitive ratio of any deterministic online algorithm for problem
$\textbf{CP}_{{\rm i}}$ and also $\textbf{CP}$, where $\alpha_{s}\triangleq\min\left(1,wd_{\min}P_{\min}/\beta_{s}\right)\in[0,1]$.\end{lem}
\begin{proof}
First, we show this lemma holds for problem $\textbf{CP}_{{\rm i}}$.
We distinguish two cases:

{\em Case 1}: $w\geq\beta_{s}/(d_{\min}P_{\min})$. In this case,$\left(2-a_{s}\right)=1$,
which is clearly the lower bound of competitive ratio of any online
algorithm.

{\em Case 2}: $w<\beta_{s}/(d_{\min}P_{\min})$. Similar as the
proof of Lemma \ref{lem:GCSR ratio}, we only need to analyze behaviors of online and offline algorithms
during an \emph{idle interval} $I_{1}$.

Consider the input: $d_{t}^{i}=d_{\min}$ and $p(t)=P_{\min}$,$\forall t\in[1,T]$.
Under this input, during an $I_{1}$, we only need to consider a set
of deterministic online algorithms with the following behavior: either
keep the server idling for the whole $I_{1}$ or keep it idling for
some slots and then turn if off until the end of the $I_{1}$. The
reason is that any deterministic online algorithm not belonging to
this set will turn off the server at some time and turn on the server
before the end of $I_{1}$, and thus there must be an online algorithm
incurring less cost by turning off the server at the same time but
turning on the server at the end of $I_{1}$.

We characterize an algorithm $\mathbf{ALG}$ belonging to this set by
a parameter $\delta$, denoting the time it keeps the server idling
for given $a_{i}\equiv0$ within the lookahead window. Denote the
solutions of algorithms $\mathbf{ALG}$\emph{ }and $\mathbf{CPOFF_{s}}$
for problem $\textbf{CP}_{{\rm i}}$ to be $\boldsymbol{x_{i}^{alg}}$
and $\bar{\boldsymbol{x}}_{i}$, respectively.

If $\delta$ is infinite, the competitive ratio is apparently infinite
due to the fact that the adversary can construct an $I_{1}$ whose
duration is infinite. Thus we only consider those algorithms with
finite $\delta$. The adversary will construct inputs as follows:

If $\delta+w\geq\beta_{s}/(d_{\min}P_{\min})$, the adversary will
construct an $I_{1}$ whose duration is longer than $\delta+w$. In
this case, $\mathbf{ALG}$ will keep server idling for $\delta$ slots
and then turn if off while $\mathbf{CPOFF_{s}}$ turns off the server
at the beginning of the $I_{1}$ (c.f. Fig. \ref{fig:thm4-example1}).
Then the ratio is
\begin{eqnarray*}
\frac{{\rm C_{CP_{i}}}(\boldsymbol{x_{i}^{alg}})}{{\rm C_{CP_{i}}}(\bar{\boldsymbol{x}}_{i})} & = & \frac{\sum_{\delta}d_{\min}P_{\min}+\beta_{s}+d_{\min}P_{\min}}{\beta_{s}+d_{\min}P_{\min}}\\
 & > & 1+\frac{\left[\beta_{s}/(d_{\min}P_{\min})-w\right]d_{\min}P_{\min}}{\beta_{s}+d_{\min}P_{\min}}\\
 & = & 2-\frac{d_{\min}P_{\min}(w+1)}{\beta_{s}+d_{\min}P_{\min}}.
\end{eqnarray*}

If $\delta+w<\beta_{s}/(d_{\min}P_{\min})$, the adversary will construct
an $I_{1}$ whose duration is exactly $\delta+w$. In this case, $\mathbf{ALG}$
will keep server idling for $\delta$ slots and then turn if off while
$\mathbf{CPOFF_{s}}$ keeps the server idling during the whole $I_{1}$
(c.f. Fig. \ref{fig:thm4-example2}). Then the ratio is
\begin{eqnarray*}
\frac{{\rm C_{CP_{i}}}(\boldsymbol{x_{i}^{alg}})}{{\rm C_{CP_{i}}}(\bar{\boldsymbol{x}}_{i})} & = & \frac{\sum_{\delta}d_{\min}P_{\min}+\beta_{s}+d_{\min}P_{\min}}{d_{\min}P_{\min}(\delta+w)+d_{\min}P_{\min}}\\
 & = & \frac{d_{\min}P_{\min}(\delta+w+1)+\beta_{s}-wd_{\min}P_{\min}}{d_{\min}P_{\min}(\delta+w+1)}\\
 & \geq & 1+\frac{\beta_{s}-wd_{\min}P_{\min}}{\beta_{s}+d_{\min}P_{\min}}\\
 & = & 2-\frac{d_{\min}P_{\min}(w+1)}{\beta_{s}+d_{\min}P_{\min}}.
\end{eqnarray*}

When $d_{\min}\rightarrow0$ or $\beta_{s}\rightarrow\infty$, we
have
\[
2-\frac{d_{\min}P_{\min}(w+1)}{\beta_{s}+d_{\min}P_{\min}}\rightarrow2-\frac{d_{\min}P_{\min}w}{\beta_{s}}.
\]

Combining the above two cases establishes the lower bound for problem
$\textbf{CP}_{{\rm i}}$.
\begin{figure}[h]

\subfloat[\label{fig:thm4-example1}$\delta+w\geq\beta_{s}/(d_{\min}P_{\min})$]{\includegraphics[width=0.48\columnwidth]{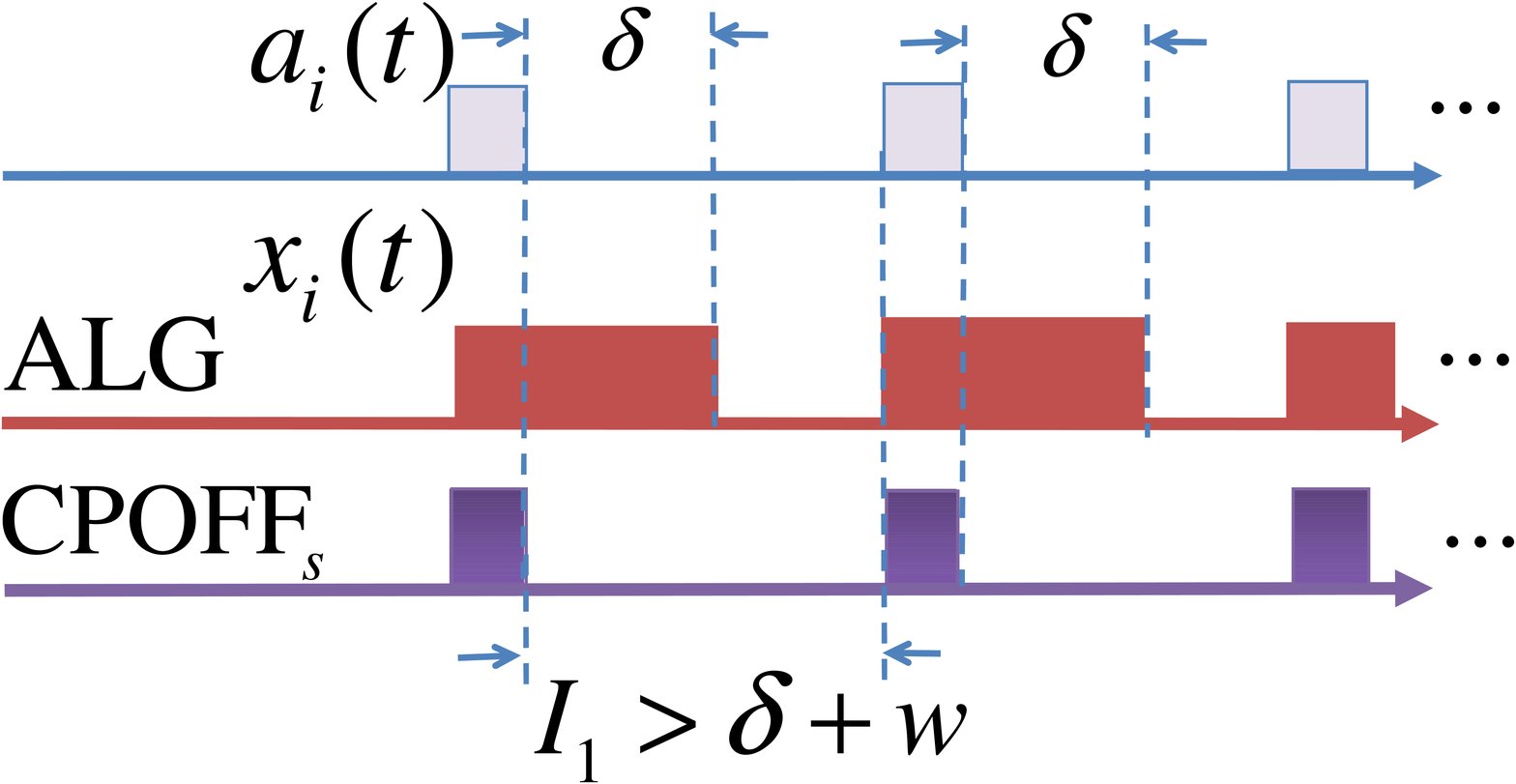}

}\subfloat[\label{fig:thm4-example2}$\delta+w<\beta_{s}/(d_{\min}P_{\min})$]{\includegraphics[width=0.48\columnwidth]{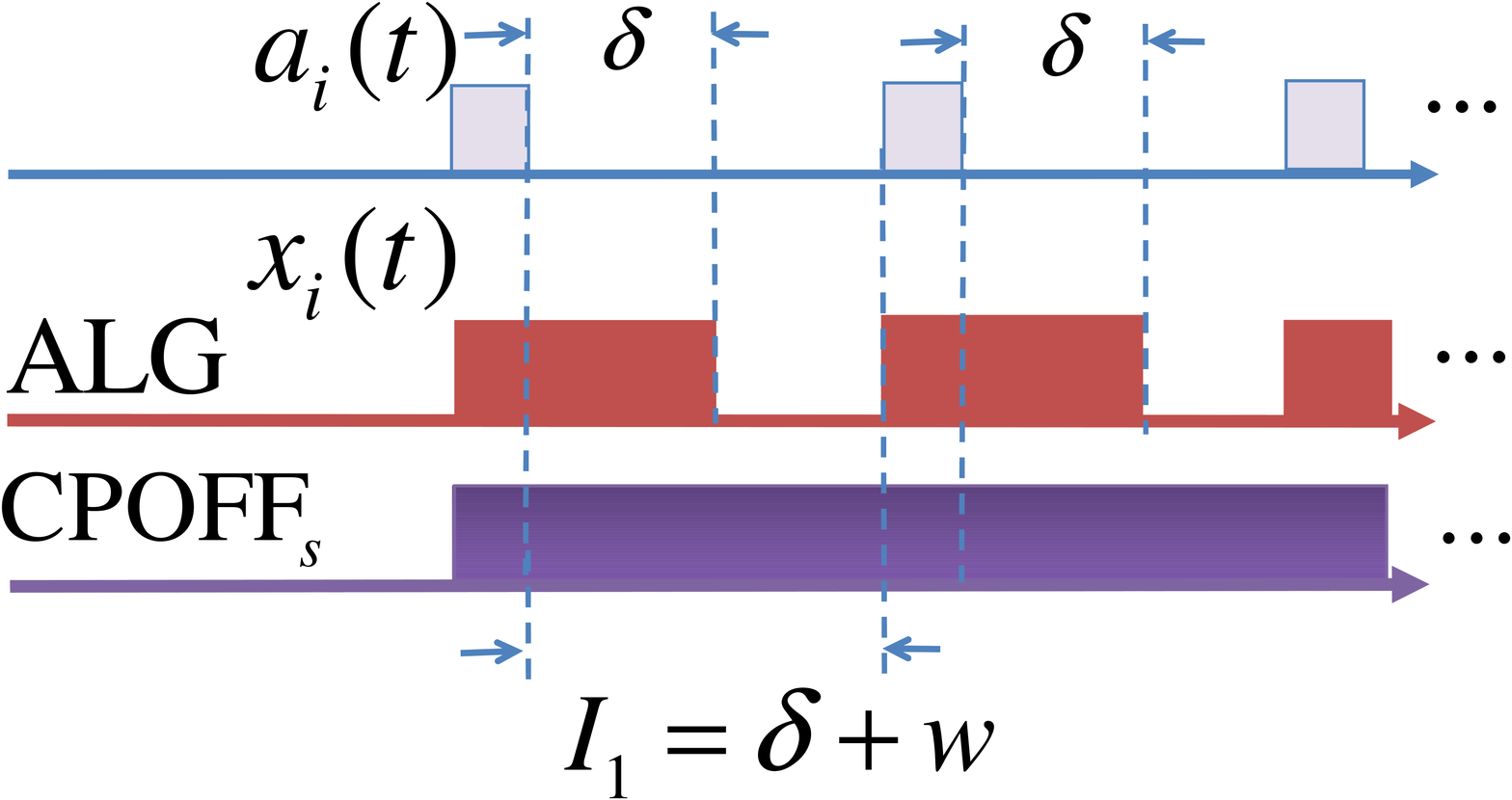}

}\caption{Worst case examples.}

\end{figure}

For problem $\textbf{CP}$, consider the case that $d_{t}(0)=0$ and
$a(t)\in[0,1],\ \forall t$. In this case, it is straightforward that
$\textbf{CP}_{{\rm 1}}$ is equivalent to $\textbf{CP}$. Thus, the
lower bound for $\textbf{CP}_{{\rm i}}$ is also a lower bound for
$\textbf{CP}$.
\end{proof}
Theorem \ref{thm:cr.GCSR(w)} follows from lemmas \ref{lem:GCSR ratio}
and \ref{lem:cr lower bound}.

\section{Proof of Theorem 6}\label{sub:Proof-of-Theorem 6}
Instead of proving this theorem directly, we prove a stronger theorem
that fully characterizes an offline optimal solution. Then Theorem
\ref{thm:EP-offline optimal } follows naturally. An very important
structure of an offline optimal solution is ``critical segments'',
which are constructed according to $R_{i}(t)$.
\begin{defn}
\label{def:critical-seg} We divide all time intervals in $[1,T]$
into disjoint parts called {\em critical segments}:
\[
[1,T_{1}^{c}],[T_{1}^{c}+1,T_{2}^{c}],[T_{2}^{c}+1,T_{3}^{c}],...,[T_{k}^{c}+1,T]
\]
The critical segments are characterized by a set of {\em critical
points}: $T_{1}^{c}<T_{2}^{c}<...<T_{k}^{c}$. We define each critical
point $T_{j}^{c}$ along with an auxiliary point $\tilde{T_{j}^{c}}$,
such that the pair $(T_{j}^{c},\tilde{T_{j}^{c}})$ satisfy the following
conditions:

\smallskip{}

(Boundary): Either \big($R_{i}(T_{j}^{c})=0$ and $R_{i}(\tilde{T_{j}^{c}})=-\beta_{g}$\big)\\
 or \big($R_{i}(T_{j}^{c})=-\beta_{g}$ and $R_{i}(\tilde{T_{j}^{c}})=0$\big).

(Interior): $-\beta<R_{i}(\tau)<0$ for all $T_{j}^{c}<\tau<\tilde{T_{j}^{c}}$.

\smallskip{}

In other words, each pair of $(T_{j}^{c},\tilde{T_{j}^{c}})$ corresponds
to an interval where $R_{i}(t)$ goes from -$\beta_{g}$ to $0$ or
$0$ to -$\beta_{g}$, without reaching the two extreme values inside
the interval. For example, $(T_{1}^{c},\tilde{T_{1}^{c}})$ and $(T_{2}^{c},\tilde{T_{2}^{c}})$
in Fig.~\ref{fig:An-example-of critical segments} are two such pairs,
while the corresponding critical segments are $(T_{1}^{c},T_{2}^{c})$
and $(T_{2}^{c},T_{3}^{c})$. It is straightforward to see that all
$(T_{j}^{c},\tilde{T_{j}^{c}})$ are uniquely defined, and hence critical
segments are well-defined. See Fig.~\ref{fig:An-example-of critical segments}
for an example.
\end{defn}
\begin{figure}[h]
\centering{}\includegraphics[width=0.9\columnwidth]{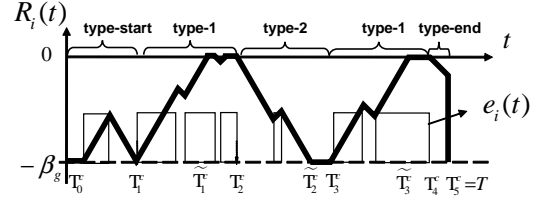}\caption{\label{fig:An-example-of critical segments}An example of critical
segments.}
\end{figure}

Once the time horizon $[1,T]$ is divided into critical segments,
we can now characterize the optimal solution.
\begin{defn}
\label{def:critical-type} We classify the {\em type} of a critical
segment by:

{\em Type-start} (also call {\em type-0}): $[1,T_{1}^{c}]$

{\em Type-1}: $[T_{j}^{c}+1,T_{j+1}^{c}]$, if $R_{i}(T_{j}^{c})=-\beta_{g}$
and $R_{i}(T_{j+1}^{c})=0$

{\em Type-2}: $[T_{j}^{c}+1,T_{j+1}^{c}]$, if $R_{i}(T_{j}^{c})=0$
and $R_{i}(T_{j+1}^{c})=-\beta_{g}$

{\em Type-end} (also call {\em type-3}): $[T_{k}^{c}+1,T]$

For completeness, we also let $T_{0}^{c}=0$ and $T_{k+1}^{c}=T$.

Then the following theorem characterizes an offline optimal solution.\end{defn}
\begin{thm}
\label{thm:OFA-optimal} An optimal solution for $\boldsymbol{EP_{i}}$
is given by
\begin{equation}
y_{{\rm OFA}}(t)\triangleq\begin{cases}
0, & \mbox{if\ }t\in[T_{j}^{c}+1,T_{j+1}^{c}]\mbox{\ is type-start/-2/-end},\\
1, & \mbox{if\ }t\in[T_{j}^{c}+1,T_{j+1}^{c}]\mbox{\ is type-1}.
\end{cases}\label{eq:y_OFA}
\end{equation}
\end{thm}
Theorem \ref{thm:EP-offline optimal } follows from Theorem \ref{thm:OFA-optimal}
and Definition \ref{def:critical-type}. Thus, it remains to prove
Theorem \ref{thm:OFA-optimal}.

\subsection{Proof of Theorem \ref{thm:OFA-optimal}}

Before we prove the theorem, we introduce a lemma.

We define the cost with regard to a segment $j$ by:
\[
\begin{array}{@{}r@{}l@{\ }l}
 &  & {\rm C_{EP_{i}^{sg-j}}}(y)\\
 & \triangleq & {\displaystyle \sum_{t=T_{j}^{c}+1}^{T_{j+1}^{c}}\psi\left(y(t),p(t),e_{i}(t)\right)}+{\displaystyle \sum_{t=T_{j}^{c}+1}^{T_{j+1}^{c}+1}\beta_{g}\cdot[y(t)-y(t-1)]^{+}}
\end{array}
\]

and define a subproblem for critical segment $j$ by:
\begin{align*}
\mathbf{EP_{i}}^{{\rm sg\mbox{-}j}}({\rm y}_{j}^{l},{\rm y}_{j}^{r}):\min\; & {\rm C_{EP_{i}^{sg-j}}}(y)\\
\mbox{s.t.}\; & y(T_{j}^{c})={\rm y}_{j}^{l},\; y(T_{j+1}^{c}+1)={\rm y}_{j}^{r},\\
\mbox{var}\; & y(t)\in\{0,1\},t\in[T_{j}^{c}+1,T_{j+1}^{c}].
\end{align*}

Note that due to the startup cost across segment boundaries, in general
${\rm C_{EP_{i}}}\neq\sum{\rm C_{EP_{i}^{sg-j}}}(y)$. In other words,
we should not expect that putting together the solutions to each segment
will lead to an overall offline optimal solution. However, the following
lemma shows an important structure property that one optimal solution
of $\mathbf{EP_{i}^{sg-j}}(y_{j}^{l},y_{j}^{r})$ is independent of
boundary conditions $(y_{j}^{l},y_{j}^{r})$ although the optimal
value depends on boundary conditions.
\begin{lem}
\label{lem:-is-theindep.with.bound.con} $(y_{{\rm OFA}}(t))_{t=T_{j}^{c}+1}^{T_{j+1}^{c}}$
in (\ref{eq:y_OFA}) is an optimal solution for $\mathbf{EP_{i}^{sg-j}}(y_{j}^{l},y_{j}^{r})$,
despite any boundary conditions $({\rm y}_{j}^{l},{\rm y}_{j}^{r})$.
\end{lem}
We first use this lemma to prove Theorem \ref{thm:OFA-optimal} and
then we prove this lemma. Suppose $(y^{\ast}(t))_{t=1}^{T}$ is an
optimal solution for $\textbf{EP}_{{\rm i}}$. For completeness, we
let $y^{\ast}(0)=0$ and $y^{\ast}(T+1)=0$. We define a sequence
$(y_{0}(t))_{t=1}^{T},(y_{1}(t))_{t=1}^{T}$, $...,(y_{k+1}(t))_{t=1}^{T}$
as follows:
\begin{enumerate}
\item $y_{0}(t)=y^{\ast}(t)$ for all $t\in[1,T]$.
\item For all $t\in[1,T]$ and $j=1,...,k$
\begin{equation}
y_{j}(t)=\begin{cases}
y_{{\rm OFA}}(t), & \mbox{if\ }t\in[1,T_{j}^{c}]\\
y^{\ast}(t), & \mbox{otherwise\ }
\end{cases}
\end{equation}

\item $y_{k+1}(t)=y_{{\rm OFA}}(t)$ for all $t\in[1,T]$.
\end{enumerate}
We next set the boundary conditions for each $\mathbf{EP_{i}^{sg-j}}$
by
\begin{equation}
{\rm y}_{j}^{l}=y_{{\rm OFA}}(T_{j}^{c})\mbox{\ and\ }{\rm y}_{j}^{r}=y^{\ast}(T_{j+1}^{c}+1)
\end{equation}
It follows that
\begin{equation}
{\rm C_{EP_{i}}}(y_{j})-{\rm C_{EP_{i}}}(y_{j+1})={\rm C_{EP_{i}^{sg-j}}}(y^{\ast})-{\rm C_{EP_{i}^{sg-j}}}(y_{{\rm OFA}})
\end{equation}
By Lemma~\ref{lem:-is-theindep.with.bound.con}, we obtain ${\rm C_{EP_{i}^{sg-j}}}(y^{\ast})\ge{\rm C_{EP_{i}^{sg-j}}}(y_{{\rm OFA}})$
for all $j$. Hence,
\begin{equation}
{\rm C_{EP_{i}}}(y^{\ast})={\rm C_{EP_{i}}}(y_{0})\ge...\ge{\rm C_{EP_{i}}}(y_{k+1})={\rm C_{EP_{i}}}(y_{{\rm OFA}})
\end{equation}

This completes the proof of Theorem \ref{thm:OFA-optimal}.

Proof of Lemma~\ref{lem:-is-theindep.with.bound.con}: Consider given
any boundary condition $({\rm y}_{j}^{l},{\rm y}_{j}^{r})$ for $\mathbf{EP_{i}^{sg-j}}$.
Suppose $(\widehat{y}(t))_{t=T_{j}^{c}+1}^{T_{j+1}^{c}}$ is an optimal
solution for $\mathbf{EP_{i}^{sg-j}}$ w.r.t. $({\rm y}_{j}^{l},{\rm y}_{j}^{r})$,
and $\widehat{y}\ne y_{{\rm OFA}}$. We aim to show ${\rm C_{EP_{i}^{sg-j}}}(\widehat{y})\ge{\rm C_{EP_{i}^{sg-j}}}(y_{{\rm OFA}})$,
by considering the types of critical segment.

(\textbf{type-1}): First, suppose that critical segment $[T_{j}^{c}+1,T_{j+1}^{c}]$
is type-1. Hence, $y_{{\rm OFA}}(t)=1$ for all $t\in[T_{j}^{c}+1,T_{j+1}^{c}]$.
Hence,
\begin{equation}
{\rm C_{EP_{i}^{sg-j}}}(y_{{\rm OFA}})=\beta_{g}\cdot(1-{\rm y}_{j}^{l})+\sum_{t=T_{j}^{c}+1}^{T_{j+1}^{c}}\psi\big(1,p(t),e_{i}(t)\big)
\end{equation}

\textbf{Case 1}: Suppose $\widehat{y}(t)=0$ for all $t\in[T_{j}^{c}+1,T_{j+1}^{c}]$.
Hence,
\begin{equation}
{\rm C_{EP_{i}^{sg-j}}}(\widehat{y})=\beta_{g}\cdot{\rm y}_{j}^{r}+\sum_{t=T_{j}^{c}+1}^{T_{j+1}^{c}}\psi\big(0,p(t),e_{i}(t)\big)
\end{equation}
We obtain:
\begin{eqnarray}
 &  & {\rm C_{EP_{i}^{sg-j}}}(\widehat{y})-{\rm C_{EP_{i}^{sg-j}}}(y_{{\rm OFA}})\notag\\
 & = & \beta_{g}\cdot{\rm y}_{j}^{r}+\sum_{t=T_{j}^{c}+1}^{T_{j+1}^{c}}r_{i}(t)-\beta_{g}(1-{\rm y}_{j}^{l})\label{eqn:delta-def1}\\
 & \geq & \beta_{g}\cdot{\rm y}_{j}^{r}+R_{i}(T_{j+1}^{c})-R_{i}(T_{j}^{c})-\beta_{g}(1-{\rm y}_{j}^{l})\label{eqn:Delta-func1}\\
 & = & \beta_{g}\cdot{\rm y}_{j}^{r}+\beta_{g}-\beta_{g}+\beta_{g}{\rm y}_{j}^{l}\geq0
\end{eqnarray}
where Eqn. (\ref{eqn:delta-def1}) follows from the definition of
$r_{i}(t)$ (see Eqn. (\ref{eq:Regret.Definition})) and Eqn. (\ref{eqn:Delta-func1})
follows from Lemma~\ref{lem:Delta-function}. This completes the
proof for Case 1.

(\textbf{Case 2}): Suppose $\widehat{y}(t)=1$ for some $t\in[T_{j}^{c}+1,T_{j+1}^{c}]$.
This implies that ${\rm C_{EP_{i}^{sg-j}}}(\widehat{y})$ has to involve
the startup cost $\beta_{g}$.

Next, we denote the minimal set of segments within $[T_{j}^{c}+1,T_{j+1}^{c}]$
by
\[
[\tau_{1}^{b},\tau_{1}^{e}],[\tau_{2}^{b},\tau_{2}^{e}],[\tau_{3}^{b},\tau_{3}^{e}],...,[\tau_{p}^{b},\tau_{p}^{e}]
\]
such that $\widehat{y}(t)\ne y_{{\rm OFA}}(t)$ for all $t\in[\tau_{l}^{b},\tau_{l}^{e}]$,
$l\in\{1,...,p\}$, where $\tau_{l}^{e}<\tau_{l+1}^{b}$.

Since $\widehat{y}\ne y_{{\rm OFA}}$, then there exists at least
one $t\in[T_{j}^{c}+1,T_{j+1}^{c}]$ such that $\widehat{y}(t)=0$.
Hence, $\tau_{1}^{b}$ is well-defined.

Note that upon exiting each segment $[\tau_{l}^{b},\tau_{l}^{e}]$,
$\widehat{y}$ switches from 0 to 1. Hence, it incurs the startup
cost $\beta_{g}$. However, when $\tau_{p}^{e}=T_{j+1}^{c}$ and ${\rm y}_{j}^{r}=0$,
the startup cost is not for critical segment $[T_{j}^{c}+1,T_{j+1}^{c}]$.

Therefore, we obtain:
\begin{eqnarray}
\hspace{-20pt} &  & {\rm C_{EP_{i}^{sg-j}}}(\widehat{y})-{\rm C_{EP_{i}^{sg-j}}}(y_{{\rm OFA}})\label{eqn:delta-all}\\
 & = & \sum_{t=\tau_{1}^{b}}^{\tau_{1}^{e}}r_{i}(t)+\beta_{g}\cdot\boldsymbol{1}[\tau_{1}^{b}\ne T_{j}^{c}+1]\label{eqn:delta-def21}\\
 &  & +\sum_{l=2}^{p-1}\Big(\sum_{t=\tau_{l}^{b}}^{\tau_{l}^{e}}r_{i}(t)+\beta_{g}\Big)\label{eqn:delta-def22}\\
 &  & +\sum_{t=\tau_{p}^{b}}^{\tau_{p}^{e}}r_{i}(t)+\beta_{g}{\rm y}_{j}^{r}\cdot\boldsymbol{1}[\tau_{p}^{e}=T_{j+1}^{c}]+\beta_{g}\cdot\boldsymbol{1}[\tau_{p}^{e}\ne T_{j+1}^{c}].\label{eqn:delta-def23}
\end{eqnarray}

Now we prove the terms (\ref{eqn:delta-def21}) (\ref{eqn:delta-def22})
and (\ref{eqn:delta-def23}) are all no less than $0.$

\textbf{First}, if $\tau_{1}^{b}=T_{j}^{c}+1,$ then
\begin{eqnarray*}
\sum_{t=\tau_{1}^{b}}^{\tau_{1}^{e}}r_{i}(t)+\beta_{g}\cdot\boldsymbol{1}[\tau_{1}^{b}\ne T_{j}^{c}+1] & = & \sum_{t=T_{j}^{c}+1}^{\tau_{1}^{e}}r_{i}(t)\\
 & \geq & R_{i}(\tau_{1}^{e})-R_{i}(T_{j}^{c})\\
 & \geq & R_{i}(\tau_{1}^{e})+\beta_{g} \geq 0.\\
\end{eqnarray*}
else then
\begin{eqnarray*}
\sum_{t=\tau_{1}^{b}}^{\tau_{1}^{e}}r_{i}(t)+\beta_{g}\cdot\boldsymbol{1}[\tau_{1}^{b}\ne T_{j}^{c}+1] & = & \sum_{t=\tau_{1}^{b}}^{\tau_{1}^{e}}r_{i}(t)+\beta_{g}\\
 & \geq & R_{i}(\tau_{1}^{e})-R_{i}(\tau_{1}^{b}-1)+\beta_{g}\\
 & \geq & R_{i}(\tau_{1}^{e})+\beta_{g} \geq 0.\\
\end{eqnarray*}

Thus, we proved (\ref{eqn:delta-def21})$\geq0.$

\textbf{Second},
\begin{eqnarray*}
\sum_{t=\tau_{l}^{b}}^{\tau_{l}^{e}}r_{i}(t)+\beta_{g} & \ge & R_{i}(\tau_{l}^{e})-R_{i}(\tau_{l}^{b}-1)+\beta_{g}\\
 & \ge & R_{i}(\tau_{l}^{e})+\beta_{g} \geq 0.\\
\end{eqnarray*}

Thus, we proved (\ref{eqn:delta-def22})$\geq0.$

\textbf{Last}, if $\tau_{p}^{e}=T_{j+1}^{c},$ then
\begin{eqnarray*}
 &  & \sum_{t=\tau_{p}^{b}}^{\tau_{p}^{e}}r_{i}(t)+\beta_{g}{\rm y}_{i}^{r}\cdot\boldsymbol{1}[\tau_{p}^{e}=T_{j+1}^{c}]+\beta_{g}\cdot\boldsymbol{1}[\tau_{p}^{e}\ne T_{j+1}^{c}]\\
 & \ge & \sum_{t=\tau_{p}^{b}}^{T_{j+1}^{c}}r_{i}(t)\geq R_{i}(T_{j+1}^{c})-R_{i}(\tau_{p}^{b}-1)\\
 & = & -R_{i}(\tau_{p}^{b}-1)\geq0.
\end{eqnarray*}
else then
\begin{eqnarray*}
 &  & \sum_{t=\tau_{p}^{b}}^{\tau_{p}^{e}}r_{i}(t)+\beta_{g}{\rm y}_{i}^{r}\cdot\boldsymbol{1}[\tau_{p}^{e}=T_{j+1}^{c}]+\beta_{g}\cdot\boldsymbol{1}[\tau_{p}^{e}\ne T_{j+1}^{c}]\\
 & = & \sum_{t=\tau_{p}^{b}}^{\tau_{p}^{e}}r_{i}(t)+\beta_{g}\geq R_{i}(\tau_{p}^{e})-R_{i}(\tau_{p}^{b}-1)+\beta_{g}\\
 & \ge & 0.
\end{eqnarray*}
Thus, we proved (\ref{eqn:delta-def23})$\geq0.$

So we obtain
\[
{\rm C_{EP_{i}^{sg-j}}}(\widehat{y})-{\rm C_{EP_{i}^{sg-j}}}(y_{{\rm OFA}})\geq0.
\]

(\textbf{type-2}): Next, suppose that critical segment $[T_{j}^{c}+1,T_{j+1}^{c}]$
is type-2. Hence, $y_{{\rm OFA}}(t)=0$ for all $t\in[T_{j}^{c}+1,T_{j+1}^{c}]$.
Note that the above argument applies similarly to type-2 setting,
when we consider (Case 1): $\widehat{y}(t)=1$ for all $t\in[T_{j}^{c}+1,T_{j+1}^{c}]$
and (Case 2): $\widehat{y}(t)=0$ for some $t\in[T_{j}^{c}+1,T_{j+1}^{c}]$.

(\textbf{type-start} and \textbf{type-end}): We note that the argument
of type-2 applies similarly to type-start and type-end settings.

Therefore, we complete the proof by showing ${\rm C_{EP_{i}^{sg-j}}}(\widehat{y})\ge{\rm C_{EP_{i}^{sg-j}}}(y_{{\rm OFA}})$
for all $j\in[0,k]$.
\begin{lem}
\label{lem:Delta-function} Suppose $\tau_{1},\tau_{2}\in[T_{j}^{c}+1,T_{j+1}^{c}]$
and $\tau_{1}<\tau_{2}$. Then,
\begin{equation}
R_{i}(\tau_{2})-R_{i}(\tau_{1})\begin{cases}
\le\sum_{t=\tau_{1}+1}^{\tau_{2}}r_{i}(t), & \mbox{if\ }[T_{j}^{c}+1,T_{j+1}^{c}]\mbox{\ is type-1}\\
\ge\sum_{t=\tau_{1}+1}^{\tau_{2}}r_{i}(t), & \mbox{if\ }[T_{j}^{c}+1,T_{j+1}^{c}]\mbox{\ is type-2}
\end{cases}
\end{equation}
\end{lem}
\begin{proof}
We recall that
\begin{equation}
R_{i}(t)\triangleq\min\Big\{0,\max\{-\beta_{g},R_{i}(t-1)+r_{i}(t)\}\Big\}
\end{equation}
First, we consider $[T_{j}^{c}+1,T_{j+1}^{c}]$ as type-1. This implies
that only $R_{i}(T_{j}^{c})=-\beta_{g}$, whereas $R_{i}(t)>-\beta_{g}$
for $t\in[T_{j}^{c}+1,T_{j+1}^{c}]$. Hence,
\begin{equation}
R_{i}(t)=\min\{0,R_{i}(t-1)+r_{i}(t)\}\le R_{i}(t-1)+r_{i}(t)
\end{equation}
Iteratively, we obtain
\begin{equation}
R_{i}(\tau_{2})\le R_{i}(\tau_{1})+\sum_{t=\tau_{1}+1}^{\tau_{2}}r_{i}(t)
\end{equation}

When $[T_{j}^{c}+1,T_{j+1}^{c}]$ is type-2, we proceed with a similar
proof, except
\begin{equation}
R_{i}(t)=\max\{-\beta_{g},R_{i}(t-1)+r_{i}(t)\}\ge R_{i}(t-1)+r_{i}(t)
\end{equation}
Therefore,
\begin{equation}
R_{i}(\tau_{2})\ge R_{i}(\tau_{1})+\sum_{t=\tau_{1}+1}^{\tau_{2}}r_{i}(t).
\end{equation}
\end{proof}

\section{Proof of Theorem 7}\label{sub:Proof-of-Theorem 7}
First, we denote the set of indexes of critical segments for type-$h$
by ${\cal T}_{h}\subseteq\{0,..,k\}$. Note that we also refer to
type-start and type-end by type-0 and type-3 respectively.

Define the sub-cost for type-$h$ by
\begin{eqnarray*}
{\rm C_{EP_{i}}^{{\rm ty\mbox{-}}h}}(y) & \triangleq & \sum_{j\in{\cal T}_{h}}\sum_{t=T_{j}^{c}+1}^{T_{j+1}^{c}}\psi\left(y(t),p(t),e_{i}(t)\right)\\
 &  & +\beta_{g}\cdot[y(t)-y(t-1)]^{+}.
\end{eqnarray*}
Hence, ${\rm {\rm C_{EP_{i}}}}(y)=\sum_{h=0}^{3}{\rm C_{EP_{i}}^{{\rm ty\mbox{-}}h}}(y)$.
We prove by comparing the sub-cost for each type-$h$. We denote the
outcome of $\mathbf{CHASE_{s}^{(w)}}$ by $\big(y_{{\rm CHASE(w)}}(t)\big)_{t=1}^{T}.$

(\textbf{type-0}): Note that both $y_{{\rm OFA}}(t)=y_{{\rm CHASE(w)}}(t)=0$
for all $t\in[1,T_{1}^{c}]$. Hence,
\[
{\rm C_{EP_{i}}^{{\rm ty\mbox{-}}0}}(y_{{\rm OFA}})={\rm C_{EP_{i}}^{{\rm ty\mbox{-}}0}}(y_{{\rm CHASE(w)}}).
\]

(\textbf{type-1}): Based on the definition of critical segment (Definition~\ref{def:critical-seg}),
we recall that there is an auxiliary point $\tilde{T_{j}^{c}}$, such
that either \big($R_{i}(T_{j}^{c})=0$ and $R_{i}(\tilde{T_{j}^{c}})=-\beta_{g}$\big)
or \big($R_{i}(T_{j}^{c})=-\beta_{g}$ and $R_{i}(\tilde{T_{j}^{c}})=0$\big).
We focus on the segment $T_{j}^{c}+1+w<\tilde{T}_{j}^{c}$. We observe
\[
y_{{\rm CHASE(w)}}(t)=\begin{cases}
0, & \mbox{for all\ }t\in[T_{j}^{c}+1,\tilde{T}_{j}^{c}-w),\\
1, & \mbox{for all\ }t\in[\tilde{T}_{j}^{c}-w,T_{j+1}^{c}].
\end{cases}
\]
We consider a particular type-1 critical segment, i.e., $k$-th type-1
critical segment: $[T_{j}^{c}+1,T_{j+1}^{c}]$. Note that by the definition
of type-1, $y_{{\rm OFA}}(T_{j}^{c})=y_{{\rm CHASE(w)}}(T_{j}^{c})=0$.
$y_{{\rm OFA}}(t)$ switches from $0$ to $1$ at time $T_{j}^{c}+1$,
while $y_{{\rm CHASE(w)}}$ switches at time $\tilde{T}_{j}^{c}-w$,
both incurring startup cost $\beta_{g}$. The cost difference between
$y_{{\rm CHASE(w)}}$ and $y_{{\rm OFA}}$ within $[T_{j}^{c}+1,T_{j+1}^{c}]$
is
\begin{eqnarray*}
 &  & \sum_{t=T_{j}^{c}+1}^{\tilde{T}_{j}^{c}-w-1}\Big(\psi\left(0,p(t),e_{i}(t)\right)-\psi\Big(1,\sigma(t),e_{i}(t)\Big)\Big)+\beta_{g}-\beta_{g}\\
 & = & \sum_{t=T_{j}^{c}+1}^{\tilde{T}_{j}^{c}-w-1}r_{i}(t)=R_{i}(\tilde{T}_{j}^{c}-w-1)-R_{i}(T_{j}^{c})=q_{k}^{1}+\beta_{g},
\end{eqnarray*}
where $q_{k}^{1}\triangleq R_{i}(\tilde{T}_{j}^{c}-w-1)$.

Recall the number of type-$h$ critical segments $m_{h}\triangleq|{\cal T}_{h}|$.
\[
{\rm C_{EP_{i}}^{{\rm ty\mbox{-}}1}}(y_{{\rm CHASE(w)}})\le{\rm C_{EP_{i}}^{{\rm ty\mbox{-}}1}}(y_{{\rm OFA}})+m_{1}\cdot\beta_{g}+\sum_{k=1}^{m_{1}}q_{k}^{1}.
\]

(\textbf{type-2}) and (\textbf{type-3}): We derive similarly for $h=2$
or $3$ as
\begin{eqnarray*}
{\rm C_{EP_{i}}^{{\rm ty\mbox{-}}h}}(y_{{\rm CHASE(w)}}) & \le & {\rm C_{EP_{i}}^{{\rm ty\mbox{-}}h}}(y_{{\rm OFA}})-\sum_{k=1}^{m_{h}}q_{k}^{h}\\
 & \leq & {\rm C_{EP_{i}}^{{\rm ty\mbox{-}}h}}(y_{{\rm OFA}})+\beta_{g}m_{h}.
\end{eqnarray*}
The last inequality comes from that $q_{k}^{h}\geq-\beta_{g}$ for
all $h,k$.

Furthermore, we note $m_{1}=m_{2}+m_{3}$. Overall, we obtain
\[
\begin{array}{@{}r@{}l@{\ }l}
 &  & {\displaystyle \frac{{\rm C_{EP_{i}}}(y_{{\rm CHASE(w)}})}{{\rm C_{EP_{i}}}(y_{{\rm OFA}})}}={\displaystyle \frac{\sum_{h=0}^{3}{\rm C_{EP_{i}}^{{\rm ty\mbox{-}}h}}(y_{{\rm CHASE(w)}})}{\sum_{h=0}^{3}{\rm C_{EP_{i}}^{{\rm ty\mbox{-}}h}}(y_{{\rm OFA}})}}\\
 & \leq & {\displaystyle \frac{m_{1}\beta_{g}+\sum_{k=1}^{m_{1}}q_{k}^{1}+(m_{2}+m_{3})\beta_{g}+\sum_{h=0}^{3}{\rm C_{EP_{i}}^{{\rm ty\mbox{-}}h}}(y_{{\rm OFA}})}{\sum_{h=0}^{3}{\rm C_{EP_{i}}^{{\rm ty\mbox{-}}h}}(y_{{\rm OFA}})}}\\
 & = & {\displaystyle 1+\frac{2m_{1}\beta_{g}+\sum_{k=1}^{m_{1}}q_{k}^{1}}{\sum_{h=0}^{3}{\rm C_{EP_{i}}^{{\rm ty\mbox{-}}h}}(y_{{\rm OFA}})}}\\
 & \leq & 1+\begin{cases}
0 & \mbox{if\ }m_{1}=0,\\
{\displaystyle \frac{2m_{1}\beta_{g}+\sum_{k=1}^{m_{1}}q_{k}^{1}}{{\rm C_{EP_{i}}^{{\rm ty\mbox{-}}1}}(y_{{\rm OFA}})}} & \mbox{otherwise}.
\end{cases}
\end{array}
\]

By Lemma \ref{lem:w lower bound type-1} and simplifications, we obtain
\begin{eqnarray}
 &  & {\displaystyle \frac{{\rm C_{EP_{i}}}(y_{{\rm CHASE(w)}})}{{\rm C_{EP_{i}}}(y_{{\rm OFA}})}}\nonumber \\
 & \leq & 1+\frac{2\beta_{g}\big(LP_{\max}-Lc_{o}-c_{m}\big)}{\beta_{g}LP_{\max}+w\cdot c_{m}P_{\max}\big(L-\frac{c_{m}}{P_{\max}-c_{o}}\big)\big)}\nonumber \\
 & \leq & 1+\frac{2\big(P_{\max}-c_{o}\big)}{P_{\max}(1+wc_{m}/\beta_{g})}.\label{eq:thm7-another-bound}
\end{eqnarray}

\begin{lem}
\label{lem:w lower bound type-1}
\begin{eqnarray*}
{\rm C_{EP_{i}}^{{\rm ty\mbox{-}}1}}(y_{{\rm OFA}}) & \geq & m_{1}\beta_{g}+\sum_{k=1}^{m_{1}}\Big(\frac{(q_{k}^{1}+\beta_{g})(Lc_{o}+c_{m})}{LP_{\max}-Lc_{o}-c_{m}}\\
 &  & \quad+w\cdot c_{m}+\frac{c_{o}(-q_{k}^{1}+w\cdot c_{m})}{P_{\max}-c_{o}}\Big)\\
 & \geq & \frac{m_{1}P_{\max}(\beta_{g}+wc_{m})}{P_{\max}-c_{o}}.
\end{eqnarray*}
\end{lem}
\begin{proof}
Consider a particular type-1 segment $[T_{j}^{c}+1,T_{j+1}^{c}]$.
Denote the costs of $\mathrm{y_{OFA}}$ during $[T_{j}^{c}+1,\tilde{T}_{j}^{c}-w-1]$
and $[\tilde{T}_{j}^{c}-w,T_{j+1}^{c}]$ by $\mathrm{Cost^{{\rm up}}}$
and $\mathrm{Cost^{{\rm pt}}}$ respectively.

\textbf{Step 1:} We bound $\mathrm{Cost^{{\rm up}}}$ as follows:
\begin{eqnarray}
 &  & \mathrm{Cost^{{\rm up}}}\nonumber \\
 & = & \beta_{g}+\sum_{t=T_{j}^{c}+1}^{\tilde{T}_{j}^{c}-w-1}\psi\left(1,p(t),e_{i}(t)\right)\nonumber \\
 & = & \beta_{g}+(\tilde{T}_{j}^{c}-w-1-T_{j}^{c})c_{m}+\sum_{t=T_{j}^{c}+1}^{\tilde{T}_{j}^{c}-w-1}\big(\psi\left(1,p(t),e_{i}(t)\right)-c_{m}\big).\label{eq:type-1 cost eq1-1}
\end{eqnarray}

On the other hand, we obtain
\begin{eqnarray}
 &  & \sum_{t=T_{j}^{c}+1}^{\tilde{T}_{j}^{c}-w-1}\big(\psi\left(1,p(t),e_{i}(t)\right)-c_{m}\big)\nonumber \\
 & = & \frac{\sum_{t=T_{j}^{c}+1}^{\tilde{T}_{j}^{c}-w-1}\big(\psi\left(1,p(t),e_{i}(t)\right)-c_{m}\big)}{\sum_{t=T_{j}^{c}+1}^{\tilde{T}_{j}^{c}-w-1}\big(\psi\left(0,p(t),e_{i}(t)\right)-\psi\left(1,p(t),e_{i}(t)\right)+c_{m}\big)}\nonumber \\
 &  & \times\sum_{t=T_{j}^{c}+1}^{\tilde{T}_{j}^{c}-w-1}\big(\psi\left(0,p(t),e_{i}(t)\right)-\psi\left(1,p(t),e_{i}(t)\right)+c_{m}\big)\nonumber \\
 & \geq & \min_{\tau\in[T_{j}^{c}+1,\tilde{T}_{j}^{c}-w-1]}\frac{\psi\left(1,p(\tau),e_{i}(\tau)\right)-c_{m}}{\psi\left(0,p(\tau),e_{i}(\tau)\right)-\psi\left(1,p(\tau),e_{i}(\tau)\right)+c_{m}}\nonumber \\
 &  & \times\sum_{t=T_{j}^{c}+1}^{\tilde{T}_{j}^{c}-w-1}\big(\psi\left(0,p(t),e_{i}(t)\right)-\psi\left(1,p(t),e_{i}(t)\right)+c_{m}\big)\nonumber \\
 & \geq & \frac{c_{o}}{P_{\max}-c_{o}}\label{eq:w type-1 cost eq2}\\
 &  & \times\sum_{t=T_{j}^{c}+1}^{\tilde{T}_{j}^{c}-w-1}\big(\psi\left(0,p(t),e_{i}(t)\right)-\psi\left(1,p(t),e_{i}(t)\right)+c_{m}\big).\nonumber
\end{eqnarray}
The last inequality follows from Lemma~\ref{lem:segment-1 minimum cost}.

Next, we bound the second term by
\begin{eqnarray*}
 &  & \sum_{t=T_{j}^{c}+1}^{\tilde{T}_{j}^{c}-w-1}\big(\psi\left(0,p(t),e_{i}(t)\right)-\psi\left(1,p(t),e_{i}(t)\right)+c_{m}\big)\\
 & \geq & \sum_{t=T_{j}^{c}+1}^{\tilde{T}_{j}^{c}-w-1}\big(r_{i}(t)+c_{m}\big)\\
 & \geq & R_{i}\big(\tilde{T}_{j}^{c}-w-1\big)-R_{i}\big(T_{j}^{c}\big)+(\tilde{T}_{j}^{c}-w-1-T_{j}^{c})c_{m}\\
 & = & q_{k}^{1}+\beta_{g}+(\tilde{T}_{j}^{c}-w-1-T_{j}^{c})c_{m}.
\end{eqnarray*}

Together, we obtain
\begin{eqnarray}
\mathrm{} &  & \mathrm{Cost^{{\rm up}}}\nonumber \\
 & \geq & \beta_{g}+(\tilde{T}_{j}^{c}-w-1-T_{j}^{c})c_{m}+\nonumber \\
 &  & \frac{c_{o}}{P_{\max}-c_{o}}\Big(q_{k}^{1}+\beta_{g}+(\tilde{T}_{j}^{c}-w-1-T_{j}^{c})c_{m}\Big)\nonumber \\
 & = & \beta_{g}+\frac{(q_{k}^{1}+\beta_{g})c_{o}+(\tilde{T}_{j}^{c}-w-1-T_{j}^{c})P_{\max}c_{m}}{P_{\max}-c_{o}}.\label{eq:w type-1 cost eq6}
\end{eqnarray}

Furthermore, we note that $\big(\tilde{T}_{j}^{c}-w-1-T_{j}^{c}\big)$
is lower bounded by the steepest descend when $p(t)=P_{\max}$ and
$e_{i}(t)=L$,
\begin{equation}
\tilde{T}_{j}^{c}-w-1-T_{j}^{c}\geq\frac{q_{k}^{1}+\beta_{g}}{L\big(P_{\max}-c_{o}\big)-c_{m}}\label{eq:w type-1 cost eq7}
\end{equation}

By Eqns. (\ref{eq:w type-1 cost eq6})-(\ref{eq:w type-1 cost eq7}),
we obtain
\begin{eqnarray}
 &  & \mathrm{Cost^{{\rm up}}}\notag\nonumber \\
 & \geq & \beta_{g}+\frac{(q_{k}^{1}+\beta_{g})c_{o}+(\tilde{T}_{j}^{c}-w-1-T_{j}^{c})P_{\max}c_{m}}{P_{\max}-c_{o}}\nonumber \\
 & \geq & \beta_{g}+\frac{(q_{k}^{1}+\beta_{g})(Lc_{o}+c_{m})}{L\big(P_{\max}-c_{o}\big)-c_{m}}.\label{eq:w type-1 cost eq8}
\end{eqnarray}

\textbf{Step 2}: We bound $\mathrm{Cost^{{\rm pt}}}$ as follows.
\begin{eqnarray*}
 &  & \mathrm{Cost^{{\rm pt}}}=\sum_{t=\tilde{T}_{j}^{c}-w}^{T_{j+1}^{c}}\psi\left(1,p(t),e_{i}(t)\right)\\
 & = & (T_{j+1}^{c}-\tilde{T}_{j}^{c}+w+1)c_{m}+\sum_{t=\tilde{T}_{j}^{c}-w}^{T_{j+1}^{c}}\big(\psi\left(1,p(t),e_{i}(t)\right)-c_{m}\big)\\
 & \geq & w\cdot c_{m}+\\
 &  & \frac{c_{o}}{P_{\max}-c_{o}}\sum_{t=\tilde{T}_{j}^{c}-w}^{T_{j+1}^{c}}\big(\psi\left(0,p(t),e_{i}(t)\right)-\psi\left(1,p(t),e_{i}(t)\right)+c_{m}\big).
\end{eqnarray*}
On the other hand, we obtain
\begin{eqnarray*}
 &  & \sum_{t=\tilde{T}_{j}^{c}-w}^{T_{j+1}^{c}}\big(\psi\left(0,p(t),e_{i}(t)\right)-\psi\left(1,p(t),e_{i}(t)\right)+c_{m}\\
 & = & \sum_{t=\tilde{T}_{j}^{c}-w}^{T_{j+1}^{c}}r_{i}(t)+(T_{j+1}^{c}-\tilde{T}_{j}^{c}+w+1)c_{m}\\
 & \geq & R_{i}(T_{j+1}^{c})-R_{i}(\tilde{T}_{j}^{c}-w-1)+w\cdot c_{m}=w\cdot c_{m}-q_{k}^{1}.
\end{eqnarray*}

Therefore,
\begin{eqnarray}
\mathrm{Cost^{{\rm pt}}} & \geq & w\cdot c_{m}+\frac{c_{o}(w\cdot c_{m}-q_{k}^{1})}{P_{\max}-c_{o}}.\label{eq:w type-1 cost roof-1}
\end{eqnarray}

Since there are $m_{1}$ type-1 critical segments, according to Eqns. (\ref{eq:w type-1 cost eq8})-(\ref{eq:w type-1 cost roof-1}),
we obtain
\begin{eqnarray*}
 &  & {\rm Cost}^{{\rm ty\mbox{-}}1}(y_{{\rm OFA}})\\
 & \geq & m_{1}\beta_{g}+\sum_{k=1}^{m_{1}}\Big(\frac{(q_{k}^{1}+\beta_{g})(Lc_{o}+c_{m})}{L\big(P_{\max}-c_{o}\big)-c_{m}}\\
 &  & \qquad+w\cdot c_{m}+\frac{c_{o}(-q_{k}^{1}+w\cdot c_{m})}{P_{\max}-c_{o}}\Big)\\
 & \geq & m_{1}\beta_{g}+\sum_{k=1}^{m_{1}}\Big(\frac{(q_{k}^{1}+\beta_{g})c_{o}}{\big(P_{\max}-c_{o}\big)}\\
 &  & \qquad+w\cdot c_{m}+\frac{c_{o}(-q_{k}^{1}+w\cdot c_{m})}{P_{\max}-c_{o}}\Big)\\
 & = & m_{1}\beta_{g}+\frac{m_{1}(\beta_{g}c_{o}+P_{\max}wc_{m})}{P_{\max}-c_{o}}\\
 & = & \frac{m_{1}P_{\max}(\beta_{g}+wc_{m})}{P_{\max}-c_{o}}.
\end{eqnarray*}
\end{proof}
\begin{lem}
\label{lem:segment-1 minimum cost}
\[
\frac{\psi\left(1,p(\tau),e_{i}(\tau)\right)-c_{m}}{\psi\left(0,p(\tau),e_{i}(\tau)\right)-\psi\left(1,p(\tau),e_{i}(\tau)\right)+c_{m}}\geq\frac{c_{o}}{P_{\max}-c_{o}}.
\]
\end{lem}
\begin{proof}
We expand $\psi\left(y(\tau),p(\tau),e_{i}(\tau)\right)$ for each
case:

{\em Case 1}: $c_{o}\geq p(\tau)$. By Eqn. (\ref{eq:the original cost function}) and $e_{i}(\tau)\leq L,\forall i,\tau$,
\begin{eqnarray*}
\psi\left(1,p(\tau),e_{i}(\tau)\right) & = & p(\tau)e_{i}(\tau)+c_{m},\\
\psi\left(0,p(\tau),e_{i}(\tau)\right) & = & p(\tau)e_{i}(\tau).
\end{eqnarray*}
Therefore,
\[
\frac{\psi\left(1,p(t),e_{i}(t)\right)-c_{m}}{\psi\left(0,p(t),e_{i}(t)\right)-\psi\left(1,p(t),e_{i}(t)\right)+c_{m}}=\infty.
\]

{\em Case 2}: $c_{o}<p(\tau)$. By Eqn. (\ref{eq:the original cost function}) and $e_{i}(\tau)\leq L,\forall i,\tau$,

Thus,
\begin{eqnarray*}
\psi\left(1,p(\tau),e_{i}(\tau)\right) & = & c_{o}e_{i}(\tau)+c_{m},\\
\psi\left(0,p(\tau),e_{i}(\tau)\right) & = & p(\tau)e_{i}(\tau).
\end{eqnarray*}
Therefore,
\begin{eqnarray*}
 &  & \frac{\psi\left(1,p(\tau),e_{i}(\tau)\right)-c_{m}}{\psi\left(0,p(\tau),e_{i}(\tau)\right)-\psi\left(1,p(\tau),e_{i}(\tau)\right)+c_{m}}\\
 & \geq & \frac{c_{o}e_{i}(\tau)}{p(\tau)e_{i}(\tau)-c_{o}e_{i}(\tau)}\\
 & \geq & \frac{c_{o}}{P_{\max}-c_{o}}.
\end{eqnarray*}

Combining both cases, we complete the proof of this lemma.\end{proof}

\section{Proof of Theorem 2}\label{sub:Proof-of-Theorem 2}
First, we prove that the factor loss in optimality is at most $LP_{\max}/\left(Lc_{o}+c_{m}\right)$.

Then, we prove that the factor loss is tight.

Let $(\boldsymbol{\bar{x}},\boldsymbol{\bar{y}})$ be the solution
obtained by solving \textbf{CP} and \textbf{EP} separately in sequence
and $\left(\boldsymbol{x}^{*},\boldsymbol{y}^{*}\right)$ be the solution
obtained by solving the joint-optimization \textbf{DCM}. Denote ${\rm C_{DCM}}(\boldsymbol{x},\boldsymbol{y})$
to be cost of \textbf{DCM }of solution $(\boldsymbol{x},\boldsymbol{y})$
and ${\rm C_{CP}}(\boldsymbol{x})$ to be cost of \textbf{CP }of solution
$\boldsymbol{x}$.

It is straightforward that
\begin{equation}
{\rm C_{DCM}}(\boldsymbol{\bar{x}},\boldsymbol{\bar{y}})\leq{\rm C_{DCM}}(\boldsymbol{\bar{x}},\boldsymbol{0}).\label{eq:thm3-eq1}
\end{equation}

Because ${\rm C_{DCM}}(\boldsymbol{x},0)={\rm C_{CP}}(\boldsymbol{x}),$
we have
\begin{equation}
{\rm C_{DCM}}(\boldsymbol{\bar{x}},\boldsymbol{0})={\rm C_{CP}}(\boldsymbol{\bar{x}})\leq{\rm C_{CP}}(\boldsymbol{x}^{*})={\rm C_{DCM}}(\boldsymbol{x}^{*},\boldsymbol{0}).\label{eq:thm3-eq2}
\end{equation}

By Eqns. (\ref{eq:thm3-eq1}) and (\ref{eq:thm3-eq2}), we obtain
\begin{equation}
\frac{{\rm C_{DCM}}(\boldsymbol{\bar{x}},\boldsymbol{\bar{y}})}{{\rm C_{DCM}}(\boldsymbol{x}^{*},\boldsymbol{y}^{*})}\leq\frac{{\rm C_{DCM}}(\boldsymbol{x}^{*},\boldsymbol{0})}{{\rm C_{DCM}}(\boldsymbol{x}^{*},\boldsymbol{y}^{*})}.\label{eq:thm3-eq3}
\end{equation}

Then, according to the following lemma, we get
\[
\rho=\frac{{\rm C_{DCM}}(\boldsymbol{\bar{x}},\boldsymbol{\bar{y}})}{{\rm C_{DCM}}(\boldsymbol{x}^{*},\boldsymbol{y}^{*})}\leq\frac{LP_{\max}}{Lc_{o}+c_{m}}.
\]

\begin{lem}
\label{lem:factor loss bound}${\rm C_{DCM}}(\boldsymbol{x}^{*},\boldsymbol{0})/{\rm C_{DCM}}(\boldsymbol{x}^{*},\boldsymbol{y}^{*})\leq LP_{\max}/\left(Lc_{o}+c_{m}\right).$\end{lem}
\begin{proof}
By plugging solutions $(\boldsymbol{x}^{*},\boldsymbol{0})$ and $(\boldsymbol{x}^{*},\boldsymbol{y}^{*})$
into \textbf{DCM} separately, we have
\begin{eqnarray}
{\rm C_{DCM}}(\boldsymbol{x}^{*},\boldsymbol{0}) & = & \sum_{t=1}^{T}\left\{ p(t)d_{t}\left(x^{*}(t)\right)\right.\nonumber \\
 &  & \left.+\beta_{s}[x^{*}(t)-x^{*}(t-1)]^{+}\right\} \label{eq:thm3-eq4}
\end{eqnarray}
and
\begin{eqnarray}
{\rm C_{DCM}}(\boldsymbol{x}^{*},\boldsymbol{y}^{*}) & = & \sum_{t=1}^{T}\left\{ \psi\left(y^{*}(t),p(t),d_{t}\left(x^{*}(t)\right)\right)\right.\nonumber \\
 &  & +\beta_{s}[x^{*}(t)-x^{*}(t-1)]^{+}\nonumber \\
 &  & \left.+\beta_{g}[y^{*}(t)-y^{*}(t-1)]^{+}\right\} \nonumber \\
 & \geq & \sum_{t=1}^{T}\left\{ \psi\left(y^{*}(t),p(t),d_{t}\left(x^{*}(t)\right)\right)\right.\nonumber \\
 &  & \left.+\beta_{s}[x^{*}(t)-x^{*}(t-1)]^{+}\right\} .\label{eq:thm3-eq5}
\end{eqnarray}

By Eqns. (\ref{eq:thm3-eq4}), (\ref{eq:thm3-eq5}) and (\ref{eq:the original cost function}),
we obtain
\begin{eqnarray*}
 &  & \frac{{\rm C_{DCM}}(\boldsymbol{x}^{*},\boldsymbol{0})}{{\rm C_{DCM}}(\boldsymbol{x}^{*},\boldsymbol{y}^{*})}\\
 & \leq & \frac{\sum_{t=1}^{T}p(t)d_{t}\left(x^{*}(t)\right)}{\sum_{t=1}^{T}\psi\left(y^{*}(t),p(t),d_{t}\left(x^{*}(t)\right)\right)}\\
 & \leq & \max_{t\in\{1,..,T\}}\frac{p(t)d_{t}\left(x^{*}(t)\right)}{\psi\left(y^{*}(t),p(t),d_{t}\left(x^{*}(t)\right)\right)}\\
 & \leq & \begin{cases}
1, & \mbox{if }p(t)\leq c_{o},\\
\frac{P_{\max}d_{t}\left(x^{*}(t)\right)}{c_{o}d_{t}\left(x^{*}(t)\right)+c_{m}\left\lceil d_{t}\left(x^{*}(t)\right)/L\right\rceil }, & \mbox{otherwise}
\end{cases}\\
 & \leq & \frac{P_{\max}d_{t}\left(x^{*}(t)\right)}{c_{o}d_{t}\left(x^{*}(t)\right)+c_{m}d_{t}\left(x^{*}(t)\right)/L}\\
 & = & \frac{P_{\max}}{c_{o}+c_{m}/L}.
\end{eqnarray*}\end{proof}

Next, we prove that the factor loss is tight.
\begin{lem}
\label{lem:factor loss tight}There exist an input such that \\${\rm C_{DCM}}(\boldsymbol{\bar{x}},\boldsymbol{\bar{y}})/{\rm C_{DCM}}(\boldsymbol{x}^{*},\boldsymbol{y}^{*})=LP_{\max}/\left(Lc_{o}+c_{m}\right).$ \end{lem}
\begin{proof}
Consider the following input:
\[
d_{t}(x(t))=e_{m}x(t),\ p(t)=P_{\max},\ \forall t,
\]
and
\[
a(t)=\begin{cases}
\frac{L}{e_{m}}, & \mbox{if }t=1+k(1+\frac{\beta_{s}}{e_{m}P_{\max}}),\ k\in\mathbb{N}^{0},\\
0, & \mbox{otherwise,}
\end{cases}
\]
where $e_{m}>0$ is a constant such that $L/e_{m}$ is an integer.

Then for the above input, according to algorithm \ref{alg:CPOFF},
it is easy to see that
\[
\bar{x}(t)=\begin{cases}
\frac{L}{e_{m}}, & \mbox{if }t=1+k(1+\frac{\beta_{s}}{e_{m}P_{\max}}),\ k\in\mathbb{N}^{0},\\
0, & \mbox{otherwise.}
\end{cases}
\]

Besides, according to algorithm \ref{alg:CPOFF}, the following $\boldsymbol{x}^{*}$
must be an optimal solution whatever $y^{*}$ is.
\[
x^{*}(t)=\frac{L}{e_{m}},\ \forall t.
\]

Without loss of generality, consider the following parameter setting:
\[
L(P_{\max}-c_{o})-c_{m}<\beta_{g},
\]
\[
L(P_{\max}-c_{o})-c_{m}-\frac{\beta_{s}}{e_{m}P_{\max}}c_{m}<0,
\]
and
\[
L(P_{\max}-c_{o})-c_{m}-\frac{\beta_{s}}{e_{m}P_{\max}}c_{m}+\frac{\beta_{s}L}{e_{m}P_{\max}}(P_{\max}-c_{o})>0.
\]

Since $\bar{\boldsymbol{x}}$ and $\boldsymbol{x}^{*}$ have been
determined by us, we can apply Theorem \ref{thm:OFA-optimal} to obtain
the corresponding $\boldsymbol{\bar{y}}$ and $\boldsymbol{y}^{*}$.
According to Eqn. (\ref{eq:Regret.Definition}) and the above
parameter setting, given $\bar{\boldsymbol{x}}$ and $\boldsymbol{a}$,
the corresponding $R_{i}(t)$ never reaches 0. However, given $\boldsymbol{x}^{*}$
and $\boldsymbol{a}$, the corresponding $R_{i}(t)$ will soon reach
0 and never fall back to $-\beta_{g}$. So we have
\[
\bar{y}(t)=0,\ \forall t
\]
 and
\[
y^{*}(t)=1,\ \forall t.
\]

\begin{figure}
\begin{centering}
\includegraphics[width=0.7\columnwidth]{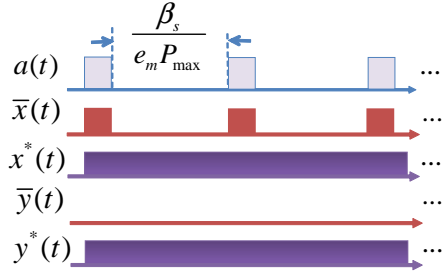}
\par\end{centering}

\caption{\label{fig:thm2-example}Example of $a(t)$, $\bar{x}(t)$, $x^{*}(t)$,
$\bar{y}(t)$ and $y^{*}(t)$. }

\end{figure}

See Fig. \ref{fig:thm2-example} as an example. By plugging the above $(\boldsymbol{\bar{x}},\boldsymbol{\bar{y}})$
and $\left(\boldsymbol{x}^{*},\boldsymbol{y}^{*}\right)$ into \textbf{DCM},
we have
\begin{eqnarray*}
\frac{{\rm C_{DCM}}(\boldsymbol{\bar{x}},\boldsymbol{\bar{y}})}{{\rm C_{DCM}}(\boldsymbol{x}^{*},\boldsymbol{y}^{*})} & = & \frac{LP_{\max}+\beta_{s}L/e_{m}}{Lc_{o}+c_{m}+(Lc_{o}+c_{m})\beta_{s}/(e_{m}P_{\max})}\\
 & = & \frac{LP_{\max}\left[1+\beta_{s}/(e_{m}P_{\max})\right]}{(Lc_{o}+c_{m})\left[1+\beta_{s}/(e_{m}P_{\max})\right]}\\
 & = & \frac{LP_{\max}}{Lc_{o}+c_{m}}.
\end{eqnarray*}
\end{proof}
Theorem \ref{thm:offline decomposition} follows from Eqn. (\ref{eq:thm3-eq3}),
lemmas \ref{lem:factor loss bound} and \ref{lem:factor loss tight}.

\section{Proof of Theorem 8}\label{sub:Proof-of-Theorem 8}
Let $(\boldsymbol{\bar{x}},\boldsymbol{\bar{y}})$ be an offline optimal
solution obtained by solving \textbf{CP} and \textbf{EP} separately
in sequence and $\left(\boldsymbol{x}^{*},\boldsymbol{y}^{*}\right)$
be an offline optimal solution obtained by solving the joint-optimization
\textbf{DCM}. Let $\boldsymbol{x^{on}}$ be the solution obtained
by $\mathbf{GCSR^{(w)}}$ and $\boldsymbol{y^{off}}$ be an offline optimal solution of \textbf{EP} given input $\boldsymbol{x^{on}}$.
Let $\left(\boldsymbol{x^{on}},\boldsymbol{y^{on}}\right)$ be the
solution obtained by $\mathbf{DCMON^{(w)}}$. Denote ${\rm C_{DCM}}(\boldsymbol{x},\boldsymbol{y})$
to be cost of \textbf{DCM }of solution $(\boldsymbol{x},\boldsymbol{y})$
and ${\rm C_{CP}}(\boldsymbol{x})$ to be cost of \textbf{CP }of solution
$\boldsymbol{x}$.

According to Theorem \ref{thm:cr.CHASE(w)}, equation (\ref{eq:thm7-another-bound})
and the fact that the available look-ahead window size is only $\left[w-\Delta_{s}\right]^{+}$
for $\mathbf{DCMON^{(w)}}$ to solve \textbf{EP} (discussed in Sec. \ref{sub:DCMON}), we have
\begin{eqnarray}
 &  & \frac{{\rm C_{DCM}}(\boldsymbol{x^{on}},\boldsymbol{y^{on}})}{{\rm C_{DCM}}(\boldsymbol{x^{on}},\boldsymbol{y^{off}})}\nonumber \\
 & \leq & 1+\frac{2\beta_{g}\left(LP{}_{\max}-Lc_{o}-c_{m}\right)}{\beta_{g}LP{}_{\max}+\left[w-\Delta_{s}\right]^{+}c_{m}P{}_{\max}\left(L-\frac{c_{m}}{P{}_{\max}-c_{o}}\right)}\nonumber \\
 & \leq & 1+\frac{2\left(LP{}_{\max}-Lc_{o}-c_{m}\right)}{LP{}_{\max}+\alpha_{g}P{}_{\max}\left(L-\frac{c_{m}}{P{}_{\max}-c_{o}}\right)}\nonumber \\
 & \leq & 1+2\frac{P_{\max}-c_{o}}{P_{\max}}\cdot\frac{1}{1+\alpha_{g}},\label{eq:thm9-eq1}
\end{eqnarray}
where $\triangle_{s}\triangleq\beta_{s}/(d_{\min}P_{\min})$ and $\alpha_{g}\triangleq\frac{c_{m}}{\beta_{g}}\left[w-\Delta_{s}\right]^{+}$
is a {}``normalized'' look-ahead window size that takes values in
$[0,+\infty)$.

According to Theorem \ref{thm:offline decomposition}, we have
\begin{equation}
\frac{{\rm C_{DCM}}(\boldsymbol{\bar{x}},\boldsymbol{\bar{y}})}{{\rm C_{DCM}}\left(\boldsymbol{x}^{*},\boldsymbol{y}^{*}\right)}\leq\frac{LP_{\max}}{Lc_{o}+c_{m}}.\label{eq:thm9-eq2}
\end{equation}

Then if we can bound ${\rm C_{DCM}}(\boldsymbol{x^{on}},\boldsymbol{y^{off}})/{\rm C_{DCM}}(\boldsymbol{\bar{x}},\boldsymbol{\bar{y}})$,
we obtain the competitive ratio upper bound of $\mathbf{DCMON^{(w)}}$\emph{.
}The following lemma gives us such a bound.
\begin{lem}
\label{lem:thm9-lem}${\rm C_{DCM}}(\boldsymbol{x^{on}},\boldsymbol{y^{off}})/{\rm C_{DCM}}(\boldsymbol{\bar{x}},\boldsymbol{\bar{y}})\leq2-\alpha_{s}$,
where $\alpha_{s}\triangleq\min\left(1,w/\triangle_{s}\right)$ and
\textup{$\triangle_{s}\triangleq\beta_{s}/(d_{\min}P_{\min})$}. \end{lem}
\begin{proof}
It is straightforward that
\begin{equation}
{\rm C_{DCM}}(\boldsymbol{x^{on}},\boldsymbol{y^{off}})\leq{\rm C_{DCM}}(\boldsymbol{x^{on}},\boldsymbol{\bar{y}}).\label{eq:thm9-eq2-1}
\end{equation}

So we seeks to bound ${\rm C_{DCM}}(\boldsymbol{x^{on}},\boldsymbol{\bar{y}})/{\rm C_{DCM}}(\boldsymbol{\bar{x}},\boldsymbol{\bar{y}})$.

For solution $\boldsymbol{x^{on}}$ and $\boldsymbol{\bar{x}}$, denote
\begin{equation}
CW(\boldsymbol{x})=\beta_{s}\sum_{t=1}^{T}\left[x(t)-x(t-1)\right]^{+},\label{eq:thm9-eq3}
\end{equation}
and
\begin{equation}
CI(\boldsymbol{x^{on}},\boldsymbol{\bar{x}})=\sum_{t=1}^{T}p(t)\left(d_{t}(x^{on}(t))-d_{t}(\bar{x}(t))\right).\label{eq:thm9-eq4}
\end{equation}

According to Eqn. (\ref{eq:cp-ob1}), we have
\begin{equation}
CW(\boldsymbol{x_{i}^{on}})=CW(\boldsymbol{\bar{x}}_{i}).\label{eq:thm9-eq5}
\end{equation}

According to lemma \ref{lem:thm9-lem2} and the fact that $x_{i}^{on}(t),\bar{x}_{i}(t)\in\{0,1\},\ \forall t,i$,
we have
\begin{eqnarray}
CW(\boldsymbol{x^{on}}) & = & CW(\sum_{i=1}^{M}\boldsymbol{x_{i}^{on}})=\sum_{i=1}^{M}CW(\boldsymbol{x_{i}^{on}})\nonumber \\
 & = & \sum_{i=1}^{M}CW(\boldsymbol{\bar{x}}_{i})=CW(\sum_{i=1}^{M}\boldsymbol{\bar{x}}_{i})\nonumber \\
 & = & CW(\boldsymbol{\bar{x}}),\label{eq:thm9-eq6}
\end{eqnarray}
and
\begin{eqnarray}
CI(\boldsymbol{x^{on}},\boldsymbol{\bar{x}}) & = & \sum_{t=1}^{T}p(t)\left(d_{t}(x^{on}(t))-d_{t}(\bar{x}(t))\right)\nonumber \\
 & = & \sum_{t=1}^{T}p(t)\left(\sum_{i=1}^{x^{on}(t)}d_{t}^{i}-\sum_{i=1}^{\bar{x}(t)}d_{t}^{i}\right)\nonumber \\
 & = & \sum_{i=1}^{M}\sum_{t=1}^{T}p(t)d_{t}^{i}(x_{i}^{on}(t)-\bar{x}_{i}(t))\nonumber \\
 & \leq & (1-\alpha_{s})\sum_{i=1}^{M}CW(\boldsymbol{\bar{x}_{i}})\nonumber \\
 & = & (1-\alpha_{s})CW(\boldsymbol{\bar{x}}),\label{eq:thm9-eq7}
\end{eqnarray}
where the last and second last inequalities come from Eqns. (\ref{eq:thm9-eq6})
and (\ref{eq:cp-ob4}), respectively.

According to Eqn. (\ref{eq:the original cost function}), we have
$\forall b\in[0,x^{on}(t)]$,
\begin{eqnarray}
 &  & \psi\left(\bar{y}(t),p(t),d_{t}\left(x^{on}(t)\right)\right)-\psi\left(\bar{y}(t),p(t),d_{t}\left(b\right)\right)\nonumber \\
 & \leq & p(t)\left(d_{t}\left(x^{on}(t)\right)-d_{t}\left(b\right)\right).\label{eq:thm9-eq8}
\end{eqnarray}

By the definition of \textbf{DCM}, Eqns. (\ref{eq:cp-ob2}), (\ref{eq:thm9-eq3}),
(\ref{eq:thm9-eq4}) and (\ref{eq:thm9-eq8}),
\begin{eqnarray}
 &  & {\rm C_{DCM}}(\boldsymbol{x^{on}},\boldsymbol{\bar{y}})\nonumber \\
 & = & \sum_{t=1}^{T}\left\{ \psi\left(\bar{y}(t),p(t),d_{t}\left(x^{on}(t)\right)\right)\right.\nonumber \\
 &  & \left.+\beta_{s}[x^{on}(t)-x^{on}(t-1)]^{+}+\beta_{g}[\bar{y}(t)-\bar{y}(t-1)]^{+}\right\} \nonumber \\
 & \leq & \sum_{t=1}^{T}\left\{ \psi\left(\bar{y}(t),p(t),d_{t}\left(\bar{x}(t)\right)\right)+p(t)\left(d_{t}\left(x^{on}(t)\right)-d_{t}\left(\bar{x}(t)\right)\right)\right.\nonumber \\
 &  & \left.+\beta_{s}[x^{on}(t)-x^{on}(t-1)]^{+}+\beta_{g}[\bar{y}(t)-\bar{y}(t-1)]^{+}\right\} \nonumber \\
 & = & \sum_{t=1}^{T}\left\{ \psi\left(\bar{y}(t),p(t),d_{t}\left(\bar{x}(t)\right)\right)+\beta_{g}[\bar{y}(t)-\bar{y}(t-1)]^{+}\right\} \nonumber \\
 &  & +CW(\boldsymbol{x_{on}})+CI(\boldsymbol{x_{on}},\boldsymbol{\bar{x}}).\label{eq:thm9-eq9}
\end{eqnarray}

Then, by Eqns. (\ref{eq:thm9-eq6}), (\ref{eq:thm9-eq7}) and
(\ref{eq:thm9-eq9}), we have
\begin{eqnarray}
 &  & \frac{{\rm C_{DCM}}(\boldsymbol{x_{on}},\boldsymbol{\bar{y}})}{{\rm C_{DCM}}\left(\boldsymbol{\bar{x}},\boldsymbol{\bar{y}}\right)}\nonumber \\
 & \leq & \frac{\sum_{t=1}^{T}\psi\left(\bar{y}(t),p(t),d_{t}\left(\bar{x}(t)\right)\right)+CI(\boldsymbol{x_{on}},\boldsymbol{\bar{x}})+CW(\boldsymbol{x_{on}})}{\sum_{t=1}^{T}\psi\left(\bar{y}(t),p(t),d_{t}\left(\bar{x}(t)\right)\right)+CW(\boldsymbol{\bar{x}})}\nonumber \\
 & \leq & \frac{(1-\alpha_{s})CW(\boldsymbol{\bar{x}})+CW(\boldsymbol{x_{on}})}{CW(\boldsymbol{\bar{x}})}\nonumber \\
 & = & \frac{(1-\alpha_{s})CW(\boldsymbol{\bar{x}})+CW(\boldsymbol{\bar{x}})}{CW(\boldsymbol{\bar{x}})}\nonumber \\
 & = & 2-\alpha_{s}.\label{eq:thm9-eq10}
\end{eqnarray}

This lemma follows from Eqns. (\ref{eq:thm9-eq2-1}) and (\ref{eq:thm9-eq10}).
\end{proof}
Theorem \ref{thm:cr.DCMON(w)} follows from Eqns. (\ref{eq:thm9-eq1}),
(\ref{eq:thm9-eq2}) and lemma \ref{lem:thm9-lem}.
\begin{lem}
\label{lem:thm9-lem2}$\boldsymbol{\bar{x}_{1}},\boldsymbol{\bar{x}_{2}},\ldots\boldsymbol{\bar{x}_{M}}$
and $\boldsymbol{x_{1}^{on}},\boldsymbol{x_{2}^{on}},\ldots\boldsymbol{x_{M}^{on}}$
are decreasing sequences, i.e., $\forall t,\ \bar{x}_{1}(t)\ge...\ge\bar{x}_{M}(t)\ \mbox{and}\ x_{1}^{on}(t)\ge...\ge x_{M}^{on}(t).$\end{lem}
\begin{proof}
Recall that $\boldsymbol{\bar{x}_{i}}$ and $\boldsymbol{x_{i}^{on}}$
are offline and online solutions obtained by $\mathbf{CPOFF_{s}}$
and $\mathbf{GCSR_{s}^{(w)}}$ for problem $\textbf{CP}_{{\rm i}}$,
respectively. According to the definition of $\textbf{CP}_{{\rm i}}$,
$a_{1}(t)\geq a_{2}(t)\ge...\ge a_{M}(t)$ is a decreasing sequence
and $d_{t}^{1}\leq d_{t}^{2}\leq...\leq d_{t}^{M}$ is an increasing
sequence. Thus, for problem $\textbf{CP}_{{\rm i}}$, the larger the
index $i$ is, the more sparse workload tends to be and the higher
power consumption tends to be. Hence, for a larger index $i$, there
are more ``idling intervals'', meanwhile both $\mathbf{CPOFF_{s}}$
and $\mathbf{GCSR_{s}^{(w)}}$ tends to keep servers idling less during
idling intervals (because idling cost is higher). So, $\boldsymbol{\bar{x}_{1}},\boldsymbol{\bar{x}_{2}},\ldots\boldsymbol{\bar{x}_{M}}$
and $\boldsymbol{x_{1}^{on}},\boldsymbol{x_{2}^{on}},\ldots\boldsymbol{x_{M}^{on}}$
are decreasing sequences, i.e., $\forall t,\ \bar{x}_{1}(t)\ge...\ge\bar{x}_{M}(t)\ \mbox{and}\ x_{1}^{on}(t)\\\ge...\ge x_{M}^{on}(t).$\end{proof}

\end{document}